\documentclass[11pt,a4paper]{article}
\usepackage{amsfonts}
\usepackage[dvips]{graphicx}
\usepackage{amsmath}
\usepackage{mathrsfs}
\usepackage{makeidx}
\usepackage{xypic}
\usepackage[nottoc]{tocbibind}
\usepackage{pstricks}
\usepackage{bbm}
\usepackage[arrow, matrix, curve]{xy}
\usepackage{amsthm}
\usepackage{amssymb}
\usepackage{epic}
\usepackage{eepic}
\usepackage{epsfig}
\usepackage{hyperref}
\usepackage{cite}

\xymatrixcolsep{0.5cm}
\xymatrixrowsep{0.5cm}
\newtheorem{theorem}{Theorem}[section]
\newtheorem{proposition}{Proposition}[section]
\newtheorem{lemma}{Lemma}[section]

\newtheorem{corollary}{Corollary}[section]

\newcommand{\beqa}{\begin{eqnarray}}
\newcommand{\eeqa}{\end{eqnarray}}

\newcommand{\bra}[1]{\langle\,#1\,|}
\newcommand{\ket}[1]{|\,#1\,\rangle}

\numberwithin{equation}{section}

\begin{document}

\title{\textbf{Lattice Sine-Gordon model}}
\author{\vspace{1cm}\vspace{2cm}}

\begin{flushright}
LPENSL-TH-02-18
\end{flushright}

\bigskip \bigskip

\begin{center}
\textbf{\LARGE Transfer matrix spectrum for cyclic representations of the
6-vertex reflection algebra II}

\vspace{50pt}
\end{center}

\begin{center}
{\large \textbf{J.~M.~Maillet}\footnote[1]{{Univ Lyon, Ens de Lyon,
Univ Claude Bernard Lyon 1, CNRS, Laboratoire de Physique, UMR 5672, F-69342
Lyon, France; maillet@ens-lyon.fr}},~~ \textbf{G. Niccoli}\footnote[2]{%
{Univ Lyon, Ens de Lyon, Univ Claude Bernard Lyon 1, CNRS,
Laboratoire de Physique, UMR 5672, F-69342 Lyon, France;
giuliano.niccoli@ens-lyon.fr}},~~ \textbf{B. Pezelier}\footnote[3]{{
Univ Lyon, Ens de Lyon, Univ Claude Bernard Lyon 1, CNRS, Laboratoire de
Physique, UMR 5672, F-69342 Lyon, France; baptiste.pezelier@ens-lyon.fr}} }
\end{center}

\begin{center}
\vspace{50pt} \today \vspace{50pt}
\end{center}

\begin{itemize}
\item[ ] \textbf{Abstract.}\thinspace \thinspace This article is a direct continuation of \cite{MNP2017} where we begun the study of the transfer matrix spectral problem for the cyclic representations of the trigonometric 6-vertex reflection algebra associated to the Bazhanov-Stroganov Lax operator. There we addressed this problem for the case where one of the $K$-matrices describing the boundary conditions is triangular.  In the present article we consider the most general integrable boundary conditions, namely the most general boundary $K$-matrices satisfying the reflection equation. The spectral analysis is developed by implementing the method of Separation of Variables (SoV). We first design a suitable gauge transformation that enable us to put into correspondence the spectral problem for the most general boundary conditions with another one having one boundary $K$-matrix in a triangular form. In these settings the SoV resolution can be obtained along an extension of the method described in \cite{MNP2017}. The transfer matrix spectrum is then completely characterized in terms of the set of solutions to a discrete system of polynomial equations in a given class of functions and equivalently as the set of solutions to an analogue of Baxter's T-Q functional equation. We further describe scalar product properties of the separate states including eigenstates of the transfer matrix. 
\end{itemize}

\newpage\tableofcontents\newpage

\section{Introduction}

In the recent years, the out-of-equilibrium behavior of close and open physical systems has attracted a lot of interest motivated in particular by new experimental results, see e.g.  \cite{Greiner2002,Kinoshita2006,Hofferberth2007,Bloch2008,Trotzky2012,Schneider2012,Fukuhara2013,Gogolin2015,Ron2013}. Microscopic models able to describe such situations are thus in general not only described through their bulk  Hamiltonian but also by specifying appropriate boundary conditions. It leads eventually to rather complicated dynamical properties with possible deformations of the bulk symmetries. In the context of strongly coupled systems, integrable models in low dimension, with boundary conditions preserving integrability properties, can be used to gain insights into the non-perturbative behavior of such out-of-equilibrium dynamics, see e.g. \cite{CEM2016} and references therein. In particular, they can also describe classical stochastic relaxation processes, like ASEP \cite{Derrida1998,Schutz2000,Alcaraz1994,Bajnok2006,deGier2005} or transport properties in one dimensional quantum systems, see e.g. \cite{SirPA09,Pro11}.

The algebraic description of quantum integrable models with non-trivial boundary conditions (namely going beyond periodic boundary conditions) goes back to Cherednik \cite{OpenCyChe84} and Sklyanin \cite{OpenCySkly88}. Such models have already a long history, that started with spin chains and Bethe ansatz \cite{OpenCyH28,OpenCyBe31,OpenCyYY661,OpenCyGa83,OpenCyGau71,Bariev1979,Bariev1980,Schulz1985,OpenCyAlca87}, and continued using its modern developments,  see e.g. \cite{OpenCySkly88,OpenCyKulS91,OpenCyMazNR90,OpenCyMazN91,OpenCyPS,ADMFBMNR,OpenCy,OpenCydeVegaG-R93,deVegaR-1994,OpenCyGhosZ94, OpenCyJimKKKM95-1, OpenCyJimKKMW95-2, OpenCyBBOY95,OpenCyKKMNST07,OpenCyDoi03, OpenCyFrapNR07, OpenCyRag2+1, OpenCyRag2+2, OpenCyBK05, OpenCyN02, OpenCyNR03, OpenCyMNS06, OpenCyGal08, OpenCyYNZ06, ADMFCaoYSW13, Xu2016, OpenCyDerKM03-2, OpenCyFSW08, OpenCyFGSW11, OpenCyFram+2, OpenCyFram+1, OpenCyGN12-open, OpenCyFalN14, OpenCyFalKN14, OpenCyKitMN14, OpenCyFanHSY96, OpenCyCao03, OpenCyZhang07, OpenCyACDFR05, OpenCyRag1+, OpenCyRag2+,Baseilhac2013,Belliard2013a,Belliard2013b,Belliard2015a,Belliard2015b,Belliard2015c, OpenCy-H1a, OpenCy-H1b, OpenCy-H1c, OpenCy-H1d, OpenCy-ShiW97,Doikou-2006}. The key point of the algebraic approaches is an extension of the standard Quantum Inverse Scattering method , see e.g.  \cite{OpenCySF78,OpenCyFT79,OpenCyS79,OpenCyKS79,OpenCyFST80,OpenCyF80,OpenCyTh81,OpenCyS82,OpenCyF82,OpenCyKS82,OpenCyIK82,OpenCyBaxBook,OpenCyF95,OpenCyJ90,OpenCyLM66,OpenCySh85}, and its associated Yang-Baxter algebra ; it takes the form of the so-called reflection equations \cite{OpenCyChe84,OpenCySkly88} satisfied by the boundary version of the quantum monodromy matrix. The integrable structure of the model with boundaries can be described in terms of the corresponding bulk quantities supplemented with boundary conditions encoded in some $K$-matrices.  To preserve integrability properties, these $K$-matrices should satisfy reflection equations driven by the $R$-matrix of the model in the bulk which solves the usual Yang-Baxter equation. As shown by Cherednik in \cite{OpenCyChe84} these reflection equations are just consequences of the factorization property of the scattering of particles on a segment having reflecting ends described by the boundary $K$-matrices. It leads to compatibility properties between the scattering in the bulk described by the $R$-matrix and the reflection properties of the ends encoded in the $K$-matrices. These are such that there still exists full series of commuting conserved quantities for the model with boundaries generated by the boundary transfer matrix \cite{OpenCySkly88}. Its expression is quadratic in the bulk monodromy matrix entries and depends on the right and left boundary $K$-matrices. Then, as for the periodic case, the local Hamiltonian for the boundary model can be obtained from this boundary transfer matrix. This is the standard framework to then address the resolution of the common spectral problem for the transfer matrix and its associated local Hamiltonian. There have been quite a number of works devoted to boundary integrable models using, as in the periodic situation, various versions of the  Bethe ansatz  \cite{OpenCyKulS91,OpenCyMazNR90,OpenCyMazN91,OpenCyPS,ADMFBMNR,OpenCy,OpenCydeVegaG-R93,deVegaR-1994,OpenCyGhosZ94, OpenCyJimKKKM95-1, OpenCyJimKKMW95-2, OpenCyBBOY95, OpenCyKKMNST07,OpenCyDoi03, OpenCyFrapNR07, OpenCyRag2+1, OpenCyRag2+2, OpenCyBK05, OpenCyN02, OpenCyNR03, OpenCyMNS06, OpenCyGal08, OpenCyYNZ06, ADMFCaoYSW13, Xu2016, OpenCyDerKM03-2, OpenCyFSW08, OpenCyFGSW11, OpenCyFram+2, OpenCyFram+1, OpenCyGN12-open, OpenCyFalN14, OpenCyFalKN14, OpenCyKitMN14, OpenCyFanHSY96, OpenCyCao03, OpenCyZhang07, OpenCyACDFR05, OpenCyRag1+, OpenCyRag2+,Baseilhac2013,Belliard2013a,Belliard2013b,Belliard2015a,Belliard2015b,Belliard2015c, OpenCy-H1a, OpenCy-H1b, OpenCy-H1c, OpenCy-H1d, OpenCy-ShiW97,Doikou-2006}. It appeared however, that while for special models and boundary $K$-matrices a method very similar to the standard algebraic Bethe ansatz (ABA), here based on the reflection equations, can be applied, the case of the most general boundary conditions (and associated $K$-matrices) preserving integrability turns out of be out of the reach of these methods. This motivated the use of different approaches like in particular the use of $q$-Onsager algebras, see e.g. \cite{OpenCyBK05,Baseilhac2013}, modifications of the Bethe ansatz \cite{OpenCyRag2+1,OpenCyRag2+2,Belliard2013a,Belliard2013b,Belliard2015a,Belliard2015b,Belliard2015c,ADMFCaoYSW13} and the implementation for this case of the separation of variable (SoV) method \cite{OpenCySk1,Skl1985,OpenCySk2,OpenCySk3,OpenCyBBS96,OpenCySm98,DerKM03,OpenCyBT06,OpenCyGIPS06,OpenCyGIPST07,OpenCyNT-10,OpenCyN-10,OpenCyGN12,OpenCyGMN12-SG,OpenCyNWF09,OpenCyN12-0,OpenCyN12-1,OpenCyNicT15-2,OpenCyNicT15,OpenCyLevNT15,OpenCyKitMNT15,KitMNT16}. For an extensive discussion and comparison of these various methods in the case of boundary integrable models, we refer to the general discussion given in the introduction of our first article \cite{MNP2017} and to the references therein. 

In our article \cite{MNP2017} we started the implementation of the SoV method to cyclic representation of the 6-vertex reflection equation associated to the most general Bazhanov-Stroganov quantum Lax operator \cite{RocheA-1989,JimboDMM1990,JimboDMM1991,Tarasov-1992,Tarasov-1993,OpenCyBS90}. Let us recall that the periodic boundary conditions case (spectrum and form factors) was considered in previous works  \cite{OpenCyGIPS06,OpenCyGIPST07,OpenCyNT-10,OpenCyN-10,OpenCyGN12,OpenCyGMN12-SG}, generalizing in particular \cite{OpenCyKitMT99,OpenCyMaiT00}. The interest in such a problem is due to the fact that special cases include the Sine-Gordon lattice model at roots of unity and the Chiral Potts model  \cite{OpenCyauYMcPTY,OpenCyMcPTS,OpenCyauYMcPT,OpenCyBaPauY,OpenCyBBP90,OpenCyBa89,OpenCyAMcP0,OpenCyTarasovSChP,OpenCyBa04,OpenCy-Au-YP08}. In \cite{MNP2017} we started the analysis considering the special case where one of the boundary $K$-matrices has triangular form (which is equivalent to one constraint on the boundary parameters). For that situation we have been able to apply successfully the SoV method by identifying the separate basis as the eigenstate basis of a special diagonalizable $B$-operator with simple spectrum which can be  constructed from the boundary monodromy matrix entries. Then using this separate basis, the spectrum (eigenvalues and eigenstates) for the boundary transfer matrix was completely characterized in terms of the set of solutions to a discrete system of polynomial equations in a given class of functions. 

The purpose of the present article is to address this spectral problem for the most general boundary conditions preserving integrability, namely for the most general $K$-matrices solution of the reflection equation. The method to reach this goal is to design a gauge transformation that enable us to put this general situation into correspondence with the previous one, namely with a model having one triangular $K$-matrix. For that purpose, the standard idea of Baxter's gauge transformations, see e.g. \cite{OpenCyBaxBook,OpenCyFT79} and references therein, has to be adapted in a way similar to \cite{OpenCyFanHSY96,OpenCyFalKN14} and generalized to these cyclic representations of the 6-vertex Yang-Baxter algebra. Then using this correspondence, the method and tools obtained in our first paper \cite{MNP2017} can be used, leading to the complete characterization of the spectrum (again eigenvalues and eigenstates) of the general boundary transfer matrix. We also give determinant formula for the scalar products of the separate states. Further, we show that the spectrum characterization admits a representation in terms of functional equations of Baxter T-Q equation type. Let us note that an analogous inhomogeneous Baxter's like equation has been already proposed in \cite{Xu2016} for this model on the basis of pure functional arguments on the fusion of transfer matrices. Thanks to our SoV construction, we prove in the present article that our inhomogeneous Baxter's like equation does characterize the full transfer matrix spectrum. Let us further remark that the inhomogeneous term in the proposal of \cite{Xu2016} is presented in terms of the averages of the entries of the monodromy matrix. For general cyclic representations these quantities are defined only through (unresolved) recursion formula in \cite{Xu2016}. Hence in \cite{Xu2016} this inhomogeneous term is in fact not given explicitly which makes the comparison with our explicit functional equation not directly possible. Moreover, we would like to stress that in our formulation the Q-function are Laurent polynomial of a smaller degree compared to the one in \cite{Xu2016}. This is due to the fact that the inhomogeneous term being computed explicitly in our SoV derivation it allows us to remove $4p$ ($p$ being an integer characteristic value defining the cyclic representation, see \eqref{def-p} in section 2) irrelevant  zeros from the T-Q functional equation as they can be factored out from each term of this equation (see section 5).

This article is organized as follows. In section 2 we just recall the basics of the cyclic representations associated to the Bazhanov-Stroganov quantum Lax operator. In section 3 we define the gauged transformed reflection algebra that  put into correspondence the most general boundary condition $K$-matrix with a triangular one. It enables us to adapt the SoV method that we already described in our first article \cite{MNP2017} to this more general context, leading in section 4 to the transfer matrix spectrum complete characterization in this SoV basis. There we also present the scalar product formulae for the so-called separate states containing the transfer matrix eigenstates.
In section 5 we show that the spectrum characterization admit a representation in terms of functional equations of Baxter T-Q equation type. Details about the construction of the gauge transformation are given in Appendices A and B together with determinant identities used in the spectrum characterization in Appendix C. 

\section{Cyclic representations of 6-vertex reflection algebra.}

Following Sklyanin's paper \cite{OpenCySkly88}, we consider the most general cyclic
solutions of the 6-vertex reflection equation associated to the Bazhanov-Stroganov Lax operator \cite{OpenCyBS90}:%
\begin{equation}
R_{12}(\lambda /\mu )\,\mathcal{U}_{1,-}(\lambda )\,R_{21}(\lambda \mu /q)\,%
\mathcal{U}_{2,-}(\mu )=\mathcal{U}_{2,-}(\mu )\,R_{21}(\lambda \mu /q)\,%
\mathcal{U}_{1,-}(\lambda )\,R_{12}(\lambda /\mu )  \label{bYB}
\end{equation}%
where the two sides of the equation belong to $\text{End}(V_{1}\otimes
V_{2}\otimes \mathcal{H})$ and are defined by the following boundary monodromy
matrices:%
\begin{equation}
\mathcal{U}_{a,-}(\lambda )=M_{a}(\lambda )K_{a,-}(\lambda )\hat{M}%
_{a}(\lambda )=\left( 
\begin{array}{cc}
\mathcal{A}_{-}(\lambda ) & \mathcal{B}_{-}(\lambda ) \\ 
\mathcal{C}_{-}(\lambda ) & \mathcal{D}_{-}(\lambda )%
\end{array}%
\right) _{a}\in \text{End}(V_{a}\otimes \mathcal{H}),  \label{BMM}
\end{equation}%
where:%
\begin{equation}
\hat{M}_{a}(\lambda )=(-1)^{\mathsf{N}}\,\sigma
_{a}^{y}\,M_{a}^{t_{a}}(1/\lambda )\,\sigma _{a}^{y},  \label{Inverse-M}
\end{equation}%
and $V_{a}\simeq \mathbb{C}^{2}$ is the so-called auxiliary space. Here, 
\begin{equation}
M_{a}(\lambda )=\left( 
\begin{array}{cc}
A(\lambda ) & B(\lambda ) \\ 
C(\lambda ) & D(\lambda )%
\end{array}%
\right) _{a}\equiv L_{a,\mathsf{N}}(\lambda q^{-1/2})\cdots L_{a,1}(\lambda
q^{-1/2})\in \text{End}(V_{a}\otimes \mathcal{H}),
\label{24}
\end{equation}%
is the cyclic solution of the 6-vertex Yang-Baxter equation:%
\begin{equation}
R_{12}(\lambda /\mu )M_{1}(\lambda )M_{2}(\mu )=M_{2}(\mu )M_{1}(\lambda
)R_{12}(\lambda /\mu )\in \text{End}(V_{1}\otimes V_{2}\otimes \mathcal{H}),
\end{equation}%
associated to the R-matrix%
\begin{equation}
R_{ab}(\lambda )=\left( 
\begin{array}{cccc}
q\lambda -q^{-1}\lambda ^{-1} & 0 & 0 & 0 \\[-1mm] 
0 & \lambda -\lambda ^{-1} & q-q^{-1} & 0 \\[-1mm] 
0 & q-q^{-1} & \lambda -\lambda ^{-1} & 0 \\[-1mm] 
0 & 0 & 0 & q\lambda -q^{-1}\lambda ^{-1}%
\end{array}%
\right) \in \text{End}(V_{a}\otimes V_{b}),  \label{Rlsg}
\end{equation}%
and defined in terms of the
Bazhanov-Stroganov's Lax operators \cite{OpenCyBS90}:%
\begin{equation}
L_{a,n}(\lambda )\equiv \left( 
\begin{array}{cc}
\lambda \alpha _{n}v_{n}-\beta _{n}\lambda ^{-1}v_{n}^{-1} & u_{n}\left(
q^{-1/2}a_{n}v_{n}+q^{1/2}b_{n}v_{n}^{-1}\right) \\ 
u_{n}^{-1}\left( q^{1/2}c_{n}v_{n}+q^{-1/2}d_{n}v_{n}^{-1}\right) & \gamma
_{n}v_{n}/\lambda -\delta _{n}\lambda /v_{n}%
\end{array}%
\right) _{a}\in \text{End}(V_{a}\otimes \mathcal{R}_{n}),
\label{laxopppp}
\end{equation}%
where:%
\begin{equation}
\gamma _{n}=a_{n}c_{n}/\alpha _{n},\text{ \ \ \ \ }\delta
_{n}=b_{n}d_{n}/\beta _{n}.
\end{equation}%
The $u_{n}\in $End$(\mathcal{R}_{n})$ and $v_{m}\in $End$(\mathcal{R}_{m})$
are unitary Weyl algebra generators:%
\begin{equation}
u_{n}v_{m}=q^{\delta _{n,m}}v_{m}u_{n}\text{ \ \ with \ }%
u_{n}^{p}=v_{m}^{p}=1\text{\ }\forall n,m\in \{1,...,\mathsf{N}\},
\end{equation}%
and: 
\begin{equation}
q=e^{-i\pi \beta ^{2}}\text{, }\beta ^{2}=p^{\prime }/p\text{ with }%
p^{\prime }\text{ even and }p=2l+1\text{ odd, } l \in \mathbb{N}.
\label{def-p}
\end{equation}%
The local quantum spaces $\mathcal{R}_{n}$ are $p$-dimensional Hilbert
spaces and the full representation space of the cyclic Yang-Baxter and
reflection algebra is defined by the tensor product of the local quantum
spaces, i.e. $\mathcal{H}=\otimes _{n=1}^{\mathsf{N}}\mathcal{R}_{n}$.
Moreover, we consider here the most general boundary matrices defined as:%
\begin{equation}
K_{a,\pm }(\lambda )=\left( 
\begin{array}{cc}
a_{\pm }\left( \lambda \right) & b_{\pm }\left( \lambda \right) \\ 
c_{\pm }\left( \lambda \right) & d_{\pm }\left( \lambda \right)%
\end{array}%
\right) _{a}\equiv K_{a}(\lambda q^{(1\pm 1)/2};\zeta _{\pm },\kappa _{\pm
},\tau _{\pm }),
\label{Kmatuno}
\end{equation}%
where:%
\begin{equation}
K_{a}(\lambda ;\zeta ,\kappa ,\tau )=\frac{1}{\zeta -\frac{1}{\zeta }}\left( 
\begin{array}{cc}
\frac{\lambda \zeta }{q^{1/2}}-\frac{q^{1/2}}{\lambda \zeta } & \kappa
e^{\tau }\left( \frac{\lambda ^{2}}{q}-\frac{q}{\lambda ^{2}}\right) \\ 
\kappa e^{-\tau }\left( \frac{\lambda ^{2}}{q}-\frac{q}{\lambda ^{2}}\right)
& \frac{q^{1/2}\zeta }{\lambda }-\frac{\lambda }{\zeta q^{1/2}}%
\end{array}%
\right) _{a}\in \text{End}(V_{a}).
\label{Kmatuno2}
\end{equation}%
We introduce the functions:%
\begin{equation}
\mathsf{A}_{-}(\lambda )\equiv g_{-}(\lambda )a(\lambda
q^{-1/2})d(1/(q^{1/2}\lambda )),\text{ \ }\mathsf{D}_{-}(\lambda )=k(\lambda
)\mathsf{A}_{-}(q/\lambda ),
\label{premierA}
\end{equation}%
where:%
\begin{eqnarray}
a(\lambda ) &\equiv &a_{0}\prod_{n=1}^{\mathsf{N}}(\frac{\beta _{n}}{\lambda 
}+q^{-1}\frac{b_{n}\alpha _{n}}{a_{n}}\lambda ),\text{ \ }k(\lambda )=\frac{%
\left( \lambda ^{2}-1/\lambda ^{2}\right) }{(\lambda
^{2}/q^{2}-q^{2}/\lambda ^{2})}, \\
d(\lambda ) &\equiv &\frac{(-1)^{\mathsf{N}}}{a_{0}}\prod_{n=1}^{\mathsf{N}}%
\frac{a_{n}c_{n}}{\alpha _{n}}(\frac{1}{\lambda }+q\frac{d_{n}\alpha _{n}}{%
c_{n}\beta _{n}}\lambda ),
\end{eqnarray}%
$a_{0}$ is a free nonzero parameter and%
\begin{equation}
g_{\epsilon }(\lambda )\equiv \frac{(\lambda \alpha _{\epsilon
}/q^{1/2}-q^{1/2}/(\lambda \alpha _{\epsilon }))(\lambda \beta _{\epsilon
}^{-\epsilon }/q^{1/2}+q^{1/2}/(\lambda \beta _{\epsilon }^{-\epsilon }))}{%
\left( \alpha _{\epsilon }-1/\alpha _{\epsilon }\right) \left( \beta
_{\epsilon }+1/\beta _{\epsilon }\right) },
\end{equation}%
where\ $\epsilon =\pm 1$ and we have defined:%
\begin{equation}
\left( \alpha _{\epsilon }-1/\alpha _{\epsilon }\right) \left( \beta
_{\epsilon }+1/\beta _{\epsilon }\right) \equiv \frac{\zeta _{\epsilon
}-1/\zeta _{\epsilon }}{\kappa _{\epsilon }},\text{ \ }\left( \alpha
_{\epsilon }+1/\alpha _{\epsilon }\right) \left( \beta _{\epsilon }-1/\beta
_{\epsilon }\right) \equiv \frac{\zeta _{\epsilon }+1/\zeta _{\epsilon }}{%
\kappa _{\epsilon }},  \label{Def-alfa-beta}
\end{equation}%
Moreover, later we will use: 
\begin{equation}
\mu _{n,h}\equiv \left\{ 
\begin{array}{c}
iq^{1/2}\left( a_{n}\beta _{n}/\alpha _{n}b_{n}\right) ^{1/2}\text{ \ \ }h=+,
\\ 
iq^{1/2}\left( c_{n}\beta _{n}/\alpha _{n}d_{n}\right) ^{1/2}\text{ \ \ }h=-.%
\end{array}%
\right.
\end{equation}%
Following the Sklyanin's paper \cite{OpenCySkly88} the next proposition
holds:

\begin{proposition}
The most general boundary transfer matrix associated to the Bazhanov-Stroganov Lax operator in the cyclic
representations of the reflection algebra is defined by:%
\begin{eqnarray}
\mathcal{T}(\lambda ) &\equiv &\text{tr}_{a}\{K_{a,+}(\lambda )\mathcal{U}%
_{a,-}(\lambda )\} \\
&=&a_{+}(\lambda )\mathcal{A}_{-}(\lambda )+d_{+}(\lambda )\mathcal{D}%
_{-}(\lambda )+b_{+}(\lambda )\mathcal{C}_{-}(\lambda )+c_{+}(\lambda )%
\mathcal{B}_{-}(\lambda ),
\end{eqnarray}%
It is a one parameter family of commuting operators satisfying the following
symmetries proprieties:%
\begin{equation}
\mathcal{T}(\lambda )=\mathcal{T}(1/\lambda ),\text{ \ \ }\mathcal{T}%
(-\lambda )=\mathcal{T}(\lambda ).  \label{symmetry-transfer}
\end{equation}%
The boundary quantum determinant:%
\begin{eqnarray}
\text{det}_{q}\ \mathcal{U}_{a,-}(\lambda ) &\equiv &(\left( \lambda /q\right)
^{2}-\left( q/\lambda \right) ^{2})[\mathcal{A}_{-}(\lambda q^{1/2})\mathcal{%
A}_{-}(q^{1/2}/\lambda )+\mathcal{B}_{-}(\lambda q^{1/2})\mathcal{C}%
_{-}(q^{1/2}/\lambda )]  \label{Bound-q-detU_1} \\
&=&(\left( \lambda /q\right) ^{2}-\left( q/\lambda \right) ^{2})[\mathcal{D}%
_{-}(\lambda q^{1/2})\mathcal{D}_{-}(q^{1/2}/\lambda )+\mathcal{C}%
_{-}(\lambda q^{1/2})\mathcal{B}_{-}(q^{1/2}/\lambda )],
\label{Bound-q-detU_2}
\end{eqnarray}%
is a central element in the reflection algebra, i.e.%
\begin{equation}
\lbrack \text{det}_{q}\ \mathcal{U}_{a,-}(\lambda ),\mathcal{U}_{a,-}(\mu )]=0,
\end{equation}%
and its explicit expression reads:%
\begin{equation}
\text{det}_{q}\ \mathcal{U}_{a,-}(\lambda )=(\lambda ^{2}/q^{2}-q^{2}/\lambda
^{2})\mathsf{A}_{-}(\lambda q^{1/2})\mathsf{A}_{-}(q^{1/2}/\lambda ).
\label{B-q-detU_-exp}
\end{equation}
\end{proposition}

\section{Gauged cyclic reflection algebra and SoV representations}

In our previous paper we solved the spectral problem associated to the
transfer matrix of the cyclic representations under the requirement that one
of the boundary matrices is triangular, i.e. $b_{+}(\lambda )\equiv 0$. In
this paper we want to solve the same type of spectral problem but for the
most general boundary conditions. In order to do so we can follow the same
approach used in the case of the transfer matrix associated to the spin-1/2
reflection algebra \cite{OpenCyFalKN14}. That is, we introduce the following linear
combinations of the original reflection algebra generators:%
\begin{align}
\mathcal{A}_{-}(\lambda |\beta )& =\left[ -\left( \lambda q^{3/2}/\beta
\right) \mathcal{A}_{-}(\lambda )-\alpha q\mathcal{B}_{-}(\lambda )+\mathcal{%
C}_{-}(\lambda )/(\alpha q)+\beta \mathcal{D}_{-}(\lambda )/\left( \lambda
q^{3/2}\right) \right] /(\beta /q^{2}-q^{2}/\beta )  \label{A-gauge-def} \\
\mathcal{B}_{-}(\lambda |\beta )& =\left[ -\left( \lambda \beta
/q^{1/2}\right) \mathcal{A}_{-}(\lambda )-\alpha q\mathcal{B}_{-}(\lambda
)+\left( \beta ^{2}/\alpha q\right) \mathcal{C}_{-}(\lambda )+\left( \beta
q^{1/2}/\lambda \right) \mathcal{D}_{-}(\lambda )\right] /(\beta -1/\beta )
\\
\mathcal{C}_{-}(\lambda |\beta )& =\left[ \left( \lambda q^{3/2}/\beta
\right) \mathcal{A}_{-}(\lambda )+\alpha q\mathcal{B}_{-}(\lambda )-\left(
q^{3}/\alpha \beta ^{2}\right) \mathcal{C}_{-}(\lambda )-\left(
q^{5/2}/\lambda \beta \right) \mathcal{D}_{-}(\lambda )\right] /(\beta
/q^{2}-q^{2}/\beta ) \\
\mathcal{D}_{-}(\lambda |\beta )& =\left[ \left( \lambda \beta
/q^{1/2}\right) \mathcal{A}_{-}(\lambda )+\alpha q\mathcal{B}_{-}(\lambda )-%
\mathcal{C}_{-}(\lambda )/\alpha q-\left( q^{1/2}/\lambda \beta \right) 
\mathcal{D}_{-}(\lambda )\right] /(\beta -1/\beta ),
\end{align}%
where $\beta \neq \pm 1,\pm q^{2}$ and $\alpha $ are arbitrary complex
values; to simplify the notation, we won't explicit the dependance in $%
\alpha $. As it is discussed in the appendices A and B, these operators families
still satisfy a set of commutation relations which are gauged versions of 
 the reflection algebra commutation relations. In the
following we will refer to these families as the gauge transformed
reflection algebra generators. In the same appendices, we prove the following
theorem, characterizing the representation of these generators:

\begin{theorem}
\label{B-Pseudo-diago}For almost all the values of the boundary-bulk-gauge
parameters there exit a left $\langle \Omega _{\beta }|$ and a right $%
|\Omega _{\beta }\rangle $ pseudo-eigenstate of $\mathcal{B}_{-}(\lambda
|\beta )$:%
\begin{equation}
\langle \Omega _{\beta }|B_{-}(\lambda |\beta )=\text{\textsc{b}}_{\text{%
\textbf{0}}}(\lambda |\beta )\langle \Omega _{\beta /q^{2}}|,\text{ \ }%
B_{-}(\lambda |\beta )|\Omega _{\beta }\rangle =|\Omega _{q^{2}\beta
}\rangle \text{\textsc{b}}_{\text{\textbf{0}}}(\lambda |\beta ),
\label{L/R-Pseudo-eigenstate}
\end{equation}%
where:%
\begin{equation}
\text{\textsc{b}}_{\text{\textbf{h}}}(\lambda |\beta )=\text{\textsc{b}}%
_{-}(\beta )(\frac{\lambda ^{2}}{q}-\frac{q}{\lambda ^{2}})\prod_{a=1}^{%
\mathsf{N}}(\frac{\lambda }{\text{\textsc{b}}_{-,a}(\beta )q^{h_{a}}}-\frac{%
\text{\textsc{b}}_{-,a}(\beta )q^{h_{a}}}{\lambda })(\lambda q^{h_{a}}\text{%
\textsc{b}}_{-,a}(\beta )-\frac{1}{\lambda q^{h_{a}}\text{\textsc{b}}%
_{-,a}(\beta )}),  \label{pseudo-eigenvalue}
\end{equation}%
for \textbf{h}$=(h_{1},...,h_{\mathsf{N}})\in \left\{ 0,...,p-1\right\} ^{%
\mathsf{N}}$ with%
\begin{eqnarray}
\text{\textsc{b}}_{-,n}^{2}(\beta ) &\neq &q^{1-2h}\alpha _{-}^{2\epsilon },%
\text{\ \textsc{b}}_{-,n}^{2}(\beta )\neq -q^{1-2h}\beta _{-}^{2\epsilon },%
\text{\ }  \label{SOV-cond1} \\
\text{\textsc{b}}_{-,n}^{2}(\beta ) &\neq &q^{1-2h}\mu _{m,+}^{2\epsilon },%
\text{\ \textsc{b}}_{-,n}^{2}(\beta )\neq q^{1-2h}\mu _{m,-}^{2\epsilon },
\label{SOV-cond2}
\end{eqnarray}%
for any $\epsilon =\pm 1,$ $n,m\in \{1,...,\mathsf{N}\}$ and $h\in \left\{
1,...,p-1\right\} $. Then, the following set of states:%
\begin{eqnarray}
\langle \beta ,h_{1},...,h_{\mathsf{N}}| &=&\frac{1}{\text{\textsc{n}}%
_{\beta }}\langle \Omega _{\beta }|\prod_{n=1}^{\mathsf{N}%
}\prod_{k_{n}=1}^{h_{n}}\frac{\mathcal{A}_{-}(q^{1-k_{n}}/\text{\textsc{b}}%
_{-,n}(\beta )|\beta q^{2})}{\mathsf{A}_{-}(q^{1-k_{n}}/\text{\textsc{b}}%
_{-,n}(\beta ))},  \label{Def-left-P-B-basis} \\
|\beta ,h_{1},...,h_{\mathsf{N}}\rangle &\equiv &\frac{1}{\text{\textsc{n}}%
_{\beta /q^{2}}}\prod_{n=1}^{\mathsf{N}}\prod_{k_{n}=1}^{h_{n}}\frac{%
\mathcal{D}_{-}(q^{1-k_{n}}/\text{\textsc{b}}_{-,n}(\beta )|\beta )}{\mathsf{%
D}_{-}(q^{1-k_{n}}/\text{\textsc{b}}_{-,n}(\beta ))}|\Omega _{\beta }\rangle
,  \label{Def-right-P-B-basis}
\end{eqnarray}%
form a left and a right basis of the representation space defining the
following decomposition of the identity:%
\begin{equation}
\mathbb{I}\equiv \sum_{h_{1},...,h_{\mathsf{N}}=0}^{p-1}\prod_{1\leq b<a\leq 
\mathsf{N}}(X_{a}^{(h_{a})}-X_{a}^{(h_{a})})|\beta q^{2},h_{1},...,h_{%
\mathsf{N}}\rangle \langle \beta ,h_{1},...,h_{\mathsf{N}}|,
\end{equation}
with%
\begin{equation}
\langle \beta ,h_{1},...,h_{\mathsf{N}}|\beta q^{2},k_{1},...,k_{\mathsf{N}%
}\rangle =\prod_{1\leq a\leq \mathsf{N}}\delta _{h_{a},k_{a}}\prod_{1\leq
b<a\leq \mathsf{N}}\frac{1}{X_{a}^{(h_{a})}-X_{b}^{(h_{b})}},
\end{equation}%
where%
\begin{equation}
X_{b}^{(h_{b})}=(\text{\textsc{b}}_{-,b}(\beta )q^{h_{b}})^{2}+1/(\text{%
\textsc{b}}_{-,b}(\beta )q^{h_{b}})^{2}
\label{313}
\end{equation}%
for the non-zero normalization fixed by%
\begin{equation}
\text{\textsc{n}}_{\beta }=\left( \prod_{1\leq b<a\leq \mathsf{N}}\left(
X_{a}^{\left( p-1\right) }-X_{b}^{\left( p-1\right) }\right) \langle \Omega
_{\beta }|\Omega _{\beta q^{2}}\rangle \right) ^{1/2}.  \label{Normaliz-def}
\end{equation}%
In this basis the operator family $\mathcal{B}_{-}(\lambda |\beta )$ is
pseudo-diagonalized:%
\begin{eqnarray}
\langle \beta ,h_{1},...,h_{\mathsf{N}}|\mathcal{B}_{-}(\lambda |\beta ) &=&%
\text{\textsc{b}}_{\text{\textbf{h}}}(\lambda |\beta )\langle \beta
/q^{2},h_{1},...,h_{\mathsf{N}}|, \\
\mathcal{B}_{-}(\lambda |\beta )|\beta ,h_{1},...,h_{\mathsf{N}}\rangle
&=&|q^{2}\beta ,h_{1},...,h_{\mathsf{N}}\rangle \text{\textsc{b}}_{\text{%
\textbf{h}}}(\lambda |\beta ),
\end{eqnarray}%
with simple pseudo-spectrum%
\begin{equation}
\text{\textsc{b}}_{-,n}^{2p}(\beta )\neq \pm 1,\text{ \textsc{b}}%
_{-,m}^{p}(\beta )\neq \text{\textsc{b}}_{-,n}^{p}(\beta ),\text{ \ }\forall
n\neq m\in \{1,...,\mathsf{N}\},
\end{equation}%
and the operator families $\mathcal{A}_{-}(\lambda |\beta )$ and $\mathcal{D}%
_{-}(\lambda |\beta )$ in the zeros $\zeta_{a}^{(h_{a})}$ of $\mathcal{B}_{-}(\lambda |\beta )$
act as simple shift operators:%
\begin{align}
\langle \beta ,h_{1},...,h_{\mathsf{N}}|\mathcal{A}_{-}(\zeta
_{a}^{(h_{a})}|\beta q^{2})& =\mathsf{A}_{-}(\zeta _{a}^{(h_{a})})\langle
\beta ,h_{1},...,h_{\mathsf{N}}|T_{a}^{-\varphi _{a}},  \label{Left-A-action}
\\
\mathcal{D}_{-}(\zeta _{a}^{(h_{a})}|\beta )|\beta ,h_{1},...,h_{\mathsf{N}%
}\rangle & =T_{a}^{-\varphi _{a}}|\beta ,h_{1},...,h_{\mathsf{N}}\rangle 
\mathsf{D}_{-}(\zeta _{a}^{(h_{a})}),  \label{Right-D-action}
\end{align}%
where%
\begin{eqnarray}
\langle \beta ,h_{1},...,h_{a},...,h_{\mathsf{N}}|T_{a}^{\pm } &=&\langle
\beta ,h_{1},...,h_{a}\pm 1,...,h_{\mathsf{N}}|, \\
T_{a}^{\pm }|\beta ,h_{1},...,h_{a},...,h_{\mathsf{N}}\rangle &=&|\beta
,h_{1},...,h_{a}\pm 1,...,h_{\mathsf{N}}\rangle ,
\end{eqnarray}%
\smallskip and:%
\begin{eqnarray}
\zeta _{n}^{(h)} &=&\left( \text{\textsc{b}}_{-,n}(\beta )q^{h}\right)
^{\varphi _{n}}\text{\ \ for \ }h\in \{0,...,p-1\}\text{\ \ and\ \ }\forall
n\in \{1,...,2\mathsf{N}\}, \\
\varphi _{a} &=&1-2\theta (a-\mathsf{N})\text{ \ \ with \ }\theta (x)=\{0%
\text{ for }x\leq 0,\text{ }1\text{ for }x>0\}.
\end{eqnarray}
\end{theorem}

Let us comment that the existence of the states $\langle \Omega _{\beta }|$
and $|\Omega _{\beta }\rangle $ can be proven by a general argument which we
present in Appendix B. For general representations, the pseudo-spectrum of $%
\mathcal{B}_{-}(\lambda |\beta )$, i.e. the values of $\text{\textsc{b}}%
_{-,n}(\beta )$ and $\text{\textsc{b}}_{-}(\beta )$, must be computed by
recursion on the number of sites. However, in Appendix B we present the
explicit expression for $\text{\textsc{b}}_{-,n}(\beta )$ and $\text{\textsc{%
b}}_{-}(\beta )$ in some particular representations.

The interest in these gauge transformed boundary generators is due to the
possibility to use them to rewrite the transfer matrix associated to the
most general cyclic 6-vertex reflection algebra representations in a simple
form, as presented in the following proposition:

\begin{proposition}
The quantum determinant can be written in terms of the gauge transformed
boundary generators as:%
\begin{align}
\frac{\text{det}_{q}\ \mathcal{U}_{-}(\lambda )}{(\lambda
^{2}/q^{2}-q^{2}/\lambda ^{2})}& =\mathcal{A}_{-}(q^{1/2}\lambda ^{\epsilon
}|\beta q^{2})\mathcal{A}_{-}(q^{1/2}/\lambda ^{\epsilon }|\beta q^{2})+%
\mathcal{B}_{-}(q^{1/2}\lambda ^{\epsilon }|\beta )\mathcal{C}%
_{-}(q^{1/2}/\lambda ^{\epsilon }|\beta q^{2}) \\
& =\mathcal{D}_{-}(q^{1/2}\lambda ^{\epsilon }|\beta )\mathcal{D}%
_{-}(q^{1/2}/\lambda ^{\epsilon }|\beta )+\mathcal{C}_{-}(q^{1/2}\lambda
^{\epsilon }|\beta q^{2})\mathcal{B}_{-}(q^{1/2}/\lambda ^{\epsilon }|\beta
),
\end{align}%
$\epsilon =\pm 1$. Moreover, if we set the gauge parameter $\alpha $ to:%
\begin{equation}
\alpha =-\beta \beta _{+}/q^{2}\alpha _{+}e^{\tau _{+}},  \label{guage-alpha}
\end{equation}%
($\alpha_+$ and $\beta_+$ are defined in \eqref{Def-alfa-beta}, they are linked to the boundary parameters $\zeta_+,\kappa_+$ and $\tau_+$, see \eqref{Kmatuno}-\eqref{Kmatuno2}) then the transfer matrix can be written as%
\begin{align}
\mathcal{T}(\lambda )& = \mathsf{a}_{+}(\lambda )\mathcal{A}_{-}(\lambda
|\beta )+\mathsf{a}_{+}(1/\lambda )\mathcal{A}_{-}(1/\lambda |\beta )+q%
\mathsf{c}_{+}(\lambda |\beta )\mathcal{B}_{-}(\lambda |\beta /q^{2})
\label{T-triangu-L} \\
\mathcal{T}(\lambda )& = \mathsf{d}_{+}(\lambda )\mathcal{D}_{-}(\lambda
|\beta )+\mathsf{d}_{+}(1/\lambda )\mathcal{D}_{-}(1/\lambda |\beta )+%
\mathsf{c}_{+}(\lambda |\beta )\mathcal{B}_{-}(\lambda |\beta )/q,
\label{T-triangu-R}
\end{align}%
where we have defined: 
\begin{align}
\mathsf{a}_{+}(\lambda )& =-\frac{\lambda ^{2}q-1/q\lambda ^{2}}{\lambda
^{2}-1/\lambda ^{2}}g_{+}(\lambda ),\text{ \ }\mathsf{d}_{+}(\lambda )=\frac{%
\lambda ^{2}q-1/q\lambda ^{2}}{\lambda ^{2}-1/\lambda ^{2}}g_{+}(q/\lambda ), \label{deuxiemeA}
\\
\mathsf{c}_{+}(\lambda |\beta )& =\frac{-q(\lambda ^{2}q-1/q\lambda
^{2})\left( \beta \beta _{+}/q\alpha _{+}-q\alpha _{+}/\beta \beta
_{+}\right) }{\beta \left( \alpha _{+}-1/\alpha _{+}\right) \left( \beta
_{+}+1/\beta _{+}\right) }.
\end{align}
\end{proposition}

\begin{proof}
The proof of this statement coincides with the one given in \cite{OpenCyFalKN14} for
the XXZ spin 1/2 quantum chain with general integrable boundaries; in fact,
this statement is representation independent. The only difference is that
here we have used a Laurent polynomial form while in the XXZ case it was a
trigonometric form.
\end{proof}

The simple representations $(\ref{T-triangu-L})$-$(\ref{T-triangu-R})$ of
the transfer matrix in terms of the gauge transformed boundary generators
and the known actions $(\ref{Left-A-action})$-$(\ref{Right-D-action})$ of
these operators imply that the transfer matrix spectral problem is separated
in the pseudo-eigenbasis of $\mathcal{B}_{-}(\lambda |\beta )$.

\section{$\mathcal{T}$-spectrum characterization in SoV basis and scalar products}

In this section we present the complete characterization of the spectrum of
the transfer matrix $\mathcal{T}(\lambda )$ associated to the cyclic
representations of the 6-vertex reflection algebra. We first present some
preliminary properties satisfied by all the eigenvalue functions of the
transfer matrix $\mathcal{T}(\lambda )$:

\begin{lemma}
Denote by $\Sigma _{\mathcal{T}}$ the transfer matrix spectrum, then any $%
\tau (\lambda )\in \Sigma _{\mathcal{T}}$ is an even function of $\lambda $\
symmetrical under the transformation $\lambda \rightarrow 1/\lambda $ which
admits the following interpolation formula:%
\begin{align}
\tau (\lambda )=& \sum_{a=1}^{\mathsf{N}}\frac{\Lambda ^{2}-X^{2}}{%
(X_{a}^{(0)})^{2}-X^{2}}\prod_{\substack{ b=1  \\ b\neq a}}^{\mathsf{N}}%
\frac{\Lambda -X_{b}^{(0)}}{X_{a}^{(0)}-X_{b}^{(0)}}\tau (\zeta
_{a}^{(0)})+(-1)^{\mathsf{N}}\frac{(\Lambda +X)}{2}\prod_{b=1}^{\mathsf{N}}%
\frac{\Lambda -X_{b}^{(0)}}{X-X_{b}^{(0)}}\text{det}_{q}M(1)  \notag \\
& -(-1)^{\mathsf{N}}\frac{(\Lambda -X)}{2}\prod_{b=1}^{\mathsf{N}}\frac{%
\Lambda -X_{b}^{(0)}}{X+X_{b}^{(0)}}\frac{(\zeta _{+}+1/\zeta _{+})}{(\zeta
_{+}-1/\zeta _{+})}\frac{(\zeta _{-}+1/\zeta _{-})}{(\zeta _{-}-1/\zeta _{-})%
}\text{det}_{q}M(i)  \notag \\
& +(\Lambda ^{2}-X^{2})\tau _{\infty }\prod_{b=1}^{\mathsf{N}}(\Lambda
-X_{b}^{(0)}),  \label{set-tau}
\end{align}%
where:
\begin{equation}
\Lambda\equiv(\lambda ^{2}+1/\lambda ^{2}) \ \ \text{ and \ \ \ }X\equiv q+1/q
\label{defX}
\end{equation}
and
\begin{equation}
\tau _{\infty }\equiv \frac{\kappa _{+}\kappa _{-}(e^{\tau _{+}-\tau
_{-}}\prod_{b=1}^{\mathsf{N}}\delta _{b}\gamma _{b}+e^{\tau _{-}-\tau
_{+}}\prod_{b=1}^{\mathsf{N}}\alpha _{b}\beta _{b})}{\left( \zeta
_{+}-1/\zeta _{+}\right) \left( \zeta _{-}-1/\zeta _{-}\right) }\text{ .\ \
\ }
\end{equation}
\end{lemma}
We recall that $\zeta_-,\kappa_-,\tau_-,\zeta_+,\kappa_+$ and $\tau_+$ are the boundary parameters, see \eqref{Kmatuno}-\eqref{Kmatuno2}, and $\alpha_n,\beta_n,\gamma_n$ and $\delta_n$ are the bulk parameters, see \eqref{laxopppp}.
\begin{proof}
This lemma coincides with Lemma 5.1 of our previous paper.
\end{proof}

We introduce the following one-parameter family $D_{\tau }(\lambda )$ of $%
p\times p$ matrices:%
\begin{equation}
D_{\tau }(\lambda )\equiv 
\begin{pmatrix}
\tau (\lambda ) & -\text{\textsc{a}}{}(1/\lambda ) & 0 & \cdots & 0 & -\text{%
\textsc{a}}(\lambda ) \\ 
-\text{\textsc{a}}(q\lambda ) & \tau (q\lambda ) & -\text{\textsc{a}}%
(1/\left( q\lambda \right) ) & 0 & \cdots & 0 \\ 
0 & {\quad }\ddots &  &  &  & \vdots \\ 
\vdots &  & \cdots &  &  & \vdots \\ 
\vdots &  &  & \cdots &  & \vdots \\ 
\vdots &  &  &  & \ddots {\qquad } & 0 \\ 
0 & \ldots & 0 & -\text{\textsc{a}}(q^{2l-1}\lambda ) & \tau
(q^{2l-1}\lambda ) & -\text{\textsc{a}}(1/\left( q^{2l-1}\lambda \right) )
\\ 
-\text{\textsc{a}}(1/\left( q^{2l}\lambda \right) ) & 0 & \ldots & 0 & -%
\text{\textsc{a}}(q^{2l}\lambda ) & \tau (q^{2l}\lambda )%
\end{pmatrix}%
,  \label{D-matrix}
\end{equation}%
where for now $\tau (\lambda )$ is a generic function and we have defined:%
\begin{equation}
\text{\textsc{a}}(\lambda )=\mathsf{a}_{+}(\lambda )\mathsf{A}_{-}(\lambda ),
\label{45}
\end{equation}%
(from \eqref{premierA} and \eqref{deuxiemeA}) where the coefficient \textsc{a}$(\lambda )$ satisfies the quantum
determinant condition:%
\begin{equation}
\text{\textsc{a}}(\lambda q^{1/2})\text{\textsc{a}}(q^{1/2}/\lambda )=\frac{%
\mathsf{a}_{+}(\lambda q^{1/2})\mathsf{a}_{+}(q^{1/2}/\lambda )\text{det}_{q}%
\mathcal{U}_{-}(\lambda )}{\left( \lambda /q\right) ^{2}-\left( q/\lambda
\right) ^{2}}.
\end{equation}%
The separation of variables lead to the following discrete characterization
of the transfer matrix spectrum.

\begin{theorem}
\label{discrete-SoV-ch} For almost all the values of the boundary-bulk
parameters $\mathcal{T}(\lambda )$ is diagonalizable and it has simple spectrum and $\Sigma _{%
\mathcal{T}}$ coincides with the set of polynomials $\tau (\lambda )$ of the
form $(\ref{set-tau})$ which satisfy the following discrete system of
equations:%
\begin{equation}
\text{det}\ \text{$D$}_{\tau }(\zeta _{a}^{(0)})=0,\text{ }\forall a\in
\{1,...,\mathsf{N}\}.  \label{OpenCyI-Functional-eq}
\end{equation}

\begin{itemize}
\item[\textsf{I)}] The right $\mathcal{T}$-eigenstate corresponding to $\tau
(\lambda )\in \Sigma _{\mathcal{T}}$ is defined by the following
decomposition in the right SoV-basis:%
\begin{equation}
|\tau \rangle =\sum_{h_{1},...,h_{\mathsf{N}}=0}^{p-1}\prod_{a=1}^{\mathsf{N}%
}q_{\tau ,a}^{(h_{a})}\prod_{1\leq b<a\leq \mathsf{N}%
}(X_{a}^{(h_{a})}-X_{b}^{(h_{b})})|\beta ,h_{1},...,h_{\mathsf{N}}\rangle ,
\label{OpenCyeigenT-r-D}
\end{equation}%
where the gauge parameters $\alpha $ and $\beta $ satisfy the condition $%
\left( \ref{guage-alpha}\right) $ and the $q_{\tau ,a}^{(h_{a})}$ are the
unique nontrivial solutions up to normalization of the linear homogeneous
system:%
\begin{equation}
\text{$D$}_{\tau }(\zeta _{a}^{(0)})\left( 
\begin{array}{c}
q_{\tau ,a}^{(0)} \\ 
\vdots \\ 
q_{\tau ,a}^{(p-1)}%
\end{array}%
\right) =\left( 
\begin{array}{c}
0 \\ 
\vdots \\ 
0%
\end{array}%
\right) .  \label{OpenCyt-Q-relation}
\end{equation}

\item[\textsf{II)}] The left $\mathcal{T}$-eigenstate corresponding to $\tau
(\lambda )\in \Sigma _{\mathcal{T}}$ is defined by the following
decomposition in the left SoV-basis:%
\begin{equation}
\langle \tau |=\sum_{h_{1},...,h_{\mathsf{N}}=0}^{p-1}\prod_{a=1}^{\mathsf{N}%
}\hat{q}_{\tau ,a}^{(h_{a})}\prod_{1\leq b<a\leq \mathsf{N}%
}(X_{a}^{(h_{a})}-X_{b}^{(h_{b})})\langle h_{1},...,h_{\mathsf{N}},\beta
/q^{2}|,  \label{OpenCyeigenT-l-D}
\end{equation}%
where the gauge parameters $\alpha $ and $\beta $ satisfy the condition $%
\left( \ref{guage-alpha}\right) $ and the $\hat{q}_{\tau ,a}^{(h_{a})}$ are
the unique nontrivial solutions up to normalization of the linear
homogeneous system:%
\begin{equation}
\left( 
\begin{array}{ccc}
\hat{q}_{\tau ,a}^{(0)} & \ldots & \hat{q}_{\tau ,a}^{(p-1)}%
\end{array}%
\right) \left( \text{$\hat{D}$}_{\tau }(\zeta _{a}^{(0)})\right)
^{t_{0}}=\left( 
\begin{array}{ccc}
0 & \ldots & 0%
\end{array}%
\right) ,
\end{equation}%
and $\hat{D}$$_{\tau }(\lambda )$ is the family of $p\times p$ matrices
defined substituting in $D_{\tau }(\lambda )$ the coefficient \textsc{a}$%
(\lambda )$ with 
\begin{equation}
\text{\textsc{d}}(\lambda )=\mathsf{d}_{+}(\lambda )\mathsf{D}_{-}(\lambda ).
\label{412}
\end{equation}
\end{itemize}
defined from \eqref{premierA} and \eqref{deuxiemeA}.
\end{theorem}

\begin{proof}
The Theorem \ref{B-Pseudo-diago} implies that for almost all the values of
the gauge-boundary-bulk parameters the conditions $\left( \ref{SOV-cond1}%
\right) $-$\left( \ref{SOV-cond2}\right) $ hold. Here, we need to prove
also that for almost all the values of the boundary-bulk parameters we have,%
\begin{equation}
\text{\textsc{b}}_{-,n}^{2}(\beta )\neq q^{1-2h}\alpha _{+}^{\pm 2},\text{\ 
\textsc{b}}_{-,n}^{2}(\beta )\neq -q^{1-2h}\beta _{+}^{\pm 2},\text{ \ }%
\forall h\in \{1,...,p-1\},n\in \{1,...,\mathsf{N}\},  \label{SoV-T-simple}
\end{equation}%
once we set the ratio $\alpha /\beta $ as in $\left( \ref{guage-alpha}\right) $%
. Let us first observe that $\mathcal{B}_{-}(\lambda |\beta )$ is a Laurent
polynomial in $\alpha ,$ $\beta $, the inner boundary parameters and the
bulk parameters. So that by $\left( \ref{guage-alpha}\right) $, the one
parameter family $\mathcal{B}_{-}(\lambda |\beta )$ becomes Laurent
polynomial in the outer boundary parameters too. Consequently, to prove that 
$\left( \ref{SoV-T-simple}\right) $ is satisfied for almost all the values
of the boundary-bulk parameters it is enough to prove that we can find some
values of these parameters for which $\left( \ref{SoV-T-simple}\right) $ is
satisfied. Indeed, we can chose arbitrary boundary-bulk parameters
satisfying the following inequalities:%
\begin{equation}
\mu _{+,n}^{p}\neq \alpha _{+}^{\pm p}\text{ and }\mu _{+,n}^{p}\neq -\beta
_{+}^{\pm p},\text{ \ }\forall n\in \{1,...,\mathsf{N}\},
\end{equation}%
together with those in $\left( \ref{Special-B-simple}\right) $ and $\left( %
\ref{Special-B-simple2}\right) $ and impose the $\mathsf{N}$ conditions $%
\left( \ref{Condi-bulk-quasi-L}\right) $. Under these conditions, Theorem %
\ref{B-Pseudo-diago} implies the pseudo-diagonalizability of $\mathcal{B}%
_{-}(\lambda |\beta )$ and fixes the spectrum of its zeros \textsc{b}$%
_{-,n}(\beta )$ by $\left( \ref{Spectrum-Zeros-B}\right) $; so that the
inequality $\left( \ref{SoV-T-simple}\right) $ is satisfied.

As we have proven that for almost all the values of the boundary-bulk
parameters the inequalities $\left( \ref{SoV-T-simple}\right) $, $\left( \ref%
{Special-B-simple}\right) $ and $\left( \ref{Special-B-simple2}\right) $
hold, to prove this theorem we have just to follow the same proof given in
the non-gauged case, i.e. the proof of Theorem 5.1 of our previous paper. 

Let us comment that with respect to this last theorem here we are stating also the diagonalizability of the transfer matrix for almost any value of the parameters of the representation. This last statement can be proven as it follows.
Let us consider the following special representation, where the bulk parameters satisfy:
\begin{equation}
c_n=-b_n^* \ \ ; \ \ d_n=-a_n^* \ \ \text{and} \ \ \ \alpha_n^* \beta_n = a_n^* b_n
\end{equation}
and where the boundary matrices are diagonal, $K_{a,-}(\lambda)=K_a(\lambda;\zeta_-,0,0)$ and \\
$K_{a,+}(\lambda)=K_a(q \lambda;\zeta_+,0,0)$ (see \eqref{Kmatuno}-\eqref{Kmatuno2}), with the associated boundary parameters satisfying moreover $|\zeta_-|=|\zeta_+|=1$. The $*$ operation is the complex conjugation. A simple direct calculation made for example in \cite{OpenCyGN12} leads to the following Hermitian conjugate of the monodromy matrix \eqref{24}:
\begin{equation}
M_{a}^{\dag}(\lambda)  = \sigma_a^y M_{a}(\lambda^*)\sigma_a^y
\end{equation}
where $\sigma_a^y$ denotes the Pauli matrix.

From this relation, and using the specific inner boundary matrix introduced, one can compute the Hermitian conjugate of the boundary monodromy matrix \eqref{BMM}:
\begin{equation}
{\cal{U}}_{a,-}^{\dag}(\lambda) = {\cal{U}}_{a,-}^{t_a}(1/\lambda^*)
\end{equation}
Then, from the definition of the boundary transfer matrix, and for the special choice of representation here chosen, we can show:
\begin{equation}
{\cal{T}}^{\dag}(\lambda) = {\cal{T}}(1/\lambda^*) 
\end{equation}
Thus for this special representation the boundary transfer matrix is normal. Then it follows that the determinant of the $p^\mathsf{N}\times p^\mathsf{N}$ matrix of elements $\langle e_{i}|\tau_{j}\rangle $, where $\bra{e_{i}}$ is the generic element of a given basis of covectors and  $\ket{\tau_{j}}$ is the generic transfer matrix eigenvector, is non zero. 

Noticing that this determinant is a fractional function of the bulk and boundary parameters, non zero for the special choice of the parameters above defined, it follows that it is non zero for almost every choice of the parameters. Which concludes the proof.
\end{proof}

It is also interesting to remark that we can obtain the coefficients of a
left transfer matrix eigenstates in terms of those of the right one. The
following lemma defines this characterization and can be proven as in the standard case \cite{MNP2017}:

\begin{lemma}
Let $\tau (\lambda )\in \Sigma _{\mathcal{T}}$ then it holds:%
\begin{equation}
\frac{\hat{q}_{\tau ,a}^{(h)}}{\hat{q}_{\tau ,a}^{(h-1)}}=\frac{\text{%
\textsc{a}}(1/\zeta _{a}^{(h-1)})}{\text{\textsc{d}}(1/\zeta _{a}^{(h-1)})}%
\frac{q_{\tau ,a}^{(h)}}{q_{\tau ,a}^{(h-1)}}.
\label{419}
\end{equation}
\end{lemma}

Let us introduce a class of left and right states, the so-called separate
states, characterised by the following type of decompositions in the left
and right separate basis: 
\begin{align}
\langle \alpha |& \equiv \sum_{h_{1},...,h_{\mathsf{N}}=0}^{p-1}\prod_{a=1}^{%
\mathsf{N}}\alpha _{a}^{(h_{a})}\prod_{1\leq b<a\leq \mathsf{N}%
}(X_{a}^{(h_{a})}-X_{b}^{(h_{b})})\langle \beta/q^2, h_{1},...,h_{\mathsf{N}}|
\label{Separate-left-SoV} \\
|\beta \rangle & =\sum_{h_{1},...,h_{\mathsf{N}}=0}^{p-1}\prod_{a=1}^{%
\mathsf{N}}\beta _{a}^{(h_{a})}\prod_{1\leq b<a\leq \mathsf{N}%
}(X_{a}^{(h_{a})}-X_{b}^{(h_{b})})|\beta,h_{1},...,h_{\mathsf{N}}\rangle 
\label{Separate-right-SoV}
\end{align}%
where the coefficients $\alpha _{a}^{(h_{a})}$ and $\beta _{a}^{(h_{a})}$ are arbitrary complex numbers, meaning that the coefficients of these separate states have a factorised form in these basis. ($X_{a}^{(h_{a})}$ defined in \eqref{313}).

These separate states are interesting at least for two reasons:  the
eigenstates of the boundary transfer matrix are special separate states, and they admit a simple determinant scalar
product, as it is stated in the next proposition:
\begin{proposition}
Let us take an arbitrary separate left state $\langle \alpha |$ and an arbitrary separate right state $|\beta \rangle $. Then it holds:
\begin{equation}
\langle \alpha |\beta \rangle =det\ \mathcal{M}^{\left(\alpha ,\beta \right) }
 \label{T2-Sov-Sc-p00}
\end{equation}
where the elements of the size N matrix $\mathcal{M}^{\left(\alpha ,\beta \right) }$ are given by:
\begin{equation}
\forall (a,b) \in [1,N]^2,\ \  \mathcal{M}_{a,b}^{\left( \alpha,\beta \right) }\equiv \sum_{h=0}^{p-1}\alpha _{a}^{(h)}\beta
_{a}^{(h)}(X_{a}^{(h)})^{(b-1)} \label{T2-Sov-Sc-p1}
\end{equation}
\end{proposition}
The proof is quite straightforward, it is based on the fact that one can see a Vandermonde determinant when computing the scalar product. One of the main corollary is the orthogonality of two eigenstates $\langle \tau |$ and $|\tau^{\prime }\rangle $  of the boundary transfer matrix associated to two different eigenvalues $\tau (\lambda )$ and $\tau ^{\prime }(\lambda )$:
\begin{equation}
 \langle \tau |\tau^{\prime }\rangle=0
\end{equation}

The computation of such scalar products is the very first step towards the dynamics, several further steps being required to reach this characterization for the models associated to cyclic representations of the 6-vertex reflection algebra: the reconstruction of the local operators in separate variables, the identification of the ground state, the homogeneous and the thermodynamic limit. For example a rewriting of the determinant representations for the form factors obtained from separation of variable will be necessary to overcome the standard problems  related to the homogeneous limit. This problem has been addressed and solved for the XXX spin 1/2 chain, linking the separation of variable type determinants  with Izergin's, Slavnov's and Gaudin's type determinants \cite{OpenCyKitMNT15,KitMNT16}. \\ 

\section{Functional equation characterizing the $\mathcal{T}$-spectrum}
The purpose of this section is to characterize the spectrum by functional relations analogous to Baxter's T-Q equation. To begin with, we first need the following property. 
\begin{lemma}
Let $\tau (\lambda )$ be a function of $\lambda $\ invariant under the
transformation $\lambda \rightarrow 1/\lambda $ and $\lambda \rightarrow
-\lambda $ then det$_{p}D_{\tau }(\lambda )$ (from \eqref{D-matrix}) is a function of 
\begin{equation}
Z=\lambda ^{2p}+\frac{1}{\lambda ^{2p}},
\label{defZ}
\end{equation}%
i.e. it is a function of $\lambda ^{p}$ invariant under the transformations $\lambda ^{p}\rightarrow
1/\lambda ^{p}$ and $\lambda \rightarrow -\lambda $. Moreover, if $\tau
(\lambda )$ is a Laurent polynomial of degree $\mathsf{N}+2$ in $\Lambda =\lambda ^{2}+\frac{1}{\lambda ^{2}}$
then det$_{p} D_{\tau }(\lambda )$ is a Laurent polynomial of degree $\mathsf{%
N}+2$ in $Z$.
\end{lemma}

\begin{proof}
The first part of this lemma about the dependence w.r.t. $Z$ of det$%
_{p}D_{\tau }(\lambda )$ has been already proven in Lemma 5.2 of our
previous paper \cite{MNP2017} while the second part of this lemma can be proven following
the proof given in Proposition 6.1 of the same paper. To adapt this
proof here, let us observe that the matrix $D_{\tau }(i^{a}q^{h+1/2})$ for $%
a\in \{0,1\}$ and $h\in \{0,...,p-1\}$ contains one row with two divergent
elements, i.e. $-$\textsc{a}$(\pm 1)$ and $-$\textsc{a}$(\pm i)$,
respectively for $a=0$ and $a=1$. Nevertheless the determinants det$%
_{p}D_{\tau }(i^{a}q^{h+1/2})$ are all finites for any $a\in \{0,1\}$ and $%
h\in \{0,...,p-1\}$ if $\tau (i^{b}q^{k+1/2})$ are finite for any $b\in
\{0,1\}$ and $k\in \{0,...,p-1\}$. Indeed, by the symmetries $\lambda
^{p}\rightarrow 1/\lambda ^{p}$ and $\lambda \rightarrow -\lambda $ all the
determinants det$_{p}D_{\tau }(q^{h+1/2})$ coincide as well as all the
determinants det$_{p}D_{\tau }(iq^{h+1/2})$. So that we have to prove our
statement for one value of $q^{h+1/2}$ and one value of $iq^{h+1/2}$. Now,
we can use the expansion of the determinant w.r.t. the central row: 
\begin{align}
\text{det}_{p} D_{\tau }(\lambda q^{1/2})& =\tau (\lambda )\text{det}%
_{p-1}D_{\tau ,(p+1)/2,(p+1)/2}(\lambda q^{1/2})+\frac{\text{\textsc{x}}%
(\lambda )\text{det}_{p-1}D_{\tau ,(p+1)/2,(p+1)/2-1}(\lambda q^{1/2})}{%
\lambda ^{2}-1/\lambda ^{2}}  \notag \\
& -\frac{\text{\textsc{x}}(1/\lambda )\text{det}_{p-1}D_{\tau
,(p+1)/2,(p+1)/2+1}(\lambda q^{1/2})}{\lambda ^{2}-1/\lambda ^{2}},
\end{align}%
where%
\begin{equation}
\text{\textsc{x}}(\lambda )=\left( \lambda ^{2}-\frac{1}{\lambda ^{2}}%
\right) \text{\textsc{a}}(\lambda ),
\end{equation}%
and $D_{\tau ,i,j}(\lambda )$ denotes the $(p-1)\times (p-1)$ matrix
obtained from $D_{\tau }(\lambda )$ removing the row $i$ and the column $j$.
From the identity:%
\begin{equation}
\text{det}_{p-1}D_{\tau ,(p+1)/2,(p+1)/2+1}(\lambda q^{1/2})=\text{det}%
_{p-1}D_{\tau ,(p+1)/2,(p+1)/2-1}(q^{1/2}/\lambda ),
\end{equation}%
and the regularity of these two determinants for $\lambda \rightarrow \pm 1$
and $\lambda \rightarrow \pm i$, it follows that det$_{p}D_{\tau
}(i^{a}q^{1/2})$ are finites too for $a\in \{0,1\}$. Now, our statement
about the Laurent polynomiality of degree $\mathsf{N}+2$ of det$_{p}D_{\tau
}(\lambda )$ w.r.t. $Z$ follows from the symmetries and from the fact that $%
\tau (\lambda )$ and \textsc{x}$(\lambda )$ are Laurent polynomials in $%
\lambda $ of degree $2\mathsf{N}+4$.
\end{proof}

Let us introduce the following notations:%
\begin{eqnarray}
\text{\textsc{a}}_{\infty } &=&\lim_{\lambda \rightarrow +\infty }\lambda
^{-2(\mathsf{N}+2)}\text{\textsc{a}}(\lambda )=\frac{(-1)^{\mathsf{N}%
+1}\kappa _{+}\kappa _{-}\alpha _{-}\beta _{-}\alpha _{+}\prod_{n=1}^{%
\mathsf{N}}b_{n}c_{n}}{q^{3+\mathsf{N}}\beta _{+}\left( \zeta _{+}-1/\zeta
_{+}\right) \left( \zeta _{-}-1/\zeta _{-}\right) }, \label{55a}\\
\text{\textsc{a}}_{0} &=&\lim_{\lambda \rightarrow 0}\lambda ^{2(\mathsf{N}%
+2)}\text{\textsc{a}}(\lambda )=\frac{(-1)^{\mathsf{N}+1}q^{3+\mathsf{N}%
}\kappa _{+}\kappa _{-}\beta _{+}\prod_{n=1}^{\mathsf{N}}a_{n}d_{n}}{\alpha
_{-}\beta _{-}\alpha _{+}\left( \zeta _{+}-1/\zeta _{+}\right) \left( \zeta
_{-}-1/\zeta _{-}\right) }, 
\end{eqnarray}%
and%
\begin{equation}
F(\lambda )=\prod_{b=1}^{2\mathsf{N}}\left( \frac{\lambda ^{p}}{\left( \zeta
_{b}^{(0)}\right) ^{p}}-\frac{\left( \zeta _{b}^{(0)}\right) ^{p}}{\lambda
^{p}}\right) ,
\label{55c}
\end{equation}

where we recall that $\zeta_-,\kappa_-,\tau_-,\alpha_-,\beta_-,\zeta_+,\kappa_+,\tau_+,\alpha_+$ and $\beta_+$ are the boundary parameters (see \eqref{Kmatuno},\eqref{Kmatuno2} and \eqref{Def-alfa-beta}), while $a_n,b_n,c_n$ and $d_n$ are the bulk parameters, see \eqref{laxopppp}. Then the following results hold:

\begin{proposition}
\label{General F-Eq}For almost all the values of the boundary-bulk
parameters, $\mathcal{T}(\lambda )$ has simple spectrum and $\tau (\lambda )$
of the form $(\ref{set-tau})$ is an element of $\Sigma _{\mathcal{T}}$ (the set of the eigenvalues of $\mathcal{T}(\lambda )$) if
and only if det$_{p}D_{\tau }(\lambda )$ is a Laurent polynomial of degree $%
\mathsf{N}+2$ in the variable $Z$ (see \eqref{defZ}) which satisfies the following functional
equation:%
\begin{equation}
\text{det}_{p}D_{\tau }(\lambda )=F(\lambda )\left( \lambda ^{2p}-\frac{1}{%
\lambda ^{2p}}\right) ^{2}\prod_{k=0}^{p-1}(\tau _{\infty }-(q^{k}\text{%
\textsc{a}}_{\infty }+q^{-k}\text{\textsc{a}}_{0})).  \label{Func-EQ-1}
\end{equation}
\end{proposition}

\begin{proof}
The SoV characterization of the spectrum implies that $\tau (\lambda )\in
\Sigma _{\mathcal{T}}$ if and only if it holds:%
\begin{equation}
\text{det}_{p}D_{\tau }(\zeta _{a}^{(0)})=0,\text{ }\forall a\in \{1,...,%
\mathsf{N}\},
\end{equation}%
and $\tau (\lambda )$ has the form $(\ref{set-tau})$. In the previous lemma
we have shown that det$_{p}D_{\tau }(\lambda )$ is a Laurent polynomial of
degree $\mathsf{N}+2$ in $Z$, here we show that from $\tau (\lambda )$ of
form $(\ref{set-tau})$ it follows the identities:%
\begin{equation}
\lim_{\lambda \rightarrow \pm 1,\pm i}\text{det}_{p}D_{\tau }(\lambda
q^{1/2+h})=0\text{ \ \ }\forall h\in \{0,...,p-1\}.
\end{equation}
For the symmetry it is enough to consider the above limit in the case $h=0$.
Let us denote with $\bar{D}_{\tau }(\lambda q^{1/2})$ the matrix whose first
row is the sum of the first and the last row of $D_{\tau }(\lambda q^{1/2})$
divided for $(\lambda ^{2}-1/\lambda ^{2})$ and whose row $\left( p+1\right)
/2$ is the row $\left( p+1\right) /2$ of $D_{\tau }(\lambda q^{1/2})$
multiplied for $(\lambda ^{2}-1/\lambda ^{2})$ while all the others rows of $%
\bar{D}_{\tau }(\lambda q^{1/2})$ and $D_{\tau }(\lambda q^{1/2})$ coincide.
Clearly it holds:%
\begin{equation}
\text{det}_{p}\bar{D}_{\tau }(\lambda q^{1/2})=\text{det}_{p}D_{\tau
}(\lambda q^{1/2}),
\end{equation}%
so that we can compute the limits directly for det$_{p}\bar{D}_{\tau
}(\lambda q^{1/2})$. The interesting point is that now all the rows of the
matrix $\bar{D}_{\tau }(\lambda q^{1/2})$ are finites in the limits $\lambda
\rightarrow \pm 1,\pm i$, this is a consequence of the identities:%
\begin{equation}
\tau (\pm i^{a}q^{1/2})=\text{\textsc{a}}(\pm i^{a}q^{1/2}),\text{ \ \ 
\textsc{a}}(\pm i^{a}q^{-1/2})=0,\text{ }\forall a\in \{0,1\}.
\end{equation}%
Explicitly, we have that the nonzero elements of the rows $1$, $\left(
p+1\right) /2$ and $p$ are:%
\begin{eqnarray}
&&\left[ \bar{D}_{\tau }(\pm i^{a}q^{1/2})\right] _{1,1}\left. =\right.
r_{a,\pm },\text{ }\left[ \bar{D}_{\tau }(\pm i^{a}q^{1/2})\right]
_{1,2}\left. =\right. s_{a,\pm }, \\
&&\left[ \bar{D}_{\tau }(\pm i^{a}q^{1/2})\right] _{1,p-1}\left. =\right.
-s_{a,\pm },\text{ }\left[ \bar{D}_{\tau }(\pm i^{a}q^{1/2})\right]
_{1,p}\left. =\right. -r_{a,\pm }, \\
&&\left[ \bar{D}_{\tau }(\pm i^{a}q^{1/2})\right] _{(p+1)/2,(p+1)/2-1}\left.
=\right. -\text{\textsc{x}}(\pm i^{a}),\text{ } \\
&&\left[ \bar{D}_{\tau }(\pm i^{a}q^{1/2})\right] _{(p+1)/2,(p+1)/2+1}\left.
=\right. \text{\textsc{x}}(\pm i^{a}), \\
&&\left[ \bar{D}_{\tau }(\pm i^{a}q^{1/2})\right] _{p,1}\left. =\right.
-\tau (\pm i^{a}q^{1/2}),\text{ }\left[ \bar{D}_{\tau }(\pm i^{a}q^{1/2})%
\right] _{p,p}\left. =\right. \tau (\pm i^{a}q^{1/2}),
\end{eqnarray}%
where we have defined:%
\begin{equation}
r_{a,\pm }=(-1)^{a}\lim_{\lambda \rightarrow \pm 1}\frac{\tau
(i^{a}q^{1/2}\lambda )-\text{\textsc{a}}(i^{a}q^{1/2}/\lambda )}{\lambda
^{2}-1/\lambda ^{2}},\text{ \ }s_{a,\pm }=(-1)^{a}\lim_{\lambda \rightarrow
\pm 1}\frac{\text{\textsc{a}}(i^{-a}q^{-1/2}/\lambda )}{\lambda
^{2}-1/\lambda ^{2}}.
\end{equation}%
The remaining rows of $\bar{D}_{\tau }(\pm i^{a}q^{1/2})$ produce the
tridiagonal part of this matrix. Then, it is possible to prove that this
matrix has linear dependent rows; so that det$_{p}\bar{D}_{\tau }(\pm
i^{a}q^{1/2})=0$. Finally, we can compute the following asymptotic formulae:%
\begin{align}
\Delta _{\infty }& \equiv \lim_{\lambda \rightarrow \infty }\lambda ^{-2p(%
\mathsf{N}+2)}\text{det}_{p}D_{\tau }(\lambda )=\text{det}_{p}\left[
\lim_{\lambda \rightarrow \infty }\lambda ^{-2(\mathsf{N}+2)}D_{\tau
}(\lambda )\right] \\
& =\lim_{\lambda \rightarrow 0}\lambda ^{2p(\mathsf{N}+2)}\text{det}%
_{p}D_{\tau }(\lambda )=\text{det}_{p}\left[ \lim_{\lambda \rightarrow
0}\lambda ^{2(\mathsf{N}+2)}D_{\tau }(\lambda )\right] ^{t} \\
& =\text{det}_{p}%
\begin{pmatrix}
\tau _{\infty } & -\text{\textsc{a}}_{0} & 0 & \cdots & 0 & -\text{\textsc{a}%
}_{\infty } \\ 
-x\text{\textsc{a}}_{\infty } & x\tau _{\infty } & -x\text{\textsc{a}}_{0} & 
0 & \cdots & 0 \\ 
0 & -x^{2}\text{\textsc{a}}_{\infty } & x^{2}\tau _{\infty } & -x^{2}\text{%
\textsc{a}}_{0} & \ddots & \vdots \\ 
\vdots &  & \ddots & \ddots & 0 & 0 \\ 
0 & \ldots & 0 & -x^{2l-1}\text{\textsc{a}}_{\infty } & x^{2l-1}\tau
_{\infty } & -x^{2l-1}\text{\textsc{a}}_{0} \\ 
-x^{2l}\text{\textsc{a}}_{0} & 0 & \ldots & 0 & -x^{2l}\text{\textsc{a}}%
_{\infty } & x^{2l}\tau _{\infty }%
\end{pmatrix}%
,
\end{align}%
where we have denoted with $^{t}$ the transpose of the matrix and $x=q^{2(%
\mathsf{N}+2)}$. We have that $\Delta _{\infty }$ is a degree $p$ polynomial
in $\tau _{\infty }$ whose zeros are known from the identities:%
\begin{equation}
\left. \Delta _{\infty }\right\vert _{\tau _{\infty }=q^{k}\text{\textsc{a}}%
_{\infty }+q^{-k}\text{\textsc{a}}_{0}}=0\text{ \ }\forall k\in
\{0,...,p-1\},
\end{equation}%
so that we get:%
\begin{equation}
\Delta _{\infty }=\prod_{k=0}^{p-1}(\tau _{\infty }-(q^{k}\text{\textsc{a}}%
_{\infty }+q^{-k}\text{\textsc{a}}_{0})).
\end{equation}%
This means that we have determined det$_{p}D_{\tau }(\lambda )$ in $\mathsf{N%
}+2$ different values of $Z$ together with the asymptotic for $Z\rightarrow
\infty $. From which the characterization $(\ref{Func-EQ-1})$ trivially
follows.
\end{proof}

The discrete characterization of the spectrum given in Theorem \ref%
{discrete-SoV-ch} can be reformulated in terms of Baxter's type
T-Q functional equations and the eigenstates admit an algebraic Bethe ansatz
like reformulation, as we show in the next theorem. These type of
reformulations of the spectrum holds for several models once they admit SoV
description, see for example \cite%
{OpenCyNT-10,OpenCyN-10,OpenCyGN12,OpenCyKitMN14,OpenCyNicT15-2,OpenCyLevNT15,OpenCyNicT15}%
.

In the following we denote with $Q(\lambda )$ a polynomial in $\Lambda=\lambda^2+\frac{1}{\lambda^2}$ of
degree $\mathsf{N}_{Q}$ of the form:%
\begin{equation}
Q(\lambda )=\prod_{b=1}^{\mathsf{N}_{Q}}\left( \Lambda -\Lambda _{b}\right) .
\label{Q-form}
\end{equation}

\begin{theorem}
For almost all the values of the boundary-bulk parameters such that:%
\begin{equation}
\tau _{\infty }\neq q^{-k}\textsc{a}_{\infty }+q^{k}\textsc{a}_{0}\text{ \ }%
\forall k\in \{0,...,p-1\},  \label{Most-general.case}
\end{equation}%
$\tau (\lambda )\in \Sigma _{\mathcal{T}}$ (the set of the eigenvalues of $\mathcal{T}(\lambda)$) if and only if $\tau (\lambda )$
is an entire function and there exists and is unique a polynomial $Q(\lambda
)$ of the form $\left( \ref{Q-form}\right) $ with $\mathsf{N}_{Q}=\left(
p-1\right) \mathsf{N}$, satisfying the following functional equation:%
\begin{equation}
\tau (\lambda )Q(\lambda )=\text{\textsc{a}}(\lambda )Q(\lambda /q)+\text{%
\textsc{a}}(1/\lambda )Q(\lambda q)+\left[ \tau _{\infty }-(q^{-\mathsf{N}%
_{Q}}\text{\textsc{a}}_{\infty }+q^{\mathsf{N}_{Q}}\text{\textsc{a}}_{0})%
\right] \left( \Lambda ^{2}-X^{2}\right) F(\lambda ),  \label{Inho-Baxter-EQ}
\end{equation}%
and the conditions:%
\begin{equation}
(Q(\zeta _{a}^{\left( 0\right) }),...,Q(\zeta _{a}^{\left( p-1\right)
}))\neq (0,...,0)\text{ \ \ }\forall a\in \{1,...,\mathsf{N}\}\text{.}
\label{Q-condition}
\end{equation}
We recall that $\textsc{a}_{\infty },\textsc{a}_{0}$ and $F(\lambda)$ are defined in \eqref{55a}-\eqref{55c} and that $X=q+1/q$ \eqref{defX}.
\end{theorem}

\begin{proof}
Let us prove first that if it exists a $Q(\lambda )$ of the form $\left( \ref%
{Q-form}\right) $ with $\mathsf{N}_{Q}=\left( p-1\right) \mathsf{N}$
satisfying $\left( \ref{Q-condition}\right) $ and $\left( \ref%
{Inho-Baxter-EQ}\right) $ with $\tau (\lambda )$ an entire function, then $\tau (\lambda )\in \Sigma _{\mathcal{T}}$. The r.h.s of the equation $%
\left( \ref{Inho-Baxter-EQ}\right) $ is a Laurent polynomial in $\lambda $
as we have:%
\begin{equation}
\text{\textsc{a}}(\lambda )Q(\lambda /q)+\text{\textsc{a}}(1/\lambda
)Q(\lambda q)=\frac{\text{\textsc{x}}(\lambda )Q(q/\lambda )-\text{\textsc{x}%
}(1/\lambda )Q(\lambda q)}{\lambda ^{2}-1/\lambda ^{2}}
\end{equation}%
which is finite in the limits $\lambda \rightarrow \pm 1,$ $\lambda
\rightarrow \pm i$. So that the r.h.s. of $\left( \ref{Inho-Baxter-EQ}%
\right) $ is a polynomial of degree $p\mathsf{N}+2$ in $\Lambda $, as it is
invariant w.r.t. the transformations $\lambda \rightarrow -\lambda $ and $%
\lambda \rightarrow 1/\lambda $. Then, the assumption that $\tau (\lambda )$
is entire in $\lambda $ implies by the equation $\left( \ref{Inho-Baxter-EQ}%
\right) $ that $\tau (\lambda )$ is a polynomial in $\Lambda $ of the form $(%
\ref{set-tau})$ and that it satisfies the equations:%
\begin{equation}
\text{det}_{p}D_{\tau }(\zeta_{a}^{(0)})=0,\text{ }\forall a\in \{1,...,%
\mathsf{N}\},
\end{equation}%
thanks to $\left( \ref{Inho-Baxter-EQ}\right) $ and $\left( \ref{Q-condition}%
\right) $, so that we obtain by SoV characterization $\tau (\lambda )\in
\Sigma _{\mathcal{T}}$.

Let us now prove the reverse statement, i.e. we assume $\tau (\lambda )\in
\Sigma _{\mathcal{T}}$\ and we prove that there exists $Q(\lambda )$ of the
form $\left( \ref{Q-form}\right) $ with degree $\mathsf{N}_{Q}=\left(
p-1\right) \mathsf{N}$ satisfying $\left( \ref{Q-condition}\right) $ and $%
\left( \ref{Inho-Baxter-EQ}\right) $. Let us consider the system of
equations:%
\begin{equation}
D_{\tau }(\lambda )\left( 
\begin{array}{l}
X_{0}(\lambda ) \\ 
X_{1}(\lambda ) \\ 
\vdots \\ 
\vdots \\ 
X_{p-1}(\lambda )%
\end{array}%
\right) _{p\times 1}=\left[ \tau _{\infty }-(q^{-\mathsf{N}_{Q}}\text{%
\textsc{a}}_{\infty }+q^{\mathsf{N}_{Q}}\text{\textsc{a}}_{0})\right]
F(\lambda )\left( 
\begin{array}{l}
\Lambda _{0}^{2}-X^{2} \\ 
\Lambda _{1}^{2}-X^{2} \\ 
\vdots \\ 
\vdots \\ 
\Lambda _{p-1}^{2}-X^{2}%
\end{array}%
\right) _{p\times 1},  \label{Sys-eq-cyclic}
\end{equation}%
where we have used the notations:%
\begin{equation}
\Lambda _{i}=q^{2i}\lambda ^{2}+\frac{1}{q^{2i}\lambda ^{2}}.
\end{equation}%
From the condition $\tau (\lambda )\in \Sigma _{\mathcal{T}}$ and the
assumption of general values of the boundary-bulk parameters $\left( \ref%
{Most-general.case}\right) $, we know that det$_{p}D_{\tau }(\lambda )$ is a
non-zero polynomial, so defining:%
\begin{equation}
Z_{det_{p}D_{\tau }}=\left\{ \pm i^{a}q^{h+1/2},\pm \zeta _{n}^{\left(
h\right) }\text{ }\forall a\in \left\{ 0,1\right\} ,n\in \{1,...,2\mathsf{N}%
\},h\in \left\{ 0,...,p-1\right\} \right\} ,
\end{equation}%
we can solve the previous system of equations for any value of $\lambda \in 
\mathbb{C}\backslash Z_{det_{p}D_{\tau }}$ by the Cramer's rule:%
\begin{equation}
X_{i}(\lambda )=\frac{\tau _{\infty }-(q^{-\mathsf{N}_{Q}}\text{\textsc{a}}%
_{\infty }+q^{\mathsf{N}_{Q}}\text{\textsc{a}}_{0})}{\left( Z^{2}-4\right)
\prod_{k=0}^{p-1}\left[ \tau _{\infty }-(q^{k}\text{\textsc{a}}_{\infty
}+q^{-k}\text{\textsc{a}}_{0})\right] }\text{det}_{p}D_{\tau
}^{(i+1)}(\lambda ),  \label{Def-X_i}
\end{equation}%
where $D_{\tau }^{(i)}(\lambda )$ is the $p\times p$ matrix obtained
replacing the column $i$ by the column at the r.h.s. of $\left( \ref%
{Sys-eq-cyclic}\right) $. Let us now rewrite the system of equation $\left( %
\ref{Sys-eq-cyclic}\right) $ bringing the first element in the last one for
the two column vectors:%
\begin{equation}
\tilde{D}_{\tau }(\lambda )\left( 
\begin{array}{l}
X_{1}(\lambda ) \\ 
X_{2}(\lambda ) \\ 
\vdots \\ 
X_{p-1}(\lambda ) \\ 
X_{0}(\lambda )%
\end{array}%
\right) _{p\times 1}=\left[ \tau _{\infty }-(q^{-\mathsf{N}_{Q}}\text{%
\textsc{a}}_{\infty }+q^{\mathsf{N}_{Q}}\text{\textsc{a}}_{0})\right]
F(\lambda )\left( 
\begin{array}{l}
\Lambda _{1}^{2}-X^{2} \\ 
\Lambda _{2}^{2}-X^{2} \\ 
\vdots \\ 
\Lambda _{p-1}^{2}-X^{2} \\ 
\Lambda _{0}^{2}-X^{2}%
\end{array}%
\right) _{p\times 1},
\end{equation}%
where it is easy to see that $\tilde{D}_{\tau }(\lambda )=D_{\tau }(\lambda
q)$. Rescaling now the argument of the functions, we can rewrite it as it
follows:%
\begin{equation}
D_{\tau }(\lambda )\left( 
\begin{array}{l}
X_{1}(\lambda /q) \\ 
X_{2}(\lambda /q) \\ 
\vdots \\ 
X_{p-1}(\lambda /q) \\ 
X_{0}(\lambda /q)%
\end{array}%
\right) _{p\times 1}=\left[ \tau _{\infty }-(q^{-\mathsf{N}_{Q}}\text{%
\textsc{a}}_{\infty }+q^{\mathsf{N}_{Q}}\text{\textsc{a}}_{0})\right]
F(\lambda )\left( 
\begin{array}{l}
\Lambda _{0}^{2}-X^{2} \\ 
\Lambda _{1}^{2}-X^{2} \\ 
\vdots \\ 
\Lambda _{p-2}^{2}-X^{2} \\ 
\Lambda _{p-1}^{2}-X^{2}%
\end{array}%
\right) _{p\times 1},
\end{equation}%
so that it must hold:%
\begin{equation}
X_{i+1}(\lambda /q)=X_{i}(\lambda )\text{ \ }\forall \lambda \in \mathbb{C}%
\backslash Z_{det_{p}D_{\tau }}\text{, }i\in \left\{ 0,...,p-1\right\}
\end{equation}%
where we have used the notation $X_{p}(\lambda )\equiv X_{0}(\lambda )$, or
equivalently:%
\begin{equation}
X_{a}(\lambda )=X_{0}(\lambda q^{a})\text{ \ }\forall \lambda \in \mathbb{C}%
\backslash Z_{det_{p}D_{\tau }}\text{, }a\in \left\{ 1,...,p-1\right\} .
\end{equation}%
Let us observe now that, from their definition, $X_{a}(\lambda )$ are continuous
functions of $\lambda $ so the above equation must be indeed satisfied for
any value of $\lambda \in \mathbb{C}$. Moreover, from the identity:%
\begin{equation}
\text{det}_{p}D_{\tau }^{(1)}(\lambda )=\text{det}_{p}D_{\tau
}^{(1)}(1/\lambda ),
\end{equation}%
which we can prove by some simple exchange of rows and columns, and from the
fact that: 
\begin{equation}
\forall i\in \{0...p-1\},\ \lambda \rightarrow 1/\lambda \ \Rightarrow \
\Lambda _{i}\rightarrow \Lambda _{p-i}
\end{equation}%
we get the symmetry:%
\begin{equation}
X_{0}(\lambda )=X_{0}(1/\lambda ),
\end{equation}%
which together with the symmetry $X_{0}(\lambda )=X_{0}(-\lambda )$ implies
that $X_{0}(\lambda )$ is a function of $\Lambda $.

By using this last result we can rewrite the first equation of the system $%
\left( \ref{Sys-eq-cyclic}\right) $ as it follows $\forall \lambda \in \mathbb{C}\ $:%
\begin{equation}
\tau (\lambda )X_{0}(\lambda )-\text{%
\textsc{a}}(\lambda )X_{0}(\lambda /q)-\text{\textsc{a}}(1/\lambda
)X_{0}(\lambda q)=\left[ \tau _{\infty }-(q^{-\mathsf{N}_{Q}}\text{\textsc{a}%
}_{\infty }+q^{\mathsf{N}_{Q}}\text{\textsc{a}}_{0})\right] \left( \Lambda
^{2}-X^{2}\right) F(\lambda ).
\end{equation}

Let us now prove that det$_{p}D_{\tau }^{(1)}(\lambda )$ is indeed a
polynomial of degree $(p-1)\mathsf{N}+2p$ in $\Lambda $. Note that in the
following when we refer to a row $k\in \mathbb{Z}$ what we mean\ is the row $%
k^{\prime }\in \left\{ 1,...,p\right\} $ with $k^{\prime }=k$ mod$\,p$. In
the row $\bar{h}=(p+1)/2+h$ of $D_{\tau }^{(1)}(\pm i^{a}q^{1/2-h}\lambda )$
at least one of the three non-zero elements is diverging under the limit $%
\lambda \rightarrow \pm 1,\pm i$. We can proceed as done in the previous
theorem, we define the matrix $\bar{D}_{\tau ,h}^{(1)}(\lambda )$ as the
matrix with all the rows coinciding with those of $D_{\tau }^{(1)}(\lambda )$
except the row $h+1$, which is obtained by summing the row $h$ and $h+1$ of $%
D_{\tau }^{(1)}(\lambda )$ and dividing them by $((i^{a}q^{h-1/2}\lambda
)^{2}-1/(i^{a}q^{h-1/2}\lambda )^{2})$, and the row $\bar{h}$, obtained
multiplying the row $\bar{h}$ of $D_{\tau }^{(1)}(\lambda )$ by $%
((i^{a}q^{h-1/2}\lambda )^{2}-1/(i^{a}q^{h-1/2}\lambda )^{2})$. Clearly we have:%
\begin{equation}
\text{det}_{p}\bar{D}_{\tau ,h}^{(1)}(\lambda )=\text{det}_{p}D{}_{\tau
}^{(1)}(\lambda ),
\end{equation}%
and the interesting point is that now all the rows of the matrix $D_{\tau
,h}^{(1)}(\pm i^{a}q^{1/2-h}\lambda )$ are finite in the limits $\lambda
\rightarrow \pm 1,\pm i$. We have that the nonzero elements of the rows $h$, 
$h+1$ and $\bar{h}$ of $\bar{D}_{\tau ,h}^{(1)}(\pm i^{a}q^{1/2-h})$ reads:%
\begin{eqnarray}
&&\left[ \bar{D}_{\tau ,h}^{(1)}(\pm i^{a}q^{1/2-h})\right] _{h+1,h-1}\left.
=\right. -s_{a,\pm }(1-\delta _{h-1,1})+\delta _{h-1,1}\omega _{a}, \\
&&\left[ \bar{D}_{\tau ,h}^{(1)}(\pm i^{a}q^{1/2-h})\right] _{h+1,h}\left.
=\right. -r_{a,\pm }(1-\delta _{h,1})+\delta _{h,1}\omega _{a}, \\
&&\left[ \bar{D}_{\tau ,h}^{(1)}(\pm i^{a}q^{1/2-h})\right] _{h+1,h+1}\left.
=\right. r_{a,\pm }(1-\delta _{h+1,1})+\delta _{h+1,1}\omega _{a}, \\
&&\left[ \bar{D}_{\tau ,h}^{(1)}(\pm i^{a}q^{1/2-h})\right] _{h+1,h+2}\left.
=\right. s_{a,\pm }(1-\delta _{h+2,1})+\delta _{h+2,1}\omega _{a}, \\
&&\left[ \bar{D}_{\tau ,h}^{(1)}(\pm i^{a}q^{1/2-h})\right] _{\bar{h},\bar{h}%
-1}\left. =\right. -\text{\textsc{x}}(\pm i^{a})(-1)^{a}(1-\delta _{\bar{h}%
-1,1}),\text{ } \\
&&\left[ \bar{D}_{\tau ,h}^{(1)}(\pm i^{a}q^{1/2-h})\right] _{\bar{h},\bar{h}%
+1}\left. =\right. \text{\textsc{x}}(\pm i^{a})(-1)^{a}(1-\delta _{\bar{h}%
+1,1}), \\
&&\left[ \bar{D}_{\tau ,h}^{(1)}(\pm i^{a}q^{1/2-h})\right] _{h,h-1}\left.
=\right. \tau (\pm i^{a}q^{1/2})(1-\delta _{h+1,1}), \\
&&\left[ \bar{D}_{\tau ,h}^{(1)}(\pm i^{a}q^{1/2-h})\right] _{h,h+1}\left.
=\right. -\tau (\pm i^{a}q^{1/2})(1-\delta _{h+1,1}),
\end{eqnarray}%
where we have defined:%
\begin{equation}
\omega _{a}=(-1)^{a}\left( q^{2}-1/q^{2}\right) .
\end{equation}%
The remaining rows of $\bar{D}_{\tau ,h}^{(1)}(\pm i^{a}q^{1/2-h})$ produce
the tridiagonal part of this matrix. It is possible to prove than that for
any $h\in \left\{ 0,...,p-1\right\} $ the matrix $\bar{D}_{\tau
,h}^{(1)}(\pm i^{a}q^{1/2-h})$ has linear dependent rows; so that 
det$_{p}D_{\tau }^{(1)}(\pm i^{a}q^{1/2-h})=0$ and the following
factorization holds:%
\begin{equation}
\text{det}_{p}D_{\tau }^{(1)}(\lambda )=\left( \lambda ^{2p}-\frac{1}{%
\lambda ^{2p}}\right) P_{\tau }\left( \lambda \right) .
\end{equation}%
Here $P_{\tau }\left( \lambda \right) $ is a Laurent polynomial of degree $%
2(p-1)\mathsf{N}+2p$ in $\lambda $, with the following odd parity:%
\begin{equation}
P_{\tau }(1/\lambda )=-P_{\tau }\left( \lambda \right) ,
\end{equation}%
being det$_{p}D_{\tau }^{(1)}(\lambda )$ a polynomial of degree $(p-1)%
\mathsf{N}+2p$ in $\Lambda $. Here, we want to prove that in fact:%
\begin{equation}
\text{det}_{p}D_{\tau }^{(1)}(\lambda )=\left( \lambda ^{2p}-\frac{1}{%
\lambda ^{2p}}\right) ^{2}\bar{Q}_{\tau }\left( \lambda \right) ,
\label{Full-zero structure}
\end{equation}%
where $\bar{Q}_{\tau }\left( \lambda \right) $ is a polynomial of degree $%
(p-1)\mathsf{N}$ in $\Lambda $. In order to do so we write down the equation:%
\begin{align}
\tau (\lambda )R_{\tau }(\lambda )& =\text{\textsc{a}}(\lambda )R_{\tau
}(\lambda /q)+\text{\textsc{a}}(1/\lambda )R_{\tau }(\lambda q)  \notag \\
& +\left( Z^{2}-4\right) \left( \Lambda ^{2}-X^{2}\right) \prod_{k=0}^{p-1} 
\left[ \tau _{\infty }-(q^{k}\text{\textsc{a}}_{\infty }+q^{-k}\text{\textsc{%
a}}_{0})\right] F(\lambda ),
\end{align}%
where for convenience we have denoted $R_{\tau }(\lambda )=$det$_{p}D_{\tau
}^{(1)}(\lambda )$, and we recall $Z=\lambda^{2p}+\frac{1}{\lambda^{2p}}$. The above equation is a direct consequence of the
equation satisfied by $X_{0}(\lambda )$ and of the definition of this last
function in terms of det$_{p}D_{\tau }^{(1)}(\lambda )$. Now let us consider
the following limit on the above equation $\lambda \rightarrow \pm i^{a}$
with $a\in \left\{ 0,1\right\} $:%
\begin{align}
\tau (\pm i^{a})R_{\tau }(\pm i^{a})& =\frac{1}{\pm 2i^{a}}\frac{d\text{%
\textsc{x}}}{d\lambda }(\pm i^{a})\left( R_{\tau }(\pm i^{a}/q)-R_{\tau
}(\pm i^{a}q)\right)  \notag \\
& +\text{\textsc{x}}(\pm i^{a})\lim_{\lambda \rightarrow \pm i^{a}}\left[ 
\frac{R_{\tau }(\lambda /q)}{\lambda ^{2}-1/\lambda ^{2}}-\frac{R_{\tau
}(\lambda q)}{\lambda ^{2}-1/\lambda ^{2}}\right] ,
\end{align}%
now by using the known identities:%
\begin{eqnarray}
R_{\tau }(\pm i^{a}) &=&R_{\tau }(\pm i^{a}/q)=R_{\tau }(\pm i^{a}q)=0, \\
\frac{R_{\tau }(\lambda /q)}{\lambda ^{2}-1/\lambda ^{2}} &=&\frac{R_{\tau
}(q/\lambda )}{\lambda ^{2}-1/\lambda ^{2}},
\end{eqnarray}%
we get:
\begin{equation}
\lim_{\lambda \rightarrow \pm i^{a}}\frac{R_{\tau }(\lambda /q)}{\lambda
^{2}-1/\lambda ^{2}}=-\lim_{\lambda \rightarrow \pm i^{a}}\frac{R_{\tau
}(\lambda q)}{\lambda ^{2}-1/\lambda ^{2}},  \label{Zero-deriv+-q}
\end{equation}
and so being \textsc{x}$(\pm i^{a})\neq 0$
\begin{equation}
\lim_{\lambda \rightarrow \pm i^{a}}\frac{R_{\tau }(\lambda /q)}{\lambda
^{2}-1/\lambda ^{2}}=0.  
\end{equation} 
These results imply the identities:%
\begin{equation}
P_{\tau }(\pm i^{a}/q)=-P_{\tau }(\pm i^{a}q)=0.  \label{Zero-deriv+-q+}
\end{equation}%
We can now write the functional equation for $P_{\tau }(\lambda )$:%
\begin{align}
\tau (\lambda )P_{\tau }(\lambda )& =\text{\textsc{a}}(\lambda )P_{\tau
}(\lambda /q)+\text{\textsc{a}}(1/\lambda )P_{\tau }(\lambda q)  \notag \\
& +\left( \lambda ^{2p}-\frac{1}{\lambda ^{2p}}\right) \left( \Lambda
^{2}-X^{2}\right) \prod_{k=0}^{p-1}\left[ \tau _{\infty }-(q^{k}\text{%
\textsc{a}}_{\infty }+q^{-k}\text{\textsc{a}}_{0})\right] F(\lambda ).
\end{align}%
Taking the limit $\lambda \rightarrow \pm i^{a}$ with $a\in \left\{
0,1\right\} $, we obtain:%
\begin{align}
\tau (\pm i^{a})P_{\tau }(\pm i^{a})& =\frac{1}{\pm 2i^{a}}\frac{d\text{%
\textsc{x}}}{d\lambda }(\pm i^{a})\left( P_{\tau }(\pm i^{a}/q)-P_{\tau
}(\pm i^{a}q)\right)  \notag \\
& +\text{\textsc{x}}(\pm i^{a})\lim_{\lambda \rightarrow \pm i^{a}}\left[ 
\frac{P_{\tau }(\lambda /q)}{\lambda ^{2}-1/\lambda ^{2}}-\frac{P_{\tau
}(\lambda q)}{\lambda ^{2}-1/\lambda ^{2}}\right] ,
\end{align}%
so that using the previous result $\left( \ref{Zero-deriv+-q+}\right) $ and
the identity:%
\begin{equation}
\lim_{\lambda \rightarrow \pm i^{a}}\frac{P_{\tau }(\lambda /q)}{\lambda
^{2}-1/\lambda ^{2}}=\lim_{\lambda \rightarrow \pm i^{a}}\frac{P_{\tau
}(\lambda q)}{\lambda ^{2}-1/\lambda ^{2}}
\end{equation}%
we obtain%
\begin{equation}
P_{\tau }(\pm i^{a})=0,
\end{equation}%
being $\tau (\pm i^{a})\neq 0$. Let us now compute the functional equation
for $P_{\tau }(\lambda )$ in the points $\lambda =\pm i^{a}q^{\epsilon }$
for $a\in \left\{ 0,1\right\} $, $\epsilon \in \left\{ -1,1\right\} $, we
obtain:%
\begin{eqnarray}
\tau (\pm i^{a}q)P_{\tau }(\pm i^{a}q) &=&\text{\textsc{a}}(\pm
i^{a}q)P_{\tau }(\pm i^{a})+\text{\textsc{a}}(\pm i^{a}/q)P_{\tau }(\pm
i^{a}q^{2}), \\
\tau (\pm i^{a}/q)P_{\tau }(\pm i^{a}/q) &=&\text{\textsc{a}}(\pm
i^{a}/q)P_{\tau }(\pm i^{a}/q^{2})+\text{\textsc{a}}(\pm i^{a}q)P_{\tau
}(\pm i^{a}),
\end{eqnarray}%
implying:%
\begin{equation}
P_{\tau }(\pm i^{a}/q^{2})=-P_{\tau }(\pm i^{a}q^{2})=0,
\end{equation}%
being \textsc{a}$(\pm i^{a}q^{\epsilon })\neq 0$ for $a,\epsilon \in \left\{
0,1\right\} $. We can iterate these computations for $\lambda =\pm
i^{a}q^{b\epsilon }$ for any $a\in \left\{ 0,1\right\} $, $\epsilon \in
\left\{ -1,1\right\} $ and $b\in \left\{ 2,...,\left( p-3\right) /2\right\} $
obtaining that:%
\begin{equation}
P_{\tau }(\pm i^{a}/q^{2b})=-P_{\tau }(\pm i^{a}q^{2b})=0,\text{ for any\ }%
b\in \left\{ 1,...,\left( p-3\right) /2\right\} .
\end{equation}%
In the cases $\lambda =\pm i^{a}q^{\pm 1/2}$ as \textsc{a}$(\pm
i^{a}/q^{1/2})=0$ the functional equation for $P_{\tau }(\lambda )$ give us:%
\begin{equation}
\tau (\pm i^{a}q^{\pm 1/2})P_{\tau }(\pm i^{a}q^{\pm 1/2})=\text{\textsc{a}}%
(\pm i^{a}q^{1/2})P_{\tau }(\pm i^{a}q^{\mp 1/2}),
\end{equation}%
which being $P_{\tau }(\pm i^{a}q^{1/2})=-P_{\tau }(\pm i^{a}q^{-1/2})$ and $%
\tau (\pm i^{a}q^{\pm 1/2})=$\textsc{a}$(\pm i^{a}q^{1/2})\neq 0$ implies
the identity:%
\begin{equation}
P_{\tau }(\pm i^{a}q^{1/2})=-P_{\tau }(\pm i^{a}q^{-1/2})=0,
\label{Final-zero}
\end{equation}%
so that the factorization $\left( \ref{Full-zero structure}\right) $ is
proven and we get that:%
\begin{equation}
X_{0}(\lambda )=\frac{\tau _{\infty }-(q^{-\mathsf{N}_{Q}}\text{\textsc{a}}%
_{\infty }+q^{\mathsf{N}_{Q}}\text{\textsc{a}}_{0})}{\prod_{k=0}^{p-1}\left[
\tau _{\infty }-(q^{k}\text{\textsc{a}}_{\infty }+q^{-k}\text{\textsc{a}}%
_{0})\right] }\bar{Q}_{\tau }(\lambda ),
\end{equation}%
is a polynomial of degree $\mathsf{N}_{Q}=\left( p-1\right) \mathsf{N}$ in $%
\Lambda $ which has the form $\left( \ref{Q-form}\right) $. This follows by
taking the asymptotic of its functional equation so that we can fix:%
\begin{equation}
Q(\lambda )\equiv X_{0}(\lambda ),
\end{equation}%
hence giving a constructive proof of the existence of the polynomial $Q$-function
solution of the equation $\left( \ref{Inho-Baxter-EQ}\right) $. The fact
that it is unique is shown observing that if $\hat{Q}(\lambda )$ is another polynomial
solution then:%
\begin{equation}
D_{\tau }(\lambda )\left( 
\begin{array}{l}
Q(\lambda )-\hat{Q}(\lambda ) \\ 
Q(\lambda q)-\hat{Q}(\lambda q) \\ 
\vdots \\ 
\vdots \\ 
Q(\lambda q^{p-1})-\hat{Q}(\lambda q^{p-1})%
\end{array}%
\right) _{p\times 1}=\left( 
\begin{array}{l}
0 \\ 
0 \\ 
\vdots \\ 
\vdots \\ 
0%
\end{array}%
\right) _{p\times 1},
\end{equation}%
from which it follows $Q(\lambda )\equiv \hat{Q}(\lambda )$ as $D_{\tau
}(\lambda )$ is invertible for any $\lambda \in \mathbb{C}\backslash
Z_{det_{p}D_{\tau }}$.

Finally, let us show that $Q(\lambda )$ satisfies the condition $\left( \ref%
{Q-condition}\right) $. By the definition $\left( \ref{Def-X_i}\right) $, $%
Q(\lambda )$ is a continuous function of the boundary-bulk parameters, then
it is enough to prove this statement for some value of these parameters to
show that it holds for almost all the values of these parameters.

Let us impose the condition $\left( \ref{Condi-bulk-quasi-L}\right) $, where
the ratio $\beta /\alpha $ is fixed by $\left( \ref{guage-alpha}\right) $,
then the following identities are satisfied:%
\begin{equation}
\text{\textsc{a}}(\zeta _{a}^{(0)})=0\text{ \ }\forall a\in \{1,...,\mathsf{N%
}\},
\end{equation}%
and the SoV characterization of the transfer matrix spectrum holds for any
value of the boundary-bulk parameters satisfying the inequalities $\left( %
\ref{Special-B-simple}\right) $-$\left( \ref{Special-B-simple2}\right) $. So in particular if we impose:%
\begin{equation}
\mu _{n_{k},-}=1/(q^{1+k}\mu _{a,+})\text{ \ \ }\forall k\in \{1,...,p-1\}%
\text{,}
\end{equation}%
for some $n_{k}\in \{1,...,\mathsf{N}\}\backslash \{a\}$ once we have chosen any $a\in \{1,...,\mathsf{N}\}$. Under these conditions it holds:%
\begin{equation}
\text{\textsc{a}}(\zeta _{a}^{(k)})=0,\text{\ }\forall k\in \{1,...,p-1\},
\end{equation}%
and the SoV representation implies the following centrality condition:%
\begin{equation}
\prod_{k=0}^{p-1}\mathcal{T}(\zeta _{a}^{(k)})=\prod_{k=0}^{p-1}\text{%
\textsc{a}}(1/\zeta _{a}^{(k)}),
\end{equation}%
from which in particular follows:%
\begin{equation}
\prod_{k=0}^{p-1}\tau (\zeta _{a}^{(k)})=\prod_{k=0}^{p-1}\text{\textsc{a}}%
(1/\zeta _{a}^{(k)}).
\end{equation}%
Let us remark now that the r.h.s and the l.h.s of the above equation are
continuous w.r.t. the boundary-bulk parameters so that the above identity
holds also if we take the special limit $\mu _{a,-}\rightarrow q^{1-p}/\mu
_{a,+}$ for which it holds \textsc{a}$(1/\zeta _{a}^{(p-1)})=0$ and so we
get:%
\begin{equation}
\exists !\bar{h}\in \{0,...,p-1\}:\tau (\zeta _{a}^{(\bar{h})})=0.
\end{equation}%
By definition of the function $Q(\lambda )$ under these conditions and limit
on the bulk parameters we get:%
\begin{equation}
Q(\zeta _{a}^{(\bar{h})})\propto \text{det}_{p}%
\begin{pmatrix}
W_{a,\bar{h}} & -\text{\textsc{a}}{}(1/\zeta _{a}^{(\bar{h})}) & 0 & \cdots
& 0 & 0 \\ 
W_{a,\bar{h}+1} & \tau (\zeta _{a}^{(\bar{h}+1)}) & -\text{\textsc{a}}%
(1/\zeta _{a}^{(\bar{h}+1)}) & 0 & \cdots & 0 \\ 
W_{a,\bar{h}+2} & {\quad }0 & \tau (\zeta _{a}^{(\bar{h}+2)}) & -\text{%
\textsc{a}}(1/\zeta _{a}^{(\bar{h}+1)}) &  & \vdots \\ 
\vdots &  & \cdots &  &  & \vdots \\ 
W_{a,p-1} & 0 & \cdots 0 & \tau (\zeta _{a}^{(p-1)}) & 0\cdots & \vdots \\ 
\vdots &  &  &  & \ddots {\qquad } & 0 \\ 
W_{a,\bar{h}} & \ldots & 0 & 0 & \tau (\zeta _{a}^{(\bar{h}-2)}) & -\text{%
\textsc{a}}(1/\zeta _{a}^{(\bar{h}-2)}) \\ 
W_{a,\bar{h}} & 0 & \ldots & 0 & 0 & \tau (\zeta _{a}^{(\bar{h}-1)})%
\end{pmatrix}%
,  \label{For-1}
\end{equation}%
where we have defined:%
\begin{equation}
W_{a,k}=\left( (\zeta _{a}^{(k)})^{2}+1/(\zeta _{a}^{(k)})^{2}\right)
^{2}-X^{2}.
\end{equation}%
Now replacing the first row $R_{1}$ with the following linear combination of
rows:%
\begin{equation}
\bar{R}_{1}=R_{1}+\sum_{i=0}^{p-2-\bar{h}}\prod_{j=0}^{i}\frac{\text{\textsc{%
a}}{}(1/\zeta _{a}^{(\bar{h}+j)})}{\tau (\zeta _{a}^{(\bar{h}+j+1)})}R_{2+i},
\end{equation}%
we get%
\begin{equation}
\bar{R}_{1}=%
\begin{pmatrix}
\bar{W}_{a,\bar{h}} & 0 & \cdots & 0 & 0%
\end{pmatrix}%
_{1\times p}
\end{equation}%
where:%
\begin{equation}
\bar{W}_{a,\bar{h}}=W_{a,\bar{h}}+\sum_{i=0}^{p-2-\bar{h}}\prod_{j=0}^{i}%
\frac{\text{\textsc{a}}{}(1/\zeta _{a}^{(\bar{h}+j)})}{\tau (\zeta _{a}^{(%
\bar{h}+j+1)})}W_{a,\bar{h}+1+i}
\end{equation}%
and so: 
\begin{equation}
Q(\zeta _{a}^{(\bar{h})})=\bar{W}_{a,\bar{h}}\prod_{k\neq \bar{h}%
,k=0}^{p-1}\tau (\zeta _{a}^{(k)})\neq 0,
\end{equation}%
for generic values of the boundary-bulk parameters. Indeed, as the $W_{a,%
\bar{h}+1+i}$ are functions only of the bulk parameter $\mu _{a,+}$ while
the ratios \textsc{a}${}(1/\zeta _{a}^{(\bar{h}+j)})/\tau (\zeta _{a}^{(\bar{%
h}+j+1)})$ are functions of both the boundary and the bulk parameters then
we can prove that $\bar{W}_{a,\bar{h}}\neq 0$. Explicitly we can compute the
asymptotic of $\bar{W}_{a,\bar{h}}$ in the limit $\mu _{a,+}\rightarrow
\infty $, by using the know asymptotic of the transfer matrix, therefore showing
that it is non-zero for general values of boundary-bulk parameters.
\end{proof}

In the previous theorem we have excluded the boundary-bulk one-constraint
cases leading to an identically zero det$D_{\tau }(\lambda )$ for any $\tau
(\lambda )\in \Sigma _{\mathcal{T}}$, these specific cases are considered in
the next theorem.

\begin{theorem}
Let us assume that there exists $k \in \{0,...,p-1\}$ such that it holds:%
\begin{equation}
\tau _{\infty }=q^{-k}\text{\textsc{a}}_{\infty }+q^{k}\text{\textsc{a}}_{0},
\label{one-constraint-case}
\end{equation}%
then, for almost all the values of the boundary-bulk parameters, $\tau
(\lambda )\in \Sigma _{\mathcal{T}}$ (the set of the eigenvalues of $\mathcal{T}(\lambda)$) if and only if $\tau (\lambda )$ is
an entire function and there exists and is unique a polynomial $Q(\lambda )$
of the form $\left( \ref{Q-form}\right) $ with $\mathsf{N}_{Q}\leq \left(
p-1\right) (\mathsf{N}+1)$ and $\mathsf{N}_{Q}=k$ mod$\,p$, satisfying the
following homogeneous Baxter equation:%
\begin{equation}
\tau (\lambda )Q(\lambda )=\text{\textsc{a}}(\lambda )Q(\lambda /q)+\text{%
\textsc{a}}(1/\lambda )Q(\lambda q),  \label{ho-Baxter-EQ}
\end{equation}%
and the conditions:%
\begin{equation}
(Q(\zeta _{a}^{\left( 0\right) }),...,Q(\zeta _{a}^{\left( p-1\right)
}))\neq (0,...,0)\text{ \ \ }\forall a\in \{1,...,\mathsf{N}\}\text{.}
\label{Q-condition2}
\end{equation}
\end{theorem}

\begin{proof}
First let us assume that $\tau (\lambda )$ and $Q(\lambda )$ satisfies the
homogeneous Baxter equation with $\tau (\lambda )$ entire function and $%
Q(\lambda )$ polynomial of the form $\left( \ref{Q-form}\right) $ with $%
\mathsf{N}_{Q}\leq \left( p-1\right) (\mathsf{N}+1)$ and $\mathsf{N}_{Q}=k$
mod$\,p$, then from this same equation it follows that $\tau (\lambda )$ is
a polynomial of the form $(\ref{set-tau})$. Moreover, for any fixed $\lambda
\in \mathbb{C}$ we can construct the following homogeneous system of
equations:%
\begin{equation}
D_{\tau }(\lambda )\left( 
\begin{array}{l}
Q(\lambda ) \\ 
Q(\lambda q) \\ 
\vdots \\ 
\vdots \\ 
Q(\lambda q^{p-1})%
\end{array}%
\right) _{p\times 1}=\left( 
\begin{array}{l}
0 \\ 
0 \\ 
\vdots \\ 
\vdots \\ 
0%
\end{array}%
\right) _{p\times 1},
\end{equation}%
which is satisfied as a consequence of the Baxter equation. Finally, being $%
(Q(\lambda ),...,Q(\lambda q^{p-1}))$\ non-zero for any $\lambda \in \mathbb{%
C}$, up to at most a finite number of values, we get:%
\begin{equation}
\text{det}D_{\tau }(\lambda )=0\text{ \ \ }\forall \lambda \in \mathbb{C}
\end{equation}%
so that Proposition \ref{General F-Eq} implies $\tau (\lambda )\in \Sigma _{%
\mathcal{T}}$.

To prove the reverse statement we use the results of the Lemma \ref%
{Cofactor-prop} on the matrix $D_{\tau }(\lambda )$ and on its cofactors:%
\begin{equation}
\text{\textsc{C}}_{i,j}(\lambda )=(-1)^{i+j}\text{det}_{p-1}D_{\tau
,i,j}(\lambda )\text{.}  \label{cofactor-def}
\end{equation}%
We take now $\tau (\lambda )\in \Sigma _{\mathcal{T}}$ from which it holds:%
\begin{equation}
\text{det}D_{\tau }(\lambda )=0\text{ \ \ }\forall \lambda \in \mathbb{C},
\end{equation}%
and so by Lemma \ref{Cofactor-prop} it follows that rank$D_{\tau }(\lambda
)=p-1$ for any $\lambda \in \mathbb{C}\backslash K$, where $K$ is a finite
set of complex numbers if not empty. Then the matrix composed of the
cofactors of the matrix $D_{\tau }(\lambda )$ has rank $1$ for any $\lambda
\in \mathbb{C}\backslash K$. This just means the proportionality:%
\begin{equation}
\text{\textsc{V}}_{i}(\lambda )=\text{\textsc{a}}_{i,j}(\lambda )\text{%
\textsc{V}}_{j}(\lambda )\text{ }\forall \lambda \in \mathbb{C}\backslash
K,\forall i,j\in \{1,...,p\}
\end{equation}%
where we have defined:%
\begin{equation}
\text{\textsc{V}}_{i}(\lambda )\equiv (\text{\textsc{C}}_{i,1}(\lambda ),%
\text{\textsc{C}}_{i,2}(\lambda ),...,\text{\textsc{C}}_{i,p}(\lambda ))%
\text{ }\forall \lambda \in \mathbb{C}\backslash K,\forall i\in \{1,...,p\}
\end{equation}%
and \textsc{a}$_{i,j}(\lambda )$ are some functions such that: 
\begin{equation}
\text{\textsc{a}}_{i,j}(\lambda )\neq 0\text{ and finite for any }\lambda
\in \mathbb{C}\backslash \left\{ K\cup K_{0}\cup K_{i}\cup K_{j}\right\}
\end{equation}%
where $K_{0}$ is the set of the $p$-roots of unit and%
\begin{equation}
K_{a}\equiv \left\{ x\in \mathbb{C}:\text{\textsc{V}}_{a}(x)\equiv
(0,...,0)\right\} \text{ }\forall a\in \{1,...,p\}
\end{equation}%
such sets are finite if not empty, being the elements of the vectors $%
(\Lambda ^{p}-X^{p})$\textsc{V}$_{i}(\lambda )$ Laurent polynomials. The
above identities in particular imply: 
\begin{equation}
\text{\textsc{a}}_{1,2}(\lambda )\text{\textsc{C}}_{1,1}(\lambda )\text{%
\textsc{C}}_{2,2}(\lambda )=\text{\textsc{a}}_{1,2}(\lambda )\text{\textsc{C}%
}_{1,2}(\lambda )\text{\textsc{C}}_{2,1}(\lambda )\text{ }\forall \lambda
\in \mathbb{C}\backslash K  \label{proportionality}
\end{equation}%
so that for any $\lambda \in \mathbb{C}\backslash \left\{ K\cup K_{0}\cup
K_{i}\cup K_{j}\right\} $ it holds:%
\begin{equation}
\text{\textsc{C}}_{1,1}(\lambda )\text{\textsc{C}}_{2,2}(\lambda )=\text{%
\textsc{C}}_{1,2}(\lambda )\text{\textsc{C}}_{2,1}(\lambda ).
\label{Cofactor-connection}
\end{equation}%
Hence it holds for any $\lambda \in \mathbb{C}$ using continutiy
properties of the cofactors, being $\left\{ K\cup K_{0}\cup K_{i}\cup
K_{j}\right\} $ a finite set of values. Similarly, the fact that the
vectorial condition\ $D(\lambda )$\textsc{V}$_{1}(\lambda )= \ $\b{0} holds
true for any $\lambda \in \mathbb{C}\backslash K$ implies that it is indeed
satisfied for any $\lambda \in \mathbb{C}$. Here, we write explicitly the
first element of this vectorial condition:%
\begin{equation}
\tau (\lambda )\text{\textsc{C}}_{1,1}(\lambda )=\text{\textsc{a}}(\lambda )%
\text{\textsc{C}}_{1,p}(\lambda )+\text{\textsc{a}}(1/\lambda )\text{\textsc{%
C}}_{1,2}(\lambda ),  \label{Bax-eq}
\end{equation}%
together with the rewriting of (\ref{Cofactor-connection}) by using the
identity (\ref{Sym-1}):%
\begin{equation}
\text{\textsc{C}}_{1,1}(\lambda )\text{\textsc{C}}_{1,1}(\lambda q)=\text{%
\textsc{C}}_{1,2}(\lambda )\text{\textsc{C}}_{1,p}(q\lambda ).
\label{Inter-step}
\end{equation}%
Once we recall that \textsc{C}$_{1,1}(\lambda )$, \textsc{C}$_{1,2}(\lambda
) $ and \textsc{C}$_{1,p}(\lambda )$ are Laurent polynomial in $\lambda $
satisfying the factorizations $\left( \ref{F11}\right) $, $\left( \ref{F12}%
\right) $ and $\left( \ref{F1p}\right) $, respectively, it follows that the
above two equations holds as well as if written in terms of the functions $%
\widehat{\text{\textsc{C}}}_{1,1}(\lambda )$, $\widehat{\text{\textsc{C}}}%
_{1,2}(\lambda )$ and $\widehat{\text{\textsc{C}}}_{1,p}(\lambda )$.

Similarly to what has been done in the Lemma 5 of the
paper \cite{OpenCyN-10}, we can show that the two above equations for $%
\widehat{\text{\textsc{C}}}_{1,1}(\lambda )$, $\widehat{\text{\textsc{C}}}%
_{1,2}(\lambda )$ and $\widehat{\text{\textsc{C}}}_{1,p}(\lambda )$ and
their symmetry properties $\left( \ref{Sym-3}\right) $ imply that if $%
\widehat{\text{\textsc{C}}}_{1,1}(\lambda )$ has a common zero with $%
\widehat{\text{\textsc{C}}}_{1,2}(\lambda )$ then this is also a zero of $%
\widehat{\text{\textsc{C}}}_{1p}(\lambda )$ and also the inverse of such
zero is a common zero of these polynomials. Moreover, the same statement
holds exchanging $\widehat{\text{\textsc{C}}}_{1,2}(\lambda )$ with $%
\widehat{\text{\textsc{C}}}_{1,p}(\lambda )$. So we can denote with \textsc{c%
}$_{1,1}\overline{\text{\textsc{C}}}_{1,1}(\lambda )$, \textsc{c}$_{1,2}%
\overline{\text{\textsc{C}}}_{1,2}(\lambda )$ and \textsc{c}$_{1,p}\overline{%
\text{\textsc{C}}}_{1,p}(\lambda )$ the polynomials obtained simplifying the
common factors in $\widehat{\text{\textsc{C}}}_{1,1}(\lambda )$, $\widehat{%
\text{\textsc{C}}}_{1,2}(\lambda )$ and $\widehat{\text{\textsc{C}}}%
_{1,p}(\lambda )$. Then, they have to satisfy the relations: 
\begin{equation}
\overline{\text{\textsc{C}}}_{1,p}(\lambda )=y_{1,1}\text{$\overline{\text{%
\textsc{C}}}$}_{1,1}(\lambda q^{-1}),\text{ \ \ $\overline{\text{\textsc{C}}}
$}_{1,2}(\lambda )=y_{1,1}^{-1}\text{$\overline{\text{\textsc{C}}}$}%
_{1,1}(\lambda q)  \label{proportionality-cofactor}
\end{equation}%
and defined $x_{1,1}\equiv $\textsc{c}$_{1,1}/$\textsc{c}$_{1,2}=$\textsc{c}$%
_{1,p}/$\textsc{c}$_{1,1}$, we obtain the following Baxter equation in the
polynomial $\overline{\text{\textsc{C}}}_{1,1}(\lambda )$: 
\begin{equation}
t(\lambda )\text{$\overline{\text{\textsc{C}}}$}_{1,1}(\lambda )=\left(
x_{1,1}y_{1,1}\right) \text{\textsc{a}}(\lambda )\text{$\overline{\text{%
\textsc{C}}}$}_{1,1}(\lambda q^{-1})+\left( 1/(x_{1,1}y_{1,1})\right) \text{%
\textsc{a}}(1/\lambda )\text{$\overline{\text{\textsc{C}}}$}_{1,1}(\lambda
q),  \label{deform-BAX}
\end{equation}%
and computing the above equation in $\lambda =q^{1/2}$ we get:%
\begin{equation}
t(q^{1/2})\text{$\overline{\text{\textsc{C}}}$}_{1,1}(q^{1/2})=\left(
x_{1,1}y_{1,1}\right) \text{\textsc{a}}(q^{1/2})\text{$\overline{\text{%
\textsc{C}}}$}_{1,1}(q^{-1/2}),
\end{equation}%
from which it follows $x_{1,1}y_{1,1}=1$ once we recall that \textsc{C}$%
_{1,1}(q^{1/2})\neq 0$ and that $\overline{\text{\textsc{C}}}$$%
_{1,1}(\lambda )$ is even under $\lambda \rightarrow 1/\lambda $. So, we can
define: 
\begin{equation}
Q(\lambda )\equiv \text{$\overline{\text{\textsc{C}}}$}_{1,1}(\lambda ),
\end{equation}%
a polynomial in $\Lambda $ of maximal degree $\left( p-1\right) (\mathsf{N}%
+1)$, which satisfies the homogeneous Baxter equation as required.
\end{proof}

Let us introduce now the following states:%
\begin{eqnarray}
\left\langle \beta ,\omega \right\vert &=&\sum_{h_{1},...,h_{\mathsf{N}%
}=0}^{p-1}\prod_{a=1}^{\mathsf{N}}\prod_{k_{a}=0}^{h_{a}-1}\frac{\text{%
\textsc{a}}(1/\zeta _{a}^{(k_{a})})}{\text{\textsc{d}}(1/\zeta
_{a}^{(k_{a})})}\prod_{1\leq b<a\leq \mathsf{N}%
}(X_{a}^{(h_{a})}-X_{b}^{(h_{b})})\left\langle \beta ,h_{1},...,h_{\mathsf{N}%
}\right\vert , \\
|\beta ,\bar{\omega}\rangle &=&\sum_{h_{1},...,h_{\mathsf{N}%
}=0}^{p-1}\prod_{1\leq b<a\leq \mathsf{N}}(X_{a}^{(h_{a})}-X_{b}^{(h_{b})})|%
\beta ,h_{1},...,h_{\mathsf{N}}\rangle ,\label{Ref-Sta-R}
\end{eqnarray}%
(see \eqref{45}, \eqref{412} and \eqref{419}) and the following renormalization of the $\mathcal{B}_{-}$-operator family%
\begin{equation}
\mathcal{\hat{B}}_{-}(\lambda |\beta )=\frac{\mathcal{B}_{-}(\lambda |\beta
)T_{\beta }^{2}}{(\lambda ^{2}/q-q/\lambda ^{2})\text{\textsc{b}}_{-}(\beta )%
},
\end{equation}%
which is a degree $\mathsf{N}$ polynomial in $\Lambda=\lambda^2+\frac{1}{\lambda^2}$, and where $T_\beta$ is simply a shift on the gauge parameter $\beta$ (see \eqref{betashifteuh}). As first remarked
in the papers \cite{DerKM03,OpenCyDerKM03-2}, from the polynomial
characterization of the $Q$-function and the SoV characterization it follows
the Bethe-like rewriting of the transfer matrix eigenstates stated in the
following\footnote{One should remark that the logic that lead us to the ABA rewriting of the transfer matrix eigenstates is completely different from the one underling the algebraic Bethe ansatz. We get it by rewriting the original SoV form and this allows us to identity the non-trivial state that takes a role similar to a reference state. Note however that it has properties rather different from an ABA reference state as in general it is not an eigenstate of the transfer matrix! For simpler models, for which such a reference state can be naturally guessed, one can also follow the ABA logic i.e. to make an ansatz on the form of the ABA states and then to compute the action of the transfer matrix on these states deriving the Bethe equations by putting to zero the so-called unwanted terms. This is what it has been done in the paper \cite{Belliard2015c} for the quantum spin 1/2 chains.}:

\begin{corollary}
The left and right transfer matrix eigenstates associated to $\tau (\lambda
)\in \Sigma _{\mathcal{T}}$ admit the following Bethe ansatz like
representations:%
\begin{equation}
\left\langle \tau \right\vert =\left\langle \beta ,\omega \right\vert
\prod_{b=1}^{\mathsf{N}_{Q}}\mathcal{\hat{B}}_{-}(\lambda _{b}|\beta ),\text{
\ \ }|\tau \rangle =\prod_{b=1}^{\mathsf{N}_{Q}}\mathcal{\hat{B}}%
_{-}(\lambda _{b}|\beta )|\beta ,\bar{\omega}\rangle ,
\label{Bethe-like-eigenstates}
\end{equation}%
where the $\lambda _{b}$ (fixed up the symmetry $\lambda _{b}\rightarrow
-\lambda _{b},$ $\lambda _{b}\rightarrow 1/\lambda _{b}$) for $b\in \{1,...,%
\mathsf{N}_{Q}\}$ are the zeros of $Q(\lambda )$ and we have imposed the
condition $\left( \ref{guage-alpha}\right) $ on the gauge parameters.
\end{corollary}
\begin{proof}

These identities follow from the polynomiality of the $Q$-functions, which
implies the following identity:%
\begin{align}
\prod_{a=1}^{\mathsf{N}}q_{\tau ,a}^{(h_{a})}& =\prod_{a=1}^{\mathsf{N}%
}Q(\zeta _{a}^{(h)})=\prod_{a=1}^{\mathsf{N}}\prod_{b=1}^{\mathsf{N}%
_{Q}}((\zeta _{a}^{(h)})^{2}+1/(\zeta _{a}^{(h)})^{2}-\Lambda _{b})  \notag
\\
& =\left( -1\right) ^{\mathsf{N}\text{\thinspace }\mathsf{N}%
_{Q}}\prod_{a=1}^{\mathsf{N}}\prod_{b=1}^{\mathsf{N}_{Q}}(\Lambda
_{b}-((\zeta _{a}^{(h)})^{2}+1/(\zeta _{a}^{(h)})^{2}))=\left( -1\right) ^{%
\mathsf{N}\text{\thinspace }\mathsf{N}_{Q}}\prod_{b=1}^{\mathsf{N}_{Q}}\text{%
\textsc{\^{b}}}_{\text{\textbf{h}}}(\lambda _{b})
\end{align}%
where the \textsc{\^{b}}$_{\text{\textbf{h}}}(\lambda _{b})$ is the
eigenvalue of the operator $\mathcal{\hat{B}}_{-}(\lambda |\beta )$ and the 
\begin{equation}
\Lambda _{b}=\lambda _{b}^{2}+1/\lambda _{b}^{2}
\end{equation}%
are the zeros of the $Q$-function as defined in $\left( \ref{Q-form}\right) $%
. Now we have just to do the action of the monomial:%
\begin{equation}
\prod_{b=1}^{\mathsf{N}_{Q}}\mathcal{\hat{B}}_{-}(\lambda _{b}|\beta )
\end{equation}%
on the right state $\left( \ref{Ref-Sta-R}\right) $ and use that by
definition:%
\begin{equation}
\prod_{b=1}^{\mathsf{N}_{Q}}\mathcal{\hat{B}}_{-}(\lambda _{b}|\beta )|\beta
,h_{1},...,h_{\mathsf{N}}\rangle =|\beta ,h_{1},...,h_{\mathsf{N}}\rangle
\prod_{b=1}^{\mathsf{N}_{Q}}\text{\textsc{\^{b}}}_{\text{\textbf{h}}%
}(\lambda _{b})
\end{equation}%
to prove that the vector in $\left( \ref{Bethe-like-eigenstates}\right) $
coincides, up to the sign, with the vector $\left( \ref{OpenCyeigenT-r-D}%
\right) $ and so it is the corresponding transfer matrix eigenvector;
similarly one shows that the covector in  $\left( \ref%
{Bethe-like-eigenstates}\right) $ coincides with the covector $\left( \ref%
{OpenCyeigenT-l-D}\right) $.

\end{proof}

\section{Conclusions}

In this second article we have shown how to implement the SoV method to characterize the transfer matrix spectrum for integrable models associated to the Bazhanov-Stroganov quantum Lax operator and to the most general integrable  boundary conditions. For that purpose it was necessary to perform a gauge transformation so as to recast the problem in a form similar to the one studied in our first article, i.e., such that one of the boundary $K$-matrices becomes triangular after the gauge transformation. Let us stress that the separate basis was designed again as the (pseudo)-eigenvector basis of some gauged operator of the reflection algebra having simple spectrum. What remains to be done is the construction of integrable local cyclic Hamiltonian having appropriate boundary conditions and commuting with the boundary transfer matrices considered here. This amounts to use trace identities involving the fundamental $R$-matrix acting in the tensor product of two cyclic representations \cite{OpenCyBS90,Tarasov-1992,Tarasov-1993} and to construct the associated $K$-matrices, hence also acting in these cyclic representations. The reflection equations will have to be written for arbitrary choices (and mixing) of the spin-1/2 and cyclic representations. Correspondingly, there will be compatibility conditions between the different $K$-matrices acting in these two different representations. We will address this question in a  forthcoming article \cite{MNP2018b}.

\section*{Acknowledgements}
J. M. M. and G. N. are supported by CNRS and ENS de Lyon;  B. P. is supported by ENS de Lyon and ENS Cachan.



\appendix


\section*{Appendices}

\section{Gauge transformed Yang-Baxter algebra}

\subsection{Gauge transformed Yang-Baxter generators}

For arbitrary complex parameters $\alpha $ and $\beta $ let us introduce the
following two matrices:%
\begin{equation}
G(\lambda |\alpha ,\beta )\equiv \left( 
\begin{array}{cc}
1/(\alpha \beta \lambda ) & \beta /(\alpha \lambda ) \\ 
1 & 1%
\end{array}%
\right) ,\text{ \ }\bar{G}(\lambda |\gamma )\equiv \left( 
\begin{array}{cc}
1/(\lambda \gamma ) & 0 \\ 
1 & 1%
\end{array}%
\right) ,
\end{equation}%
and their inverses:%
\begin{equation}
G^{-1}(\lambda |\alpha ,\beta )=\frac{\alpha \lambda }{\beta -1/\beta }%
\left( 
\begin{array}{cc}
-1 & \beta /(\alpha \lambda ) \\ 
1 & -1/(\alpha \beta \lambda )%
\end{array}%
\right) ,\text{ \ }\bar{G}^{-1}(\lambda |\gamma )=\left( 
\begin{array}{cc}
\lambda \gamma & 0 \\ 
-\lambda \gamma & 1%
\end{array}%
\right) .
\end{equation}%
Now we can construct the gauge transformed bulk monodromy matrix (see \eqref{24}):%
\begin{equation}
M(\lambda |\alpha ,\beta ,\gamma )=G^{-1}(\lambda q^{1/2}|\alpha ,\beta
)\,M(\lambda )\bar{G}(\lambda q^{1/2}|\gamma ) \\
=\left( 
\begin{array}{ll}
A(\lambda |\alpha ,\beta ,\gamma ) & B(\lambda |\alpha ,\beta ) \\ 
C(\lambda |\alpha ,\beta ,\gamma ) & D(\lambda |\alpha ,\beta )%
\end{array}%
\right) ,
\end{equation}%
and, in a similar way, we can define (see \eqref{Inverse-M}): 
\begin{equation}
\hat{M}(\lambda |\alpha ,\beta ,\gamma )=\bar{G}^{-1}(q^{1/2}/\lambda
|\gamma )\,\hat{M}(\lambda )G(q^{1/2}/\lambda |\alpha ,\beta ) \\
=\left( 
\begin{array}{ll}
\bar{A}(\lambda |\alpha ,\beta ,\gamma ) & \bar{B}(\lambda |\alpha ,\beta
,\gamma ) \\ 
\bar{C}(\lambda |\alpha ,\beta ,\gamma ) & \bar{D}(\lambda |\alpha ,\beta
,\gamma )%
\end{array}%
\right) .
\end{equation}%
The definition here chosen of these gauge transformations differ w.r.t. that
used previously in the literature on one hand for the particular choice of
the right transformation in $M(\lambda |\alpha ,\beta ,\gamma )$ and, on the
other hand, as the parameters on the left and the right transformation are a
priori independent. It is simple to prove by direct computations that:%
\begin{eqnarray}
\bar{D}(\lambda |\alpha ,\beta ,\gamma ) &=&f(\alpha ,\beta ,\gamma
)A(1/\lambda |\alpha ,\beta ,\gamma ),\text{ \ \ }\bar{B}(\lambda |\alpha
,\beta ,\gamma )=-f(\alpha ,\beta ,\gamma )B(1/\lambda |\alpha ,\beta ), \\
\bar{C}(\lambda |\alpha ,\beta ,\gamma ) &=&-f(\alpha ,\beta ,\gamma
)C(1/\lambda |\alpha ,\beta ,\gamma ),\text{ \ \ }\bar{A}(\lambda |\alpha
,\beta ,\gamma )=f(\alpha ,\beta ,\gamma )D(1/\lambda |\alpha ,\beta ),
\end{eqnarray}%
where:%
\begin{equation}
f(\alpha ,\beta ,\gamma )=\left( -1\right) ^{N}\frac{\gamma (1-\beta ^{2})}{%
\alpha \beta }.
\end{equation}%
Moreover, the identity:%
\begin{equation}
\text{det}_{q}M(\lambda )=\left( -1\right) ^{N}M(\lambda q^{1/2})\hat{M}%
(q^{1/2}/\lambda ),
\end{equation}%
and the corollaries:%
\begin{eqnarray}
\text{det}_{q}M(\lambda ) &=&\left( -1\right) ^{N}M(\lambda q^{1/2}|\alpha
,\beta ,\gamma )\hat{M}(q^{1/2}/\lambda |\alpha q,\beta ,\gamma q) \\
&=&\left( -1\right) ^{N}\hat{M}(q^{1/2}/\lambda |\alpha q,\beta ,\gamma
q)M(\lambda q^{1/2}|\alpha ,\beta ,\gamma ),
\end{eqnarray}%
imply the following two equivalent expressions of the quantum determinant by
the gauge transformed generators:%
\begin{align}
\text{det}_{q}M(\lambda )& = & & \frac{\gamma (1-q^{2}\beta ^{2})\left[
A(\lambda q^{1/2}|\alpha ,\beta ,\gamma )D(\lambda /q^{1/2}|\alpha ,\beta
q)-B(\lambda q^{1/2}|\alpha ,\beta )C(\lambda /q^{1/2}|\alpha ,\beta
q,\gamma q)\right] }{\alpha \beta } \\
& = & & \frac{\gamma (q^{2}-\beta ^{2})\left[ D(\lambda q^{1/2}|\alpha
,\beta )A(\lambda /q^{1/2}|\alpha ,\beta /q,\gamma q)-C(\lambda
q^{1/2}|\alpha ,\beta ,\gamma )B(\lambda /q^{1/2}|\alpha ,\beta /q)\right] }{%
\alpha \beta },
\end{align}%
plus other two equivalent rewriting. The gauge transformed Yang-Baxter
generators are of special interest as they define a closed set of
commutation relations:%
\begin{align}
A(\lambda |\alpha ,\beta ,\gamma )A(\mu |\alpha ,\beta /q,\gamma /q)& =A(\mu
|\alpha ,\beta ,\gamma )A(\lambda |\alpha ,\beta /q,\gamma /q)
\label{Comm-dyn-1} \\
B(\lambda |\alpha ,\beta )B(\mu |\alpha ,\beta /q)& =B(\mu |\alpha ,\beta
)B(\lambda |\alpha ,\beta /q) \\
C(\lambda |\alpha ,\beta ,\gamma )C(\mu |\alpha ,\beta q,\gamma /q)& =C(\mu
|\alpha ,\beta ,\gamma )C(\lambda |\alpha ,\beta q,\gamma /q) \\
D(\lambda |\alpha ,\beta )D(\mu |\alpha ,\beta /q)& =D(\mu |\alpha ,\beta
)D(\lambda |\alpha ,\beta /q) \\
A(\lambda |\alpha ,\beta ,\gamma )B(\mu |\alpha ,\beta /q)& =\frac{q(\lambda
/\mu -\mu /\lambda )}{\lambda q/\mu -\mu /q\lambda }B(\mu |\alpha ,\beta
)A(\lambda |\alpha ,\beta /q,\gamma q)  \notag \\
& +\frac{(q-1/q)\mu /\lambda }{\lambda q/\mu -\mu /q\lambda }A(\mu |\alpha
,\beta ,\gamma )B(\lambda |\alpha ,\beta /q)  \label{Comm-dyn-2} \\
B(\lambda |\alpha ,\beta )A(\mu |\alpha ,\beta /q,\gamma q)& =\frac{\lambda
/\mu -\mu /\lambda }{q(\lambda q/\mu -\mu /q\lambda )}A(\mu |\alpha ,\beta
,\gamma )B(\lambda |\alpha ,\beta /q)  \notag \\
& +\frac{(q-1/q)\lambda /\mu }{\lambda q/\mu -\mu /q\lambda }B(\mu |\alpha
,\beta )A(\lambda |\alpha ,\beta /q,\gamma q),  \label{Comm-dyn-3}
\end{align}%
We can prove these commutation relations by direct computations using the
properties of the gauge transformations and their action on the Yang-Baxter
equation.

\subsection{Pseudo-reference state for the gauge transformed Yang-Baxter
algebra}

In the following, we want to study the conditions for which a nonzero state
identically annihilated by the action of the operator family $A(\lambda
|\alpha ,\beta ,\gamma )$ exists:%
\begin{equation}
\langle \Omega ,\alpha ,\beta ,\gamma |A(\lambda |\alpha ,\beta ,\gamma )=0.
\end{equation}%
It is an easy consequence of the gauge transformed Yang-Baxter commutation
relations that under the condition that this state exists and is unique then
it is a pseudo-reference state for the gauge transformed Yang-Baxter
algebra, i.e. it holds:%
\begin{align}
\langle \Omega ,\alpha ,\beta ,\gamma |A(\lambda |\alpha ,\beta ,\gamma )&
=0,\text{ \ }\langle \Omega ,\alpha ,\beta ,\gamma |B(\lambda |\alpha ,\beta
)=b(\lambda |\alpha ,\beta )\langle \Omega ,\alpha ,\beta /q,\gamma q|,\text{
}  \label{Gauged reference-1} \\
\langle \Omega ,\alpha ,\beta /q,\gamma q|C(\lambda |\alpha ,\beta q,\gamma
q)& =c(\lambda |\alpha ,\beta q)\langle \Omega ,\alpha ,\beta ,\gamma |,%
\text{ \ }\langle \Omega ,\alpha ,\beta ,\gamma |D(\lambda |\alpha ,\beta
)\neq 0,  \label{Gauged reference-2}
\end{align}%
with:%
\begin{equation}
b(\lambda q^{1/2}|\alpha ,\beta )c(\lambda q^{-1/2}|\alpha ,\beta q)=-\text{%
det}_{q}M(\lambda ).
\end{equation}%
Here, we show that we can construct such a pseudo-reference state if and
only if we impose at least $N+1$ constraints on the bulk and gauge
parameters.

Let us start our analysis looking to the local conditions to be imposed, in
order to do so let us define the local gauge transformed bulk operators:%
\begin{equation}
\left( 
\begin{array}{ll}
A_{n}(\lambda |\alpha ,\beta ,\gamma ) & B_{n}(\lambda |\alpha ,\beta ) \\ 
C_{n}(\lambda |\alpha ,\beta ,\gamma ) & D_{n}(\lambda |\alpha ,\beta )%
\end{array}%
\right) _{0}=G_{0}^{-1}(\lambda q^{1/2}|\alpha ,\beta )\,L_{0,n}(\lambda
q^{-1/2})\bar{G}_{0}(\lambda q^{1/2}|\gamma )
\end{equation}%
and let us introduce the parameters:%
\begin{eqnarray}
z_{n}^{\left( \epsilon _{n},k_{n}\right) } &=&-q^{k_{n}+1/2}\left[ \left( 
\frac{c_{n}^{p}+d_{n}^{p}}{a_{n}^{p}+b_{n}^{p}}\right) ^{1/p}\frac{a_{n}}{%
c_{n}}\right] ^{(1+\epsilon _{n})/2}\frac{\alpha _{n}}{a_{n}}, \\
w_{n}^{\left( \epsilon _{n},k_{n}\right) } &=&-q^{1/2-k_{n}}\left[ \left( 
\frac{a_{n}^{p}+b_{n}^{p}}{c_{n}^{p}+d_{n}^{p}}\right) ^{1/p}\frac{d_{n}}{%
b_{n}}\right] ^{(1+\epsilon _{n})/2}\frac{\beta _{n}}{d_{n}},
\end{eqnarray}%
where $n\in \left\{ 1,...,N\right\} ,$ $\epsilon _{n}=\pm 1,$ $k_{n}\in
\left\{ 0,...,p-1\right\} $. Let us denote with:%
\begin{equation}
\langle h_{n},n|v_{n}=q^{h_{n}}\langle h_{n},n|,\text{ \ \ }%
v_{n}|h_{n},n\rangle =|h_{n},n\rangle q^{h_{n}},
\end{equation}%
the left and right eigenbasis of the operators $v_{n}$.

\begin{lemma}
Let us assume that:%
\begin{equation}
a_{n}^{p}+b_{n}^{p}\neq 0,\text{ \ \ \ }c_{n}^{p}+d_{n}^{p}\neq 0,
\label{dis-equalities}
\end{equation}%
then the non-zero left state annihilated by the local operator $%
A_{n}(\lambda |\alpha ,\beta ,\gamma )$ exists if and only if we impose the
following two constraints on the gauge parameters:%
\begin{equation}
\beta /\alpha =w_{n}^{\left( \epsilon _{n},k_{n}\right) },\text{ \ \ }\gamma
=z_{n}^{\left( \epsilon _{n},k_{n}\right) },  \label{Two-gauge-constraint}
\end{equation}%
for some fixed $\epsilon _{n}=\pm 1$ and $k_{n}\in \left\{ 0,...,p-1\right\} 
$, moreover this state is uniquely defined by:%
\begin{equation}
\langle \Omega _{n,\alpha ,\beta ,\gamma
}|=\sum_{h_{n}=0}^{p-1}q^{h_{n}(k_{n}+1)}\left[ \prod_{r_{n}=1}^{h_{n}}\frac{%
a_{n}q^{r_{n}-1/2}+b_{n}q^{1/2-r_{n}}}{c_{n}q^{r_{n}-1/2}+d_{n}q^{1/2-r_{n}}}%
\left( \frac{c_{n}^{p}+d_{n}^{p}}{a_{n}^{p}+b_{n}^{p}}\right) ^{1/p}\right]
^{(1+\epsilon _{n})/2}\langle h_{n},n|.  \label{Gen-REF-Local}
\end{equation}%
Similarly, if the condition $\left( \ref{dis-equalities}\right) $ holds the
non-zero right state annihilated by the local operator $A_{n}(\lambda
|\alpha ,\beta ,\gamma )$ exists if and only if we impose the following two
constraints on the gauge parameters:%
\begin{equation}
\beta /\alpha =w_{n}^{\left( \epsilon _{n},k_{n}\right) },\text{ \ \ }\gamma
=z_{n}^{\left( \epsilon _{n},k_{n}-2\right) },
\end{equation}%
for some fixed $\epsilon _{n}=\pm 1$ and $k_{n}\in \left\{ 0,...,p-1\right\} 
$, moreover this state is uniquely defined by:%
\begin{equation}
|\Omega _{n,\alpha ,\beta ,\gamma }\rangle
=\sum_{h_{n}=0}^{p-1}|h_{n},n\rangle q^{-h_{n}k_{n}}\left[
\prod_{r_{n}=1}^{h_{n}}\frac{c_{n}q^{r_{n}-1/2}+d_{n}q^{1/2-r_{n}}}{%
a_{n}q^{r_{n}-1/2}+b_{n}q^{1/2-r_{n}}}\left( \frac{a_{n}^{p}+b_{n}^{p}}{%
c_{n}^{p}+d_{n}^{p}}\right) ^{1/p}\right] ^{(1+\epsilon _{n})/2}.
\end{equation}%
These are pseudo-eigenstates of the operator $B_{n}(\lambda |\alpha ,\beta )$%
:%
\begin{eqnarray}
\langle \Omega _{n,\alpha ,\beta ,\gamma }|B_{n}(\lambda |\alpha ,\beta )
&=&b_{n}(\lambda |\alpha ,\beta )\langle \Omega _{n,\alpha ,\beta /q,\gamma
q}|, \\
B_{n}(\lambda |\alpha ,\beta )|\Omega _{n,\alpha ,\beta ,\gamma }\rangle
&=&qb_{n}(\lambda q|\alpha ,\beta )|\Omega _{n,\alpha ,\beta /q,\gamma
q}\rangle ,
\end{eqnarray}%
where:%
\begin{equation}
b_{n}(\lambda |\alpha ,\beta )=\frac{\beta ^{2}\gamma _{n}}{(1-\beta
^{2})\mu _{n,\epsilon _{n}}}\left( \frac{\lambda }{q^{1/2}\mu _{n,\epsilon
_{n}}}-\frac{q^{1/2}\mu _{n,\epsilon _{n}}}{\lambda }\right) .
\end{equation}
\end{lemma}

\begin{proof}
The lemma is proven by direct construction. Let us introduce a state: 
\begin{equation}
\langle \Omega _{n,\alpha ,\beta ,\gamma }|=\sum_{h=0}^{p-1}c_{h}(n,\alpha
,\beta ,\gamma )\langle h,n|,
\end{equation}%
and look for the conditions to be imposed on $c_{h}(n,\alpha ,\beta ,\gamma
) $ in order to satisfy the equation:%
\begin{equation}
\langle \Omega _{n,\alpha ,\beta ,\gamma }|A_{n}(\lambda |\alpha ,\beta
,\gamma )=0,\text{ \ }\forall \lambda \in \mathbb{C}.
\label{Nilpotent condition}
\end{equation}%
By the definition of $A_{n}(\lambda |\alpha ,\beta ,\gamma )$ it is easy to
verify that we have:%
\begin{equation}
\langle \Omega _{n,\alpha ,\beta ,\gamma }|A_{n}(\lambda |\alpha ,\beta
,\gamma )=\left( -\lambda \sum_{h=0}^{p-1}C_{h}^{+}\langle h,n|+\frac{1}{%
\lambda }\sum_{h=0}^{p-1}C_{h}^{-}\langle h,n|\right) \frac{\beta ^{2}}{%
1-\beta ^{2}},
\end{equation}%
where:%
\begin{eqnarray}
C_{h}^{+} &=&q^{-1/2}c_{h}(\alpha _{n}q^{h}+\left( \beta \gamma /\alpha
\right) \delta _{n}q^{-h})+c_{h-1}\gamma
q^{1/2}(a_{n}q^{h-1/2}+b_{n}q^{1/2-h}), \\
C_{h}^{-} &=&q^{1/2}c_{h}(\beta _{n}q^{-h}+\left( \beta \gamma /\alpha
\right) \gamma _{n}q^{h})+\left( \beta /\alpha \right)
q^{-1/2}c_{h+1}(c_{n}q^{h+1/2}+d_{n}q^{-1/2-h}),
\end{eqnarray}%
and we omit to write explicitly the dependence on $n,\alpha ,\beta ,\gamma $ in $c_{h}$ when
it is not misleading. So that we get the following system of equations:%
\begin{equation}
C_{h}^{+}=0\text{, \ \ \ }C_{h}^{-}=0\text{ \ \ \ }\forall h\in \left\{
0,...,p-1\right\} .  \label{Nilpotence}
\end{equation}%
As we have assumed that the bulk parameters are generic and  satisfy $\left( %
\ref{dis-equalities}\right) $, the equations $C_{h}^{+}=0$ fix the
values of the ratios $E_{h}\equiv c_{h-1}/c_{h}$ and the equations $C_{h}^{-}=0$ fix
the value of ratios $F_{h}\equiv c_{h+1}/c_{h}$ for any $h\in \{0,...,p-1\}$ and
one has to impose the compatibility of these values:%
\begin{equation}
E_{h}=1/F_{h-1}\text{ }\forall h\in \left\{ 0,...,p-1\right\} ,
\end{equation}%
together with the cyclicity condition:%
\begin{equation}
\prod_{h=0}^{p-1}E_{h}=1.
\end{equation}%
Then it is easy to show that the only solution of this system of equation is
obtained fixing the two gauge parameters by $\left( \ref%
{Two-gauge-constraint}\right) $ which correspondingly fixes the form of the
state $\left( \ref{Gen-REF-Local}\right) $.

Let us compute now the action of the operator $B_{n}(\lambda |\alpha ,\beta
) $ on this state; by definition it holds:%
\begin{equation}
B_{n}(\lambda |\alpha ,\beta )=\frac{\beta ^{2}D_{n}(\lambda )-\alpha \beta
\lambda q^{1/2}B_{n}}{\beta ^{2}-1}
\end{equation}%
so that:%
\begin{eqnarray}
&&\langle \Omega _{n,\alpha ,\beta ,\gamma }|B_{n}(\lambda |\alpha ,\beta
)|h_{n},n\rangle =  \notag \\
&&\hspace{1cm}\frac{\beta ^{2}}{1-\beta ^{2}}\left\{ \lambda \left[ c_{h_{n}}%
\frac{q^{-1/2}\delta _{n}}{q^{h_{n}}}+\frac{\alpha }{\beta }%
c_{h_{n}-1}(a_{n}q^{h_{n}-1/2}+b_{n}q^{1/2-h_{n}})\right] -c_{h_{n}}\frac{%
q^{1/2}\gamma _{n}q^{h_{n}}}{\lambda }\right\}  \notag \\
&=&\frac{\beta ^{2}c_{h_{n}}q^{h_{n}}}{\beta ^{2}-1}\left\{ \frac{\lambda }{%
q^{1/2}}\frac{\alpha }{\beta \gamma }\alpha _{n}+\frac{q^{1/2}\gamma _{n}}{%
\lambda }\right\},
\end{eqnarray}%
where to get the third line we used the identity $C_{h_{n}}^{+}=0$. Now
remarking that:%
\begin{equation}
c_{h_{n}}(n,\alpha ,\beta ,\gamma )q^{h_{n}}=c_{h_{n}}(n,\alpha ,\beta
/q,\gamma q),
\end{equation}%
as the effect of $q^{h_{n}}$ is to bring $k_{n}$ to $k_{n}+1$ in the
state $\langle \Omega _{n,\alpha ,\beta ,\gamma }|$, this, for the gauge
choice $\left( \ref{Two-gauge-constraint}\right) $, being equivalent to the
above redefinitions of the gauge parameters. So that we get:%
\begin{equation}
\langle \Omega _{n,\alpha ,\beta ,\gamma }|B_{n}(\lambda |\alpha ,\beta
)=b_{n}(\lambda |\alpha ,\beta )\langle \Omega _{n,\alpha ,\beta /q,\gamma
q}|
\end{equation}%
and so:%
\begin{eqnarray}
b_{n}(\lambda |\alpha ,\beta ) &=&\frac{\beta ^{2}}{\beta ^{2}-1}\left\{ 
\frac{\lambda }{q^{1/2}}\frac{\alpha }{\beta \gamma }\alpha _{n}+\frac{%
q^{1/2}\gamma _{n}}{\lambda }\right\} \\
&=&\frac{\beta ^{2}}{\beta ^{2}-1}\left\{ 
\begin{array}{c}
\frac{q^{1/2}\gamma _{n}}{\lambda }+\frac{\lambda }{q^{1/2}}q^{-1}\frac{%
a_{n}d_{n}}{\beta _{n}}\text{ \ \ for }\epsilon _{n}=-1 \\ 
\frac{q^{1/2}\gamma _{n}}{\lambda }+\frac{\lambda }{q^{1/2}}q^{-1}\frac{%
c_{n}b_{n}}{\beta _{n}}\text{ \ \ for }\epsilon _{n}=1%
\end{array}%
\right. .
\end{eqnarray}%
Similarly, one can prove our statements for the right state and the action
on it of $B_{n}(\lambda |\alpha ,\beta )$.
\end{proof}

Let us remark that if the condition $(\ref{dis-equalities})$ are not
satisfied we can still derive the left and right local reference states
imposing some case dependent condition on the gauge parameters; here for
simplicity we have chosen to omit the description of these cases.

\begin{proposition}
Let us assume that for any $n\in \left\{ 1,...,\mathsf{N}\right\} $ the
conditions $(\ref{dis-equalities})$ is satisfied then the non-zero left
state annihilated by the operator family $A(\lambda |\alpha ,\beta ,\gamma )$
exists if and only if we impose the following $\mathsf{N}+1$ constraints on
the bulk and gauge parameters:%
\begin{equation}
\gamma =z_{1}^{\left( \epsilon _{1},k_{1}\right) },\text{ \ \ }\beta /\alpha
=w_{N}^{\left( \epsilon _{N},k_{N}\right) },\text{ \ \ }w_{n}^{\left(
\epsilon _{n},k_{n}\right) }=1/z_{n+1}^{\left( \epsilon
_{n+1},k_{n+1}\right) }\text{ \ }\forall n\in \left\{ 1,...,\mathsf{N}%
-1\right\} ,  \label{Conditions-nilpotent-g}
\end{equation}%
for fixed $\mathsf{N}$-tuples of $\epsilon _{n}=\pm 1$ and $k_{n}\in \left\{
0,...,p-1\right\} $, moreover it is uniquely defined by:%
\begin{equation}
\langle \Omega ,\alpha ,\beta ,\gamma
|=\sum_{h_{1},...,h_{N}=0}^{p-1}\prod_{n=1}^{N}q^{h_{n}(k_{n}+1)}\left[
\prod_{r_{n}=1}^{h_{n}}\frac{a_{n}q^{r_{n}-1/2}+b_{n}q^{1/2-r_{n}}}{%
c_{n}q^{r_{n}-1/2}+d_{n}q^{1/2-r_{n}}}\left( \frac{a_{n}^{p}+d_{n}^{p}}{%
a_{n}^{p}+b_{n}^{p}}\right) ^{1/p}\right] ^{(1+\epsilon
_{n})/2}\bigotimes_{n=1}^{N}\langle h_{n},n|.  \label{Left-ref}
\end{equation}%
Under the same condition the non-zero right state annihilated by the
operator family $A(\lambda |\alpha ,\beta ,\gamma )$ exists if and only if
we impose the following $\mathsf{N}+1$ constraints on the bulk and gauge
parameters:%
\begin{equation}
\gamma =z_{1}^{\left( \epsilon _{1},k_{1}-2\right) },\text{ \ \ }\beta
/\alpha =w_{N}^{\left( \epsilon _{N},k_{N}\right) },\text{ \ \ }%
w_{n}^{\left( \epsilon _{n},k_{n}\right) }=1/z_{n+1}^{\left( \epsilon
_{n+1},k_{n+1}-2\right) }\text{ \ }\forall n\in \left\{ 1,...,N-1\right\} ,
\label{Conditions-nilpotent-g-R}
\end{equation}%
for fixed $\mathsf{N}$-tuples of $\epsilon _{n}=\pm 1$ and $k_{n}\in \left\{
0,...,p-1\right\} $, moreover it is uniquely defined by:%
\begin{equation}
|\Omega ,\alpha ,\beta ,\gamma \rangle
=\sum_{h_{1},...,h_{N}=0}^{p-1}\prod_{n=1}^{\mathsf{N}}q^{-h_{n}k_{n}}\left[
\prod_{r_{n}=1}^{h_{n}}\frac{c_{n}q^{r_{n}-1/2}+d_{n}q^{1/2-r_{n}}}{%
a_{n}q^{r_{n}-1/2}+b_{n}q^{1/2-r_{n}}}\left( \frac{a_{n}^{p}+b_{n}^{p}}{%
c_{n}^{p}+d_{n}^{p}}\right) ^{1/p}\right] ^{(1+\epsilon
_{n})/2}\bigotimes_{n=1}^{\mathsf{N}}|h_{n},n\rangle .  \label{Right-ref}
\end{equation}%
These are pseudo-eigenstates of $B(\lambda |\alpha ,\beta )$:%
\begin{eqnarray}
\langle \Omega ,\alpha ,\beta ,\gamma |B(\lambda |\alpha ,\beta )
&=&b(\lambda |\alpha ,\beta )\langle \Omega ,\alpha ,\beta /q,\gamma q|, \\
B(\lambda |\alpha ,\beta )|\Omega ,\alpha ,\beta ,\gamma \rangle &=&|\Omega
,\alpha ,\beta q,\gamma /q\rangle q^{N}b(\lambda q|\alpha ,\beta ),
\end{eqnarray}%
with:%
\begin{equation}
b(\lambda |\alpha ,\beta )=\prod_{n=1}^{\mathsf{N}}b_{n}(\lambda |\alpha
,\beta ).
\end{equation}
\end{proposition}

\begin{proof}
The operator family $A(\lambda |\alpha ,\beta ,\gamma )$ is a degree $%
\mathsf{N}$ Laurent polynomial of the form:%
\begin{equation}
A(\lambda |\alpha ,\beta ,\gamma )=\sum_{n=0}^{\mathsf{N}}\lambda
^{2n-N}A_{n}(\alpha ,\beta ,\gamma )
\end{equation}%
where the $A_{n}(\alpha ,\beta ,\gamma )$ are operators, for example we
write explicitly:%
\begin{align}
A_{0}(\alpha ,\beta ,\gamma )& =\frac{-\alpha \prod_{n=1}^{\mathsf{N}}\beta
_{n}v_{n}^{-1}+\beta \gamma \prod_{n=1}^{\mathsf{N}}\gamma
_{n}v_{n}+q^{-1/2}\beta \sum_{a=1}^{\mathsf{N}}\left( \prod_{n=a+1}^{\mathsf{%
N}}\beta _{n}v_{n}^{-1}\right) C_{a}\prod_{n=1}^{a-1}\gamma _{n}v_{n}}{%
\left( \beta -1/\beta \right) \gamma }, \\
A_{N}(\alpha ,\beta ,\gamma )& =\frac{-\alpha \prod_{n=1}^{\mathsf{N}}\alpha
_{n}v_{n}+\beta \gamma \prod_{n=1}^{\mathsf{N}}\delta
_{n}v_{n}^{-1}-q^{1/2}\alpha \gamma \sum_{a=1}^{\mathsf{N}}\left(
\prod_{n=a+1}^{\mathsf{N}}\alpha _{n}v_{n}\right)
B_{a}\prod_{n=1}^{a-1}\delta _{n}v_{n}^{-1}}{\left( \beta -1/\beta \right)
\gamma }.
\end{align}%
For general values of the parameters these are invertible operators so that
we have to impose at least $\mathsf{N}+1$ constraints to have that their
common kernel is at least one dimensional. We can find the set of
constraints by using induction and decomposing $A(\lambda |\alpha ,\beta
,\gamma )$ in terms of gauged operators on two subchains one of $\mathsf{N}%
-1 $ sites and one of $1$ site. The most general decomposition reads:%
\begin{equation}
{A}_{\mathsf{N},...,1}(\lambda |\alpha ,\beta ,\gamma )={A}_{\mathsf{N}%
,...,2}(\lambda |\alpha ,\beta ,x_{1},y_{1}){A}_{1}(\lambda
|x_{1},y_{1},\gamma )+{B}_{\mathsf{N},...,2}(\lambda |\alpha ,\beta
,x_{1},y_{1}){C}_{1}(\lambda |x_{1},y_{1},\gamma )
\end{equation}%
where we have defined:%
\begin{equation}
M(\lambda |\alpha ,\beta ,x,y)=G^{-1}(\lambda q^{1/2}|\alpha ,\beta
)\,M(\lambda )G(\lambda q^{1/2}|x,y) \\
=\left( 
\begin{array}{ll}
A(\lambda |\alpha ,\beta ,x,y) & B(\lambda |\alpha ,\beta ,x,y) \\ 
C(\lambda |\alpha ,\beta ,x,y) & D(\lambda |\alpha ,\beta ,x,y)%
\end{array}%
\right) ,
\end{equation}%
and we have explicitly pointed out in the subscripts the quantum sites to
which the operator are referred. The following identities holds:%
\begin{eqnarray}
A(\lambda |\alpha ,\beta ,x,y) &=&-A(\lambda |\alpha ,\beta ,xy), \\
B(\lambda |\alpha ,\beta ,x,y) &=&A(\lambda |\alpha ,\beta ,x/y), \\
C(\lambda |x,y,\gamma ) &=&\frac{x^{2}y^{2}-1}{1-y^{2}}A(\lambda
|1,1/xy,\gamma ), \\
A(\lambda |x,y,\gamma ) &=&\frac{x^{2}-y^{2}}{1-y^{2}}A(\lambda
|1,y/x,\gamma ),
\end{eqnarray}%
from which it follows:%
\begin{align}
{A}_{\mathsf{N},...,1}(\lambda |\alpha ,\beta ,\gamma )=& \frac{\left(
y_{1}^{2}-x_{1}^{2}\right) {A}_{\mathsf{N},...,2}(\lambda |\alpha ,\beta
,x_{1}y_{1}){A}_{1}(\lambda |1,y_{1}/x_{1},\gamma )}{1-y_{1}^{2}}  \notag \\
+& \frac{(x_{1}^{2}y_{1}^{2}-1){A}_{\mathsf{N},...,2}(\lambda |\alpha ,\beta
,x_{1}/y_{1}){A}_{1}(\lambda |1,1/x_{1}y_{1},\gamma )}{1-y_{1}^{2}}.
\end{align}%
Then ${A}_{\mathsf{N},...,1}(\lambda |\alpha ,\beta ,\gamma )$ admits a
non-zero state annihilated by its action once we impose that it is true for ${A}_{\mathsf{N},...,2}(\lambda |\alpha ,\beta ,x_{1}y_{1})$ and ${A%
}_{1}(\lambda |1,1/x_{1}y_{1},\gamma )$ or for ${A}_{1}(\lambda
|1,y_{1}/x_{1},\gamma )$ and ${A}_{\mathsf{N},...,2}(\lambda |\alpha ,\beta
,x_{1}/y_{1})$, and this state is given by the tensor product of the ones on
the two subchains. As the parameters $x_{1}$ and $y_{1}$ are arbitrary in
fact these two conditions are equivalents and so we can chose just one of
them. So let us say we ask the second one and we repeat the same argument
for ${A}_{\mathsf{N},...,2}(\lambda |\alpha ,\beta ,x_{1}/y_{1})$, i.e. ${A}%
_{\mathsf{N},...,2}(\lambda |\alpha ,\beta ,x_{1}/y_{1})$ admits such a
state if ${A}_{2}(\lambda |1,y_{2}/x_{2},x_{1}/y_{1})$ and ${A}_{\mathsf{N}%
,...,3}(\lambda |\alpha ,\beta ,x_{2}/y_{2})$ do. So on by induction we get
that the existence condition is equivalent to the existence conditions for
the following $\mathsf{N}$ local operators:%
\begin{equation}
{A}_{n}(\lambda |1,y_{n}/x_{n},x_{n-1}/y_{n-1})\text{ for any }n\in \left\{
1,...,N\right\} ,  \label{One-site-A}
\end{equation}%
where we have denoted%
\begin{equation}
y_{N}/x_{N}=\beta /\alpha ,\text{ \ }x_{0}/y_{0}=\gamma ,
\end{equation}%
while the $y_{n}/x_{n}$ for any $n\in \left\{ 1,...,\mathsf{N}-1\right\} $
are free parameters to be used to satisfy the existence condition for the
local operators ${A}_{n}(\lambda |1,y_{n}/x_{n},x_{n-1}/y_{n-1})$. From the
previous lemma for ${A}_{n}(\lambda |1,s_{n},r_{n})$, the existence
condition is equivalent to:%
\begin{equation}
s_{n}=w_{n}^{\left( \epsilon _{n}\right) }\text{ \ and \ }%
r_{n}=z_{n}^{\left( \epsilon _{n}\right) }
\end{equation}%
for any $\epsilon _{n}=\pm 1$ and the right state annihilated by ${A}%
_{n}(\lambda |1,s_{n},r_{n})$ reads:%
\begin{equation}
\langle \Omega _{n},s_{n},r_{n}|=\sum_{h_{n}=0}^{p-1}q^{h_{n}(k_{n}+1)}\left[
\prod_{k_{n}=1}^{h_{n}}\frac{a_{n}q^{k_{n}-1/2}+b_{n}q^{1/2-k_{n}}}{%
c_{n}q^{k_{n}-1/2}+d_{n}q^{1/2-k_{n}}}\left( \frac{c_{n}^{p}+d_{n}^{p}}{%
a_{n}^{p}+b_{n}^{p}}\right) ^{1/p}\right] ^{(1+\epsilon _{n})/2}\langle
h_{n},n|.  \label{local reference states}
\end{equation}%
From this it is clear that the existence conditions of such a state for ${A}%
_{\mathsf{N},...,1}(\lambda |\alpha ,\beta ,\gamma )$ coincides with the
simultaneous existence for the $\mathsf{N}$ local operators $\left( \ref{One-site-A}%
\right) $ and that the state is just the tensor product of the states $%
\left( \ref{local reference states}\right) $ so that our proposition is
proven. Similarly, we can prove the statement for the right state and using
the previous lemma we can prove our statement on the action of the operator $%
B(\lambda |\alpha ,\beta )$ on these states.
\end{proof}

\section{Gauge transformed Reflection algebra}

\subsection{Gauge transformed boundary operators}

The gauged two-row monodromy matrix can be defined as it follows: 
\begin{equation}
\mathcal{U}_{-}(\lambda |\alpha ,\beta )\equiv \frac{q^{1/2}}{\lambda }%
G^{-1}(\lambda q^{1/2}|\alpha ,\beta )\,\mathcal{U}_{-}(\lambda
)\,G(q^{1/2}/\lambda |\alpha ,\beta )=\left( 
\begin{array}{ll}
\mathcal{A}_{-}(\lambda |\alpha ,\beta q^{2}) & \mathcal{B}_{-}(\lambda
|\alpha ,\beta ) \\ 
\,\mathcal{C}_{-}(\lambda |\alpha ,\beta q^{2}) & \mathcal{D}_{-}(\lambda
|\alpha ,\beta )%
\end{array}%
\right) .  \label{U-gauge}
\end{equation}%
Note that one can expand this last gauged monodromy matrix in terms of the
gauged bulk ones. Moreover, $\mathcal{U}_{-}(\lambda |\alpha ,\beta )$ does
not depend on the internal gauge parameter $\gamma $, so we are free to
chose it at will. The following decompositions hold: 
\begin{align}
\left( 
\begin{array}{l}
\mathcal{A}_{-}(\lambda |\alpha ,\beta q^{2}) \\ 
\,\mathcal{C}_{-}(\lambda |\alpha ,\beta q^{2})%
\end{array}%
\right) \left. =\right. & M(\lambda |\alpha ,\beta ,\gamma )\bar{K}%
_{-}(\lambda |\gamma )\left( 
\begin{array}{l}
\bar{A}(\lambda |\alpha ,\beta q,\gamma q) \\ 
\bar{C}(\lambda |\alpha ,\beta q,\gamma q)%
\end{array}%
\right) \\
\left( 
\begin{array}{l}
\mathcal{B}_{-}(\lambda |\alpha ,\beta ) \\ 
\mathcal{D}_{-}(\lambda |\alpha ,\beta )%
\end{array}%
\right) \left. =\right. & M(\lambda |\alpha ,\beta ,\gamma )\bar{K}%
_{-}(\lambda |\gamma )\left( 
\begin{array}{l}
\bar{B}(\lambda |\alpha ,\beta /q,\gamma q) \\ 
\bar{D}(\lambda |\alpha ,\beta /q,\gamma q)%
\end{array}%
\right) ,
\end{align}%
where 
\begin{equation}
\bar{K}_{-}(\lambda |\gamma )=\frac{q^{1/2}}{\lambda }\bar{G}^{-1}(\lambda
q^{1/2}|\gamma )\,\,K_{-}(\lambda )\,\bar{G}(q^{1/2}/\lambda |\gamma q).
\end{equation}%
Explicitly, for $\mathcal{B}_{-}(\lambda |\alpha ,\beta )$, it holds:%
\begin{align}
\frac{\mathcal{B}_{-}(\lambda |\alpha ,\beta )}{f(\alpha ,\beta /q,\gamma q)}%
& =\bar{K}_{-}(\lambda |\gamma )_{12}A(\lambda |\alpha ,\beta ,\gamma
)A(1/\lambda |\alpha ,\beta /q,\gamma q)-\bar{K}_{-}(\lambda |\gamma
)_{11}A(\lambda |\alpha ,\beta ,\gamma )B(1/\lambda |\alpha ,\beta /q) 
\notag \\
& +\bar{K}_{-}(\lambda |\gamma )_{21}B(\lambda |\alpha ,\beta )B(1/\lambda
|\alpha ,\beta /q)-\bar{K}_{-}(\lambda |\gamma )_{22}B(\lambda |\alpha
,\beta )A(1/\lambda |\alpha ,\beta /q,\gamma q).  \label{B_gauge-decompo}
\end{align}%
where:%
\begin{align}
\bar{K}_{-}(\lambda |\gamma )_{11}& =\frac{\lambda ^{2}(\zeta
_{-}/q+q^{2}\gamma \kappa _{-}e^{\tau _{-}})-q^{2}\gamma \kappa _{-}e^{\tau
_{-}}/\left( \zeta _{-}\lambda ^{2}\right) -1/\zeta _{-}}{\zeta _{-}-1/\zeta
_{-}}, \\
\bar{K}_{-}(\lambda |\gamma )_{22}& =\frac{(\zeta _{-}+q\gamma \kappa
_{-}e^{\tau _{-}})/\lambda ^{2}-\gamma \lambda ^{2}-1/\zeta _{-}}{\zeta
_{-}-1/\zeta _{-}}, \\
\bar{K}_{-}(\lambda |\gamma )_{12}& =\frac{q\gamma \kappa _{-}e^{\tau
_{-}}(\lambda ^{2}/q-q/\lambda ^{2})}{\zeta _{-}-1/\zeta _{-}} \\
\bar{K}_{-}(\lambda |\gamma )_{21}& =\frac{(q/\lambda ^{2}-\lambda ^{2}/q)%
\left[ q\gamma e^{\tau _{-}}\kappa _{-}+\zeta _{-}-\kappa _{-}/(q\gamma
e^{\tau _{-}})\right] }{\zeta _{-}-1/\zeta _{-}}.
\end{align}%
Then it holds: 
\begin{equation}
\bar{K}_{-}(\lambda |\gamma )_{21}\equiv 0,\text{ }\forall \lambda \in 
\mathbb{C}\text{ \ \ for }\gamma =\gamma _{\epsilon }
\end{equation}%
for $\epsilon =\pm $ and%
\begin{equation}
\gamma _{\epsilon }=\frac{-\zeta _{-}+\epsilon \sqrt{\zeta _{-}^{2}+4\kappa
_{-}^{2}}}{2qe^{\tau _{-}}\kappa _{-}}.  \label{K21-zeros}
\end{equation}%
These gauge transformed boundary operators satisfies the following gauge
deformed reflection algebra.

\begin{proposition}
The gauge transformed boundary operators satisfy the following commutation
relations:%
\begin{align}
\mathcal{B}_{-}(\lambda _{2}|\beta )\mathcal{B}_{-}(\lambda _{1}|\beta
/q^{2})& =\mathcal{B}_{-}(\lambda _{1}|\beta )\mathcal{B}_{-}(\lambda
_{2}|\beta /q^{2}),  \label{B-comm} \\
\mathcal{A}_{-}(\lambda _{2}|\beta q^{2})\mathcal{B}_{-}(\lambda _{1}|\beta
)& =\frac{(\lambda _{1}q/\lambda _{2}-\lambda _{2}/q\lambda _{1})(\lambda
_{1}\lambda _{2}/q-q/\lambda _{1}\lambda _{2})}{(\lambda _{1}/\lambda
_{2}-\lambda _{2}/\lambda _{1})(\lambda _{1}\lambda _{2}-1/\lambda
_{1}\lambda _{2})}\mathcal{B}_{-}(\lambda _{1}|\beta )\mathcal{A}%
_{-}(\lambda _{2}|\beta )  \notag \\
& +\frac{(\lambda _{1}\lambda _{2}/q-q/\lambda _{1}\lambda _{2})(\lambda
_{1}\beta /q\lambda _{2}-\lambda _{2}q/\beta \lambda _{1})(q-1/q)}{(\lambda
_{1}/\lambda _{2}-\lambda _{2}/\lambda _{1})(\lambda _{1}\lambda
_{2}-1/\lambda _{1}\lambda _{2})(\beta /q-q/\beta )}\mathcal{B}_{-}(\lambda
_{2}|\beta )\mathcal{A}_{-}(\lambda _{1}|\beta )  \notag \\
& +\frac{(\lambda _{1}\lambda _{2}/\beta -\beta /\lambda _{1}\lambda
_{2})(q-1/q)}{(\lambda _{1}\lambda _{2}-1/\lambda _{1}\lambda _{2})(\beta
/q-q/\beta )}\mathcal{B}_{-}(\lambda _{2}|\beta )\mathcal{D}_{-}(\lambda
_{1}|\beta ), \\
\mathcal{B}_{-}(\lambda _{1}|\beta )\mathcal{D}_{-}(\lambda _{2}|\beta )& =%
\frac{(\lambda _{1}q/\lambda _{2}-\lambda _{2}/q\lambda _{1})(\lambda
_{1}\lambda _{2}/q-q/\lambda _{1}\lambda _{2})}{(\lambda _{1}/\lambda
_{2}-\lambda _{2}/\lambda _{1})(\lambda _{1}\lambda _{2}-1/\lambda
_{1}\lambda _{2})}\mathcal{D}_{-}(\lambda _{2}|\beta q^{2})\mathcal{B}%
_{-}(\lambda _{1}|\beta )  \notag \\
& -\frac{(\lambda _{1}\lambda _{2}/q-q/\lambda _{1}\lambda _{2})(\lambda
_{2}\beta q/\lambda _{1}-\lambda _{1}/\lambda _{2}\beta q)(q-1/q)}{(\lambda
_{1}/\lambda _{2}-\lambda _{2}/\lambda _{1})(\lambda _{1}\lambda
_{2}-1/\lambda _{1}\lambda _{2})(\beta q-1/\beta q)}\mathcal{D}_{-}(\lambda
_{1}|\beta q^{2})\mathcal{B}_{-}(\lambda _{2}|\beta )  \notag \\
& -\frac{(\lambda _{1}\lambda _{2}\beta -1/\lambda _{1}\lambda _{2}\beta
)(q-1/q)}{(\lambda _{1}\lambda _{2}-1/\lambda _{1}\lambda _{2})(\beta
q-1/q\beta )}\mathcal{A}_{-}(\lambda _{1}|\beta q^{2})\mathcal{B}%
_{-}(\lambda _{2}|\beta ),
\end{align}%
and%
\begin{align}
& \mathcal{A}_{-}(\lambda _{1}|\beta q^{2})\mathcal{A}_{-}(\lambda
_{2}|\beta q^{2})-\frac{(\lambda _{1}\lambda _{2}/\beta -1/\lambda
_{1}\lambda _{2})(q-1/q)}{(\lambda _{1}\lambda _{2}-1/\lambda _{1}\lambda
_{2})(\beta /q-q/\beta )}\mathcal{B}_{-}(\lambda _{1}|\beta )\mathcal{C}%
_{-}(\lambda _{2}|\beta q^{2})\left. =\right.  \notag \\
& \text{\ \ \ \ \ \ \ \ \ \ \ \ \ \ \ \ \ \ \ \ \ \ \ \ \ \ \ \ }\mathcal{A}%
_{-}(\lambda _{2}|\beta q^{2})\mathcal{A}_{-}(\lambda _{1}|\beta q^{2})-%
\frac{(\lambda _{1}\lambda _{2}/\beta -1/\lambda _{1}\lambda _{2})(q-1/q)}{%
(\lambda _{1}\lambda _{2}-1/\lambda _{1}\lambda _{2})(\beta /q-q/\beta )}%
\mathcal{B}_{-}(\lambda _{2}|\beta )\mathcal{C}_{-}(\lambda _{1}|\beta
q^{2}).
\end{align}%
Similar commutation relations involving $\mathcal{C}_{-}(\lambda |\beta )$
can be written by using the following $\beta $-symmetries:%
\begin{equation}
\mathcal{B}_{-}(\lambda |\beta )=\mathcal{C}_{-}(\lambda |q^{2}/\beta ),%
\text{ \ \ }\mathcal{A}_{-}(\lambda |\beta )=\mathcal{D}_{-}(\lambda
|q^{2}/\beta ).  \label{B-to-C-identity}
\end{equation}%
Moreover, these gauge transformed operators satisfy the following parity
properties:%
\begin{align}
\mathcal{A}_{-}(\lambda |\beta )& =-\frac{(q-1/q)(\lambda ^{2}q/\beta -\beta
/\lambda ^{2}q)}{(\beta /q^{2}-q^{2}/\beta )(\lambda ^{2}-1/\lambda ^{2})}%
\mathcal{D}_{-}(\lambda |\beta )+\frac{(\beta /q-q/\beta )(\lambda
^{2}/q-q/\lambda ^{2})}{(\beta /q^{2}-q^{2}/\beta )(\lambda ^{2}-1/\lambda
^{2})}\mathcal{D}_{-}(1/\lambda |\beta ),  \label{A-Sym} \\
\mathcal{D}_{-}(\lambda |\beta )& =\frac{(q-1/q)(\lambda ^{2}\beta
/q-q/\lambda ^{2}\beta )}{(\beta /q^{2}-q^{2}/\beta )(\lambda ^{2}-1/\lambda
^{2})}\mathcal{A}_{-}(\lambda |\beta )+\frac{(\beta /q-q/\beta )(\lambda
^{2}/q-q/\lambda ^{2})}{(\beta -1/\beta )(\lambda ^{2}-1/\lambda ^{2})}%
\mathcal{A}_{-}(1/\lambda |\beta ), \\
\mathcal{B}_{-}(1/\lambda |\beta )& =-\frac{(\lambda ^{2}q-1/q\lambda ^{2})}{%
(\lambda ^{2}/q-q/\lambda ^{2})}\mathcal{B}_{-}(\lambda |\beta )\text{ },%
\text{ \ \ \ \ }\mathcal{C}_{-}(1/\lambda |\beta )=-\frac{(\lambda
^{2}q-1/q\lambda ^{2})}{(\lambda ^{2}/q-q/\lambda ^{2})}\mathcal{C}%
_{-}(\lambda |\beta ).  \label{B-C-Sym}
\end{align}
\end{proposition}

\begin{proof}
Both the commutation relations and the parity properties here presented
coincide with those derived in \cite{OpenCyFalKN14} for the case of the XXZ spin 1/2
quantum chain with general integrable boundaries. This is the case as they
are clearly representation independent. Here we are just writing them in a
Laurent polynomial form instead of a trigonometric form.
\end{proof}

\subsection{Representation of the gauge transformed Reflection algebra}

In the bulk of the paper we have anticipated that for almost all the values
of the boundary, bulk and gauge parameters the operator family $\mathcal{B}%
_{-}(\lambda |\beta )$ is pseudo-diagonalizable. We will show this statement
in the last subsection of this appendix, but for now we want to write explicitly the
representation of the other gauge transformed boundary operator families in
the left and right basis formed out of the pseudo-eigenstates of $\mathcal{B}%
_{-}(\lambda |\beta )$.

\begin{theorem}
The action of the reflection algebra generator $\mathcal{A}_{-}(\lambda
|\beta q^{2})$ on the generic state $\langle \beta ,\mathbf{h}|$ is given by
the following expression:%
\begin{align}
\langle \beta ,\mathbf{h}|\mathcal{A}_{-}(\lambda |\beta q^{2})&
=\sum_{a=1}^{2\mathsf{N}}\frac{\zeta _{a}^{(h_{a})}\left( \lambda
^{2}/q-q/\lambda ^{2}\right) (\lambda \zeta _{a}^{(h_{a})}-1/(\lambda \zeta
_{a}^{(h_{a})}))\mathsf{A}_{-}(\zeta _{a}^{(h_{a})})}{\lambda \left( (\zeta
_{a}^{(h_{a})})^{2}/q-q/(\zeta _{a}^{(h_{a})})^{2}\right) \left( (\zeta
_{a}^{(h_{a})})^{2}-1/(\zeta _{a}^{(h_{a})})^{2}\right) }\prod_{\substack{ %
b=1  \\ b\neq a\text{ mod}\mathsf{N}}}^{\mathsf{N}}\frac{\Lambda
-X_{b}^{(h_{b})}}{X_{a}^{(h_{a})}-X_{b}^{(h_{b})}}  \notag \\
& \times \langle \beta ,\mathbf{h}|T_{a}^{-\varphi _{a}}+(-1)^{\mathsf{N}}%
\frac{q^{1/2}}{2\lambda }\text{det}_{q}M(1)(\frac{\lambda }{q^{1/2}}+\frac{%
q^{1/2}}{\lambda })\prod_{b=1}^{\mathsf{N}}\frac{\Lambda -X_{b}^{(h_{b})}}{%
X-X_{b}^{(h_{b})}}\langle \beta ,\mathbf{h}|  \notag \\
& +(-1)^{\mathsf{N}+1}\frac{iq^{1/2}}{2\lambda }\frac{\zeta _{-}+1/\zeta _{-}%
}{\zeta _{-}-1/\zeta _{-}}\text{det}_{q}M(i)(\frac{\lambda }{q^{1/2}}-\frac{%
q^{1/2}}{\lambda })\prod_{b=1}^{\mathsf{N}}\frac{\Lambda -X_{b}^{(h_{b})}}{%
X+X_{b}^{(h_{b})}}\langle \beta ,\mathbf{h}|  \notag \\
& +q\left( \lambda ^{2}/q-q/\lambda ^{2}\right) \prod_{b=1}^{\mathsf{N}%
}(\Lambda -X_{b}^{(h_{b})})\langle \beta ,\mathbf{h}|\mathcal{A}_{-}^{\infty
}(\beta q^{2}),  \label{L-SOV A-}
\end{align}%
where%
\begin{align}
\langle \beta ,\mathbf{h}|\mathcal{A}_{-}^{\infty }(\beta q^{2})& =\frac{1}{%
q(1-\beta ^{2})}\left[ \sum_{a=1}^{2\mathsf{N}}\frac{\prod 
_{\substack{ b=1  \\ b\neq a\text{ mod}\mathsf{N}}}^{\mathsf{N}} 
(X_{a}^{(h_{a})}-X_{b}^{(h_{b})})^{-1} \mathsf{A}_{-}(\zeta
_{a}^{(h_{a})})\langle \beta ,\mathbf{h}|T_{a}^{-\varphi _{a}}}{\left(
(\zeta _{a}^{(h_{a})})^{2}/q-q/(\zeta _{a}^{(h_{a})})^{2}\right) \left(
(\zeta _{a}^{(h_{a})})^{2}-1/(\zeta _{a}^{(h_{a})})^{2}\right) }\right.  \notag \\
& \left. +(-1)^{\mathsf{N}}\left( \frac{1}{2}\text{det}_{q}M(1)\prod_{b=1}^{%
\mathsf{N}}\frac{1}{X-X_{b}^{(h_{b})}}+\frac{i}{2}\frac{\zeta _{-}+1/\zeta
_{-}}{\zeta _{-}-1/\zeta _{-}}\text{det}_{q}M(i)\prod_{b=1}^{\mathsf{N}}%
\frac{1}{X+X_{b}^{(h_{b})}}\right) \langle \beta ,\mathbf{h}|\right]  \notag
\\
& +\frac{\kappa _{-}}{q(\zeta _{-}-1/\zeta _{-})}\left( \frac{\prod_{a=1}^{%
\mathsf{N}}\gamma _{a}\delta _{a}}{q\alpha e^{\tau _{-}}}-q\alpha e^{\tau
_{-}}\prod_{a=1}^{\mathsf{N}}\alpha _{a}\beta _{a}\right) \langle \beta ,%
\mathbf{h}|,
\end{align}%
and%
\begin{equation}
\langle \beta ,h_{1},...,h_{a},...,h_{\mathsf{N}}|T_{a}^{\pm }=\langle \beta
,h_{1},...,h_{a}\pm 1,...,h_{\mathsf{N}}|,
\end{equation}%
once the parameter $\alpha $ has been fixed by \eqref{guage-alpha}.
\smallskip
\end{theorem}

\begin{proof}[Proof]
The following interpolation formula:%
\begin{align}
\langle \beta ,\mathbf{h}|\mathcal{A}_{-}(\lambda |\beta q^{2})&
=\sum_{a=1}^{2\mathsf{N}}\frac{\zeta _{a}^{(h_{a})}\left( \lambda
^{2}/q-q/\lambda ^{2}\right) \mathsf{A}_{-}(\zeta _{a}^{(h_{a})})}{\lambda
\left( (\zeta _{a}^{(h_{a})})^{2}/q-q/(\zeta _{a}^{(h_{a})})^{2}\right) }%
\prod_{\substack{ b=1  \\ b\neq a\text{ }}}^{2\mathsf{N}}\frac{\lambda
/\zeta _{b}^{(h_{b})}-\zeta _{b}^{(h_{b})}/\lambda }{\zeta
_{a}^{(h_{a})}/\zeta _{b}^{(h_{b})}-\zeta _{b}^{(h_{b})}/\zeta _{a}^{(h_{a})}%
}  \notag \\
& \times \langle \beta ,\mathbf{h}|T_{a}^{-\varphi _{a}}+(-1)^{\mathsf{N}}%
\frac{q^{1/2}}{2\lambda }\text{det}_{q}M(1)(\frac{\lambda }{q^{1/2}}+\frac{%
q^{1/2}}{\lambda })\prod_{b=1}^{\mathsf{N}}\frac{\lambda /\zeta
_{b}^{(h_{b})}-\zeta _{b}^{(h_{b})}/\lambda }{q^{1/2}/\zeta
_{b}^{(h_{b})}-\zeta _{b}^{(h_{b})}/q^{1/2}}\langle \beta ,\mathbf{h}| 
\notag \\
& +(-1)^{\mathsf{N}+1}\frac{iq^{1/2}}{2\lambda }\frac{\zeta _{-}+1/\zeta _{-}%
}{\zeta _{-}-1/\zeta _{-}}\text{det}_{q}M(i)(\frac{\lambda }{q^{1/2}}-\frac{%
q^{1/2}}{\lambda })\prod_{b=1}^{\mathsf{N}}\frac{\lambda /\zeta
_{b}^{(h_{b})}-\zeta _{b}^{(h_{b})}/\lambda }{q^{1/2}/\zeta
_{b}^{(h_{b})}-\zeta _{b}^{(h_{b})}/q^{1/2}}\langle \beta ,\mathbf{h}| 
\notag \\
& +q\left( \lambda ^{2}/q-q/\lambda ^{2}\right) \prod_{b=1}^{\mathsf{N}%
}(\Lambda -X_{b}^{(h_{b})})\langle \beta ,\mathbf{h}|\mathcal{A}_{-}^{\infty
}(\beta q^{2}),
\end{align}%
where we have defined:%
\begin{equation}
\mathcal{A}_{-}^{\infty ,0}(\beta q^{2})=\lim_{\lambda \rightarrow \infty
,0}\lambda ^{\mp 2(\mathsf{N}+1)}\mathcal{A}_{-}(\lambda |\beta q^{2}),
\end{equation}%
is a direct consequence of the functional dependence with respect to $%
\lambda $:

\begin{equation}
\mathcal{A}_{-}(\lambda |\beta )=\sum_{a=0}^{2\mathsf{N}+2}\lambda ^{2\left(
a-(\mathsf{N}+1\right) )}\mathcal{A}_{a}(\beta )
\end{equation}%
and of the identities:%
\begin{equation}
\mathcal{U}_{-}(q^{1/2})=(-1)^{\mathsf{N}}\text{det}_{q}M(1)\text{ }I_{0},%
\text{ \ \ }\mathcal{U}_{-}(iq^{1/2})=-i\frac{\zeta _{-}+1/\zeta _{-}}{\zeta
_{-}-1/\zeta _{-}}\text{det}_{q}M(i)\text{ }\sigma _{0}^{z},
\end{equation}%
which are representation independent. Instead the asymptotic operators $%
\mathcal{A}_{-}^{\infty ,0}(\beta q^{2})$ depend on the representation and
we can compute them observing that using the definition $(\ref{A-gauge-def})$
of $\mathcal{A}_{-}(\lambda |\beta q^{2})$ it holds:%
\begin{eqnarray}
\mathcal{A}_{-}^{\infty }(\beta q^{2}) &=&\left[ -\mathcal{A}_{-}^{\infty
}/\beta q^{1/2}-\alpha q\mathcal{B}_{-}^{\infty }+\mathcal{C}_{-}^{\infty
}/\alpha q\right] /(\beta -1/\beta ) \\
\mathcal{A}_{-}^{0}(\beta q^{2}) &=&\left[ \beta q^{1/2}\mathcal{D}%
_{-}^{0}-\alpha q\mathcal{B}_{-}^{0}+\mathcal{C}_{-}^{0}/\alpha q\right]
/(\beta -1/\beta )
\end{eqnarray}%
where the $\mathcal{A}_{-}^{\infty ,0},$ $\mathcal{D}_{-}^{\infty ,0},$ $%
\mathcal{B}_{-}^{\infty ,0}$ and $\mathcal{C}_{-}^{\infty ,0}$ are the
asymptotic limits of the ungauged elements of $\mathcal{U}_{-}(\lambda )$.
The identities:%
\begin{equation}
\mathcal{A}_{-}^{\infty ,0}=q^{\mp 1}\mathcal{D}_{-}^{0,\infty },\text{ }%
\mathcal{B}_{-}^{0}=-q^{2}\mathcal{B}_{-}^{\infty },\text{ }\mathcal{C}%
_{-}^{0}=-q^{2}\mathcal{C}_{-}^{\infty },
\end{equation}%
following from (\ref{A-Sym})-(\ref{B-C-Sym}), and%
\begin{equation}
\mathcal{B}_{-}^{\infty }=\frac{\kappa _{-}e^{\tau _{-}}\prod_{a=1}^{\mathsf{%
N}}\alpha _{a}\beta _{a}}{q(\zeta _{-}-1/\zeta _{-})},\text{ }\mathcal{C}%
_{-}^{\infty }=\frac{\kappa _{-}e^{-\tau _{-}}\prod_{a=1}^{\mathsf{N}}\gamma
_{a}\delta _{a}}{q(\zeta _{-}-1/\zeta _{-})},
\end{equation}%
imply the main identity:%
\begin{equation}
\beta q\mathcal{A}_{-}^{\infty }(\beta q^{2})+\mathcal{A}_{-}^{0}(\beta
q^{2})/(\beta q)=\frac{\kappa _{-}(\beta -1/\beta )}{\zeta _{-}-1/\zeta _{-}}%
\left[ \frac{\prod_{a=1}^{\mathsf{N}}\gamma _{a}\delta _{a}}{q\alpha e^{\tau
_{-}}}-q\alpha e^{\tau _{-}}\prod_{a=1}^{\mathsf{N}}\alpha _{a}\beta _{a}%
\right] .  \label{main-identity}
\end{equation}%
This identity allows to compute these asymptotic operators once we use the
interpolation formula to write $\mathcal{A}_{-}^{0}(\beta q^{2})$ in terms
of $\mathcal{A}_{-}^{\infty }(\beta q^{2})$ as it follows:%
\begin{align}
\langle \beta ,\mathbf{h}|\mathcal{A}_{-}^{0}(\beta q^{2})& =\sum_{a=1}^{2%
\mathsf{N}}\frac{q\mathsf{A}_{-}(\zeta _{a}^{(h_{a})})}{\left( (\zeta
_{a}^{(h_{a})})^{2}/q-q/(\zeta _{a}^{(h_{a})})^{2}\right) \left( (\zeta
_{a}^{(h_{a})})^{2}-1/(\zeta _{a}^{(h_{a})})^{2}\right) }\prod_{\substack{ %
b=1  \\ b\neq a\text{ mod}\mathsf{N}}}^{\mathsf{N}}\frac{1}{%
X_{a}^{(h_{a})}-X_{b}^{(h_{b})}}  \notag \\
& \times \langle \beta ,\mathbf{h}|T_{a}^{-\varphi _{a}}+(-1)^{\mathsf{N}}%
\frac{q}{2}\text{det}_{q}M(1)\prod_{b=1}^{\mathsf{N}}\frac{1}{%
X-X_{b}^{(h_{b})}}\langle \beta ,\mathbf{h}|  \notag \\
& +(-1)^{\mathsf{N}}\frac{iq}{2}\frac{\zeta _{-}+1/\zeta _{-}}{\zeta
_{-}-1/\zeta _{-}}\text{det}_{q}M(i)\prod_{b=1}^{\mathsf{N}}\frac{1}{%
X+X_{b}^{(h_{b})}}\langle \beta ,\mathbf{h}|  \notag \\
& -q^{2}\langle \beta ,\mathbf{h}|\mathcal{A}_{-}^{\infty }(\beta q^{2}).
\end{align}
\end{proof}

Similarly, the following theorem characterizes the right SoV representation
of the gauged cyclic reflection algebra:

\begin{theorem}
The action of the reflection algebra generators $\mathcal{D}_{-}(\lambda
|\beta )$ on the generic state $|\beta ,\mathbf{h}\rangle $, can be written
as it follows:%
\begin{align}
\mathcal{D}_{-}(\lambda |\beta )|\beta ,\mathbf{h}\rangle & =\sum_{a=1}^{2%
\mathsf{N}}T_{a}^{-\varphi _{a}}|\beta ,\mathbf{h}\rangle \frac{\zeta
_{a}^{(h_{a})}\left( \lambda ^{2}/q-q/\lambda ^{2}\right) (\lambda \zeta
_{a}^{(h_{a})}-1/(\lambda \zeta _{a}^{(h_{a})}))\mathsf{D}_{-}(\zeta
_{a}^{(h_{a})})}{\lambda \left( (\zeta _{a}^{(h_{a})})^{2}/q-q/(\zeta
_{a}^{(h_{a})})^{2}\right) \left( (\zeta _{a}^{(h_{a})})^{2}-1/(\zeta
_{a}^{(h_{a})})^{2}\right) } \cdot  \notag \\
\cdot  \prod_{\substack{ b=1  \\ b\neq a\text{ mod}%
\mathsf{N}}}^{\mathsf{N}}\frac{\Lambda -X_{b}^{(h_{b})}}{%
X_{a}^{(h_{a})}-X_{b}^{(h_{b})}}& +|\beta ,\mathbf{h}\rangle (-1)^{\mathsf{N}}\frac{q^{1/2}}{2\lambda }\text{%
det}_{q}M(1)(\frac{\lambda }{q^{1/2}}+\frac{q^{1/2}}{\lambda })\prod_{b=1}^{%
\mathsf{N}}\frac{\Lambda -X_{b}^{(h_{b})}}{X-X_{b}^{(h_{b})}}  \notag \\
& +|\beta ,\mathbf{h}\rangle (-1)^{\mathsf{N}}\frac{iq^{1/2}}{2\lambda }%
\frac{\zeta _{-}+1/\zeta _{-}}{\zeta _{-}-1/\zeta _{-}}\text{det}_{q}M(i)(%
\frac{\lambda }{q^{1/2}}-\frac{q^{1/2}}{\lambda })\prod_{b=1}^{\mathsf{N}}%
\frac{\Lambda -X_{b}^{(h_{b})}}{X+X_{b}^{(h_{b})}}  \notag \\
& +q\left( \lambda ^{2}/q-q/\lambda ^{2}\right) \prod_{b=1}^{\mathsf{N}%
}(\Lambda -X_{b}^{(h_{b})})\mathcal{D}_{-}^{\infty }(\beta )|\beta ,\mathbf{h%
}\rangle ,  \label{R-SOV D-}
\end{align}%
where:%
\begin{align}
\mathcal{D}_{-}^{\infty }(\beta )& =\frac{\beta }{q(\beta -1/\beta )}\left[
\sum_{a=1}^{2\mathsf{N}}\frac{\mathsf{D}_{-}(\zeta
_{a}^{(h_{a})})T_{a}^{-\varphi _{a}}|\beta ,\mathbf{h}\rangle }{\left(
(\zeta _{a}^{(h_{a})})^{2}/q-q/(\zeta _{a}^{(h_{a})})^{2}\right) \left(
(\zeta _{a}^{(h_{a})})^{2}-1/(\zeta _{a}^{(h_{a})})^{2}\right) }\prod 
_{\substack{ b=1  \\ b\neq a\text{ mod}\mathsf{N}}}^{\mathsf{N}}\frac{1}{%
X_{a}^{(h_{a})}-X_{b}^{(h_{b})}}\right.  \notag \\
& \left. +(-1)^{\mathsf{N}}\left( \frac{1}{2}\text{det}_{q}M(1)\prod_{b=1}^{%
\mathsf{N}}\frac{1}{X-X_{b}^{(h_{b})}}-\frac{i}{2}\frac{\zeta _{-}+1/\zeta
_{-}}{\zeta _{-}-1/\zeta _{-}}\text{det}_{q}M(i)\prod_{b=1}^{\mathsf{N}}%
\frac{1}{X+X_{b}^{(h_{b})}}\right) |\beta ,\mathbf{h}\rangle \right]  \notag
\\
& +\frac{\kappa _{-}[q\alpha e^{\tau _{-}}\prod_{a=1}^{\mathsf{N}}\alpha
_{a}\beta _{a}-\frac{\prod_{a=1}^{\mathsf{N}}\gamma _{a}\delta _{a}}{q\alpha
e^{\tau _{-}}}]}{q(\zeta _{-}-1/\zeta _{-})}|\beta ,\mathbf{h}\rangle .
\end{align}%
and%
\begin{equation}
T_{a}^{\pm }|\beta ,h_{1},...,h_{a},...,h_{\mathsf{N}}\rangle =|\beta
,h_{1},...,h_{a}\pm 1,...,h_{\mathsf{N}}\rangle .
\end{equation}
\end{theorem}

\begin{proof}[Proof]
The following interpolation formula is derived as in the previous theorem
using the polynomiality of the operator family $\mathcal{D}_{-}(\lambda
|\beta )$:%
\begin{equation}
\mathcal{D}_{-}(\lambda |\beta )=\sum_{a=0}^{2\mathsf{N}+2}\lambda ^{\left(
2a-(2\mathsf{N}+2)\right) }\mathcal{D}_{a}(\beta ),
\end{equation}%
we have just to compute the asymptotic operator $\mathcal{D}_{-}^{\infty
}(\beta )$. The following identities:%
\begin{eqnarray}
\mathcal{D}_{-}^{\infty }(\beta ) &=&\left[ \beta \mathcal{A}_{-}^{\infty
}/q^{1/2}+\alpha q\mathcal{B}_{-}^{\infty }-\mathcal{C}_{-}^{\infty }/\alpha
q\right] /(\beta -1/\beta ) \\
\mathcal{D}_{-}^{0}(\beta ) &=&-\left[ q^{3/2}\mathcal{A}_{-}^{\infty
}/\beta +q^{2}(\alpha q\mathcal{B}_{-}^{\infty }-\mathcal{C}_{-}^{\infty
}/\alpha q)\right] /(\beta -1/\beta )
\end{eqnarray}%
trivially follows by the definition of the operator family $\mathcal{D}%
_{-}(\lambda |\beta )$, from which we get:%
\begin{equation}
q\mathcal{D}_{-}^{\infty }(\beta )/\beta +\beta \mathcal{D}_{-}^{0}(\beta
)/q=\frac{\kappa _{-}(\beta -1/\beta )}{\zeta _{-}-1/\zeta _{-}}\left[ \frac{%
\prod_{a=1}^{\mathsf{N}}\gamma _{a}\delta _{a}}{q\alpha e^{\tau _{-}}}%
-q\alpha e^{\tau _{-}}\prod_{a=1}^{\mathsf{N}}\alpha _{a}\beta _{a}\right] ,
\label{main-identity-D}
\end{equation}%
while from the interpolation formula we get:%
\begin{align}
\mathcal{D}_{-}^{0}(\beta )& =\sum_{a=1}^{2\mathsf{N}}\frac{q\mathsf{D}%
_{-}(\zeta _{a}^{(h_{a})})T_{a}^{-\varphi _{a}}|\beta ,\mathbf{h}\rangle }{%
\left( (\zeta _{a}^{(h_{a})})^{2}/q-q/(\zeta _{a}^{(h_{a})})^{2}\right)
\left( (\zeta _{a}^{(h_{a})})^{2}-1/(\zeta _{a}^{(h_{a})})^{2}\right) }\prod 
_{\substack{ b=1  \\ b\neq a\text{ mod}\mathsf{N}}}^{\mathsf{N}}\frac{1}{%
X_{a}^{(h_{a})}-X_{b}^{(h_{b})}}  \notag \\
& +(-1)^{\mathsf{N}}\left( \frac{q}{2}\text{det}_{q}M(1)\prod_{b=1}^{\mathsf{%
N}}\frac{1}{X-X_{b}^{(h_{b})}}-\frac{iq}{2}\frac{\zeta _{-}+1/\zeta _{-}}{%
\zeta _{-}-1/\zeta _{-}}\text{det}_{q}M(i)\prod_{b=1}^{\mathsf{N}}\frac{1}{%
X+X_{b}^{(h_{b})}}\right) |\beta ,\mathbf{h}\rangle  \notag \\
& -q^{2}\mathcal{D}_{-}^{\infty }(\beta )|\beta ,\mathbf{h}\rangle .
\end{align}%
from which the statement of the theorem follows easily .
\end{proof}

\subsection{SoV spectral decomposition of the identity}

The Theorem \ref{B-Pseudo-diago} states the pseudo-diagonalizability of $%
\mathcal{B}_{-}(\lambda |\beta )$ for almost all the values of the
boundary-bulk-gauge parameters, so that for almost all the values of these
parameters the left and right states $\langle \beta ,\mathbf{h|}$ and $%
|\beta ,\mathbf{k}\rangle $ are well defined nonzero left and right states
describing a left and right basis in the space of the representation.

We can now defines the following $p^{\mathsf{N}}\times p^{\mathsf{N}}$
matrices $U^{(L,\beta )}$ and $U^{(R,\beta q^{2})}$ defining the change of
basis from the original left and right basis:%
\begin{equation}
\underline{\langle \text{\textbf{h}}|}\equiv \otimes _{n=1}^{\mathsf{N}%
}\langle h_{n},n|\text{ \ \ \ \ and \ \ \ }\underline{|\text{\textbf{h}}%
\rangle }\equiv \otimes _{n=1}^{\mathsf{N}}|h_{n},n\rangle ,
\end{equation}%
composed by $v_{n}$-eigenstates, to the left and right pseudo-eigenbasis of $%
\mathcal{B}_{-}(\lambda |\beta )$:%
\begin{equation}
\langle \beta ,\mathbf{h|}=\underline{\langle \text{\textbf{h}}|}U^{(L,\beta
)}=\sum_{i=1}^{p^{\mathsf{N}}}U_{\varkappa \left( \text{\textbf{h}}\right)
,i}^{(L,\beta )}\underline{\langle \varkappa ^{-1}\left( i\right) |}\text{ \
\ and\ \ \ }|\beta q^{2},\text{\textbf{h}}\rangle =U^{(R,\beta q^{2})}%
\underline{|\text{\textbf{h}}\rangle }=\sum_{i=1}^{p^{\mathsf{N}%
}}U_{i,\varkappa \left( \text{\textbf{h}}\right) }^{(R,\beta q^{2})}%
\underline{|\varkappa ^{-1}\left( i\right) \rangle },
\end{equation}%
where $\varkappa $ is an isomorphism between the sets $\{0,...,p-1\}^{\mathsf{N}%
}$ and $\{1,...,p^{\mathsf{N}}\}$ defined by: 
\begin{equation}
\varkappa :\text{\textbf{h}}\in \{0,...,p-1\}^{\mathsf{N}}\rightarrow
\varkappa \left( \text{\textbf{h}}\right) \equiv 1+\sum_{a=1}^{\mathsf{N}%
}p^{(a-1)}h_{a}\in \{1,...,p^{\mathsf{N}}\}.
\end{equation}%
It follows from the pseudo-diagonalizability of $\mathcal{B}_{-}(\lambda
|\beta )$ that the $p^{\mathsf{N}}\times p^{\mathsf{N}}$ square matrices $%
U^{(L,\beta )}$ and $U^{(R,\beta q^{2})}$\ are invertible matrices for which
it holds:%
\begin{equation}
U^{(L,\beta )}\mathcal{B}_{-}(\lambda |\beta )=\Delta _{\mathcal{B}%
_{-}}(\lambda |\beta )U^{(L,\beta )},\text{ \ \ }\mathcal{B}_{-}(\lambda
|\beta )U^{(R,\beta q^{2})}=U^{(R,\beta q^{2})}\Delta _{\mathcal{B}%
_{-}}(\lambda |\beta ),
\end{equation}%
where $\Delta _{\mathcal{B}_{-}}(\lambda |\beta )$ is the $p^{\mathsf{N}%
}\times p^{\mathsf{N}}$ diagonal matrix defined by: 
\begin{equation}
\left( \Delta _{\mathcal{B}_{-}}(\lambda |\beta )\right) _{i,j}\equiv \delta
_{i,j}\text{\textsc{b}}_{\varkappa ^{-1}\left( i\right) }(\lambda |\beta )%
\text{ \ }\forall i,j\in \{1,...,p^{\mathsf{N}}\}.
\end{equation}%
We can prove that it holds, with the same notation as in Theorem \ref%
{B-Pseudo-diago}:

\begin{proposition}
\label{Spectral-SoV}For almost all the values of the boundary-bulk-gauge
parameters it holds:%
\begin{equation}
\langle \Omega _{\beta }|\Omega _{\beta q^{2}}\rangle \neq 0,
\label{Nonzero-norm}
\end{equation}%
so that fixed the normalization factor: 
\begin{equation}
\text{\textsc{n}}_{\beta }=\left( \prod_{1\leq b<a\leq \mathsf{N}}\left(
X_{a}^{\left( p-1\right) }-X_{b}^{\left( p-1\right) }\right) \langle \Omega
_{\beta }|\Omega _{\beta q^{2}}\rangle \right) ^{1/2},  \label{Normaliz-defa}
\end{equation}%
in the left and right pseudo-eigenstates, then the $p^{\mathsf{N}}\times p^{%
\mathsf{N}}$ matrix $M\equiv U^{(L,\beta )}U^{(R,\beta q^{2})}$ is the
following invertible diagonal matrix:%
\begin{equation}
M_{\varkappa \left( \text{\textbf{h}}\right) \varkappa \left( \text{\textbf{k%
}}\right) }=\langle \beta ,\text{\textbf{h}}|\beta q^{2},\text{\textbf{k}}%
\rangle =\prod_{1\leq a\leq \mathsf{N}}\delta _{h_{a},k_{a}}\prod_{1\leq
b<a\leq \mathsf{N}}\frac{1}{X_{a}^{(h_{a})}-X_{b}^{(h_{b})}},  \label{T2M_jj}
\end{equation}%
from which the following spectral decomposition of the identity $\mathbb{I}$
\ follows:%
\begin{equation}
\mathbb{I}\equiv \sum_{h_{1},...,h_{\mathsf{N}}=0}^{p-1}\prod_{1\leq b<a\leq 
\mathsf{N}}(X_{a}^{(h_{a})}-X_{a}^{(h_{a})})|\beta q^{2},h_{1},...,h_{%
\mathsf{N}}\rangle \langle \beta ,h_{1},...,h_{\mathsf{N}}|.
\end{equation}
\end{proposition}

\begin{proof}
The following identity holds:%
\begin{equation}
\text{\textsc{b}}_{\text{\textbf{h}}}(\lambda |\beta )\langle \beta /q^{2},%
\mathbf{h}|\beta ,\mathbf{k}\rangle =\langle \beta ,\mathbf{h}|\mathcal{B}%
_{-}(\lambda |\beta )|\beta ,\mathbf{k}\rangle =\text{\textsc{b}}_{\text{%
\textbf{k}}}(\lambda |\beta )\langle \beta ,\mathbf{h}|\beta q^{2},\mathbf{k}%
\rangle .
\end{equation}%
From it follows the fact that the matrix $M$ is diagonal:%
\begin{equation}
\langle \beta /q^{2},\mathbf{h}|\beta ,\mathbf{k}\rangle =\langle \beta ,%
\mathbf{h}|\beta q^{2},\mathbf{k}\rangle =0\text{ \ }\forall \mathbf{h}\neq 
\mathbf{k\in }\{0,...,p-1\}^{\mathsf{N}}  \label{orthogonality-pseudo-states}
\end{equation}%
as there exists at least a $n\in \{1,...,\mathsf{N}\}$\ such that $h_{n}\neq
k_{n}$ and then:%
\begin{equation}
\text{\textsc{b}}_{\text{\textbf{h}}}(\zeta _{n}^{(k_{n})}|\beta )\neq
0,\quad \text{\textsc{b}}_{\text{\textbf{k}}}(\zeta _{n}^{(k_{n})}|\beta )=0.
\end{equation}%
Moreover, independently from the choice of the nonzero normalization factor 
\textsc{n}$_{\beta }$, the Theorem \ref{B-Pseudo-diago} implies that the
matrices $U^{(L,\beta )}$ and $U^{(R,\beta q^{2})}$ are invertible for
almost all the values of the boundary-bulk-gauge parameters so that the same
must be true for the diagonal matrix $M$, i.e. it must holds:%
\begin{equation}
M_{\varkappa \left( \mathbf{h}\right) \varkappa \left( \mathbf{h}\right)
}=\langle \mathbf{h,}\beta |\beta q^{2},\mathbf{h}\rangle \neq 0\text{ \ }%
\forall \mathbf{h\in }\{0,...,p-1\}^{\mathsf{N}},
\end{equation}%
which implies $\left( \ref{Nonzero-norm}\right) $ for $\mathbf{p-1}\equiv
(p-1,...,p-1)$ being:%
\begin{equation}
M_{\varkappa \left( \mathbf{p-1}\right) \varkappa \left( \mathbf{p-1}\right)
}=\frac{\langle \Omega _{\beta }|\Omega _{\beta q^{2}}\rangle }{\text{%
\textsc{n}}_{\beta }^{2}},
\end{equation}%
and we can define the normalization factor according to $\left( \ref%
{Normaliz-defa}\right) $. The computation now of the remaining diagonal
matrix elements $M_{\varkappa \left( \mathbf{h}\right) \varkappa \left( 
\mathbf{h}\right) }$ for $\mathbf{h}\neq \mathbf{p-1}$ can be done in a
standard way by computing the matrix elements:%
\begin{equation}
\theta _{a,h_{a}}(\beta )\equiv \langle \beta ,h_{1},...,h_{a},...,h_{%
\mathsf{N}}|\mathcal{A}_{-}(\zeta _{a}^{\left( h_{a}+1\right) }|\beta
q^{2})|\beta q^{2},h_{1},...,h_{a}+1,...,h_{\mathsf{N}}\rangle ,
\end{equation}%
where $h_{a}\in \{0,...,p-1\},$ $a\in \{1,...,\mathsf{N}\}$. Using the left
action of the operator $\mathcal{A}_{-}(\zeta _{a}^{\left( h_{a}+1\right)
}|\beta q^{2})$ we get:%
\begin{align}
\theta _{a,h_{a}}(\beta )& =\left[ \frac{1}{q\left( \zeta
_{a}^{(h_{a})}\right) ^{2}}-\frac{\left( \zeta _{a}^{(h_{a})2}q-1/q\zeta
_{a}^{(h_{a})2}\right) }{\beta (\beta -1/\beta )}\right] \prod_{\substack{ %
b=1  \\ b\neq a}}^{\mathsf{N}}\frac{X_{a}^{(h_{a}+1)}-X_{b}^{(h_{b})}}{%
X_{a}^{(h_{a})}-X_{b}^{(h_{b})}}  \notag \\
& \times \frac{\mathsf{A}_{-}(1/\zeta _{a}^{(h_{a})})(q-1/q)}{\left( (\zeta
_{a}^{(h_{a})})^{2}-1/(\zeta _{a}^{(h_{a})})^{2}\right) }\left. \langle
\beta ,h_{1},...,h_{a}+1,...,h_{\mathsf{N}}|\beta ,h_{1},...,h_{a}+1,...,h_{%
\mathsf{N}}\rangle \right. \\
& =\frac{\mathsf{A}_{-}(1/\zeta _{a}^{(h_{a})})(1/q-q)((\zeta
_{a}^{(h_{a})})^{2}q/\beta -\beta /q(\zeta _{a}^{(h_{a})})^{2})}{\left(
(\zeta _{a}^{(h_{a})})^{2}-1/(\zeta _{a}^{(h_{a})})^{2}\right) (\beta
-1/\beta )}\prod_{\substack{ b=1  \\ b\neq a}}^{\mathsf{N}}\frac{%
X_{a}^{(h_{a}+1)}-X_{b}^{(h_{b})}}{X_{a}^{(h_{a})}-X_{b}^{(h_{b})}}  \notag
\\
& \times \left. \langle \beta ,h_{1},...,h_{a}+1,...,h_{\mathsf{N}}|\beta
,h_{1},...,h_{a}+1,...,h_{\mathsf{N}}\rangle \right.
\end{align}%
while using the decomposition (\ref{A-Sym}) and the fact that:%
\begin{equation}
\langle \beta ,h_{1},...,h_{a},...,h_{\mathsf{N}}|\mathcal{D}_{-}(1/\zeta
_{a}^{\left( h_{a}+1\right) }|\beta q^{2})|\beta
q^{2},h_{1},...,h_{a}+1,...,h_{\mathsf{N}}\rangle =0
\end{equation}%
it holds:%
\begin{align}
\theta _{a,h_{a}}(\beta )& =\frac{\mathsf{A}_{-}(1/\zeta
_{a}^{(h_{a})})(1/q-q)((\zeta _{a}^{(h_{a})})^{2}q/\beta -\beta /q(\zeta
_{a}^{(h_{a})})^{2})k_{a}^{\left( h_{a}+1\right) }}{\left( (\zeta
_{a}^{(h_{a}+1)})^{2}-1/(\zeta _{a}^{(h_{a}+1)})^{2}\right) (\beta -1/\beta )%
}  \notag \\
& \times \langle \beta ,h_{1},...,h_{a},...,h_{\mathsf{N}}|\beta
q^{2},h_{1},...,h_{a},...,h_{\mathsf{N}}\rangle .
\end{align}%
These results lead to the identity:%
\begin{equation}
\frac{\langle \beta ,h_{1},...,h_{a}+1,...,h_{\mathsf{N}}|\beta
q^{2},h_{1},...,h_{a}+1,...,h_{\mathsf{N}}\rangle }{\langle \beta
,h_{1},...,h_{a},...,h_{\mathsf{N}}|\beta q^{2},h_{1},...,h_{a},...,h_{%
\mathsf{N}}\rangle }=\prod_{\substack{ b=1  \\ b\neq a}}^{\mathsf{N}}\frac{%
X_{a}^{(h_{a})}-X_{b}^{(h_{b})}}{X_{a}^{(h_{a}+1)}-X_{b}^{(h_{b})}},
\label{T2F1}
\end{equation}%
from which one can prove:%
\begin{equation}
\frac{\langle \beta ,h_{1},...,h_{a},...,h_{\mathsf{N}}|\beta
q^{2},h_{1},...,h_{a},...,h_{\mathsf{N}}\rangle }{\langle \beta
,p-1,...,p-1|\beta q^{2},p-1,...,p-1\rangle }=\prod_{1\leq b<a\leq \mathsf{N}%
}\frac{X_{a}^{\left( p-1\right) }-X_{b}^{\left( p-1\right) }}{X_{a}^{\left(
h_{a}\right) }-X_{b}^{\left( h_{b}\right) }}.  \label{T2F2}
\end{equation}%
This proves the proposition being by our choice of normalization:%
\begin{equation}
\langle \beta ,p-1,...,p-1|\beta q^{2},p-1,...,p-1\rangle =\prod_{1\leq
b<a\leq \mathsf{N}}\frac{1}{X_{a}^{\left( p-1\right) }-X_{b}^{\left(
p-1\right) }}.
\end{equation}
\end{proof}

\subsection{Proof of pseudo diagonalizability and simplicity of $\mathcal{B}%
_{-}(\protect\lambda |\protect\beta )$}

We prove the pseudo-digonalizability and pseudo-simplicity of $\mathcal{B}%
_{-}(\lambda |\beta )$ in two steps. We first consider some special
representation for which such statement is proven by direct computation then
we use this result to prove our statement for general representations.

\subsubsection{Pseudo diagonalizability and simplicity of $\mathcal{B}_{-}(%
\protect\lambda |\protect\beta )$: special representations}

The following theorem holds:

\begin{theorem}
\label{B-special-P-diago}Let us assume that the conditions on the bulk-gauge
parameters:%
\begin{equation}
\beta =\alpha w_{\mathsf{N}}^{\left( \epsilon _{\mathsf{N}},k_{\mathsf{N}%
}\right) },\text{ }z_{n+1}^{\left( \epsilon _{n+1},k_{n+1}\right)
}=1/w_{n}^{\left( \epsilon _{n},k_{n}\right) }\text{ \ \ }\forall n\in
\left\{ 1,...,\mathsf{N}-1\right\} ,  \label{Condi-bulk-quasi-L}
\end{equation}%
are satisfied for fixed $\mathsf{N}$-tuples of $\epsilon _{n}=\pm 1$ and $%
k_{n}\in \left\{ 0,...,p-1\right\} $ and that the conditions $(\ref%
{dis-equalities})$ hold together with the following ones:%
\begin{equation}
\left( z_{1}^{\left( \epsilon _{1},k_{1}\right) }\right) ^{p}\neq (-\zeta
_{-}+\epsilon _{0}\sqrt{\zeta _{-}^{2}+4\kappa _{-}^{2}})^{p}/(2qe^{\tau
_{-}}\kappa _{-})^{p},  \label{Special-B-simple}
\end{equation}%
and%
\begin{equation}
\mu _{n,\epsilon _{n}}^{2p}\neq \pm 1,\text{ }\mu _{n,\epsilon
_{n}}^{2p}\neq \alpha _{-}^{2p\epsilon },\text{\ }\mu _{n,\epsilon
_{n}}^{2p}\neq -\beta _{-}^{2p\epsilon },\text{\ }\mu _{n,+}^{2p}\neq \mu
_{m,-}^{2p\epsilon },\text{\ }\mu _{n,\epsilon _{n}}^{p}\neq \mu
_{m,\epsilon _{n}}^{p},  \label{Special-B-simple2}
\end{equation}%
for any $\epsilon _{0}=\pm 1$ and $n,m\in \{1,...,\mathsf{N}\}$, then the
operator family $\mathcal{B}_{-}(\lambda |\beta )$ has simple
pseudo-spectrum characterized by:%
\begin{align}
\text{\textsc{b}}_{-,n}(\beta )& =\mu _{n,\epsilon _{n}}q^{1/2}\text{ }%
\forall n\in \{1,...,\mathsf{N}\},\text{ i.e. independent w.r.t. }\beta ,%
\text{\ }  \label{Spectrum-Zeros-B} \\
\text{\textsc{b}}_{-}(\beta )& =f(\alpha ,\beta /q,z_{1}^{\left( \epsilon
_{1},k_{1}+1\right) })\left( \frac{\beta ^{4}/q^{2}}{(1-\beta ^{2})(1-\left(
\beta /q\right) ^{2})}\right) ^{\mathsf{N}}(-1)^{\mathsf{N}}\prod_{n=1}^{%
\mathsf{N}}\frac{\gamma _{n}^{2}}{\mu _{n,\epsilon _{n}}^{2}}  \notag \\
& \times \left[ z_{1}^{\left( \epsilon _{1},k_{1}+1\right) }e^{\tau
_{-}}\kappa _{-}+\zeta _{-}-\kappa _{-}/(z_{1}^{\left( \epsilon
_{1},k_{1}+1\right) }e^{\tau _{-}})\right] /(\zeta _{-}-1/\zeta _{-}),
\end{align}%
and the left pseudo-eigenbasis characterized by the formulae $(\ref%
{Def-left-P-B-basis})$ by fixing: 
\begin{equation}
\langle \Omega _{\beta }|=\sum_{h_{1},...,h_{N}=0}^{p-1}\bigotimes_{n=1}^{%
\mathsf{N}}q^{h_{n}(k_{n}+1)}\left[ \prod_{k_{n}=1}^{h_{n}}\frac{%
a_{n}q^{k_{n}-1/2}+b_{n}q^{1/2-k_{n}}}{c_{n}q^{k_{n}-1/2}+d_{n}q^{1/2-k_{n}}}%
\left( \frac{c_{n}^{p}+d_{n}^{p}}{a_{n}^{p}+b_{n}^{p}}\right) ^{1/p}\right]
^{(1+\epsilon _{n})/2}\langle h_{n},n|.  \label{Left-P-B-eigenS}
\end{equation}%
Similarly, let us assume that the conditions $(\ref{dis-equalities})$ and $(%
\ref{Special-B-simple})$-$(\ref{Special-B-simple2})$ are satisfied together
with:%
\begin{equation}
\beta =\alpha w_{\mathsf{N}}^{\left( \epsilon _{\mathsf{N}},k_{\mathsf{N}%
}\right) },\text{ }z_{n+1}^{\left( \epsilon _{n+1},k_{n+1}-2\right)
}=1/w_{n}^{\left( \epsilon _{n},k_{n}\right) }\text{ \ \ }\forall n\in
\left\{ 1,...,\mathsf{N}-1\right\} ,  \label{Gauge-fix-local}
\end{equation}%
for fixed $\mathsf{N}$-tuples of $\epsilon _{n}=\pm 1$ and $k_{n}\in \left\{
0,...,p-1\right\} $, then the operator family $\mathcal{B}_{-}(\lambda
|\beta )$ has simple pseudo-spectrum characterized by fixing :%
\begin{align}
\text{\textsc{b}}_{-,n}(\beta )& =\mu _{n,\epsilon _{n}}q^{-1/2}\text{ }%
\forall n\in \{1,...,\mathsf{N}\},\text{ i.e. independent w.r.t. }\beta ,%
\text{\ }  \label{Spectrum-Zeros-B2} \\
\text{\textsc{b}}_{-}(\beta )& =f(\alpha ,\beta /q,z_{1}^{\left( \epsilon
_{1},k_{1}-1\right) })\left( \frac{\beta ^{4}}{(1-\beta ^{2})(1-\left( \beta
/q\right) ^{2})}\right) ^{\mathsf{N}}(-1)^{\mathsf{N}}\prod_{n=1}^{\mathsf{N}%
}\frac{\gamma _{n}^{2}}{\mu _{n,\epsilon _{n}}^{2}}  \notag \\
& \times \left[ z_{1}^{\left( \epsilon _{1},k_{1}-1\right) }e^{\tau
_{-}}\kappa _{-}+\zeta _{-}-\kappa _{-}/(z_{1}^{\left( \epsilon
_{1},k_{1}-1\right) }e^{\tau _{-}})\right] /(\zeta _{-}-1/\zeta _{-}),
\end{align}%
and right pseudo-eigenbasis characterized by the formulae $(\ref%
{Def-right-P-B-basis})$ by fixing: 
\begin{equation}
|\Omega _{\beta }\rangle =\prod_{n=1}^{\mathsf{N}}\prod_{r_{n}=1}^{p-1}D_{-}(%
\text{\textsc{b}}_{-,n}(\beta )/q^{r_{n}}|\beta )|\bar{\Omega}_{\beta
}\rangle 
\end{equation}%
with%
\begin{equation}
|\bar{\Omega}_{\beta }\rangle =\left( \frac{\beta }{q}\right) ^{\mathsf{N}%
}\sum_{h_{1},...,h_{N}=0}^{p-1}\prod_{n=1}^{\mathsf{N}}q^{-h_{n}k_{n}}\left[
\prod_{r_{n}=1}^{h_{n}}\frac{c_{n}q^{r_{n}-1/2}+d_{n}q^{1/2-r_{n}}}{%
a_{n}q^{r_{n}-1/2}+b_{n}q^{1/2-r_{n}}}\left( \frac{a_{n}^{p}+b_{n}^{p}}{%
c_{n}^{p}+d_{n}^{p}}\right) ^{1/p}\right] ^{(1+\epsilon
_{n})/2}\bigotimes_{n=1}^{N}|h_{n},n\rangle .  \label{Right-P-B-eigenS}
\end{equation}
\end{theorem}

\begin{proof}
The conditions $(\ref{Condi-bulk-quasi-L})$ and the choice of internal gauge
parameter:%
\begin{equation}
\gamma =z_{1}^{\left( \epsilon _{1},k_{1}\right) },
\end{equation}%
imply that the states $(\ref{Left-P-B-eigenS})$ and $(\ref{Right-P-B-eigenS})
$ are annihilated by $A(\lambda |\alpha ,\beta ,\gamma )$ and $%
A(1/\lambda |\alpha ,\beta /q,\gamma q)$ respectively as the following identifications
hold:%
\begin{equation}
\langle \Omega _{\beta }|=\langle \Omega ,\alpha ,\beta ,\gamma |,\text{ \ \ 
}|\bar{\Omega}_{\beta }\rangle =|\Omega ,\alpha ,\beta /q,\gamma q\rangle
\left( \frac{\beta }{q}\right) ^{\mathsf{N}},
\end{equation}%
moreover, it holds:%
\begin{eqnarray}
\langle \Omega _{\beta }|B(\lambda |\alpha ,\beta ) &=&b(\lambda |\alpha
,\beta )\langle \Omega ,\alpha ,\beta /q,\gamma q|, \\
B(\lambda |\alpha ,\beta /q)|\bar{\Omega}_{\beta }\rangle  &=&|\Omega
,\alpha ,\beta ,\gamma \rangle \beta ^{\mathsf{N}}b(\lambda q|\alpha ,\beta
/q),
\end{eqnarray}%
so that it holds:%
\begin{eqnarray}
\langle \Omega _{\beta }|\mathcal{B}_{-}(\lambda |\alpha ,\beta )
&=&f(\alpha ,\beta /q,\gamma q)\bar{K}_{-}(\lambda |\gamma )_{21}\langle
\Omega _{\beta }|B(\lambda |\alpha ,\beta )B(1/\lambda |\alpha ,\beta /q), \\
\mathcal{B}_{-}(\lambda |\alpha ,\beta )|\bar{\Omega}_{\beta }\rangle 
&=&B(\lambda |\alpha ,\beta )B(1/\lambda |\alpha ,\beta /q)|\bar{\Omega}%
_{\beta }\rangle f(\alpha ,\beta /q,\gamma q)\bar{K}_{-}(\lambda |\gamma
)_{21},
\end{eqnarray}%
and consequently:%
\begin{eqnarray}
\langle \Omega _{\beta }|\mathcal{B}_{-}(\lambda |\alpha ,\beta )
&=&f(\alpha ,\beta /q,\gamma q)\bar{K}_{-}(\lambda |\gamma )_{21}b(\lambda
|\alpha ,\beta )b(1/\lambda |\alpha ,\beta /q)\langle \Omega _{\beta
/q^{2}}|, \\
\mathcal{B}_{-}(\lambda |\alpha ,\beta )|\bar{\Omega}_{\beta }\rangle  &=&|%
\bar{\Omega}_{\beta q^{2}}\rangle b(\lambda q|\alpha ,\beta )b(q/\lambda
|\alpha ,\beta /q)f(\alpha ,\beta /q,\gamma q)\bar{K}_{-}(\lambda |\gamma
)_{21},
\end{eqnarray}%
so that for the pseudo-eigenvalue it holds:%
\begin{eqnarray}
\text{\textsc{b}}_{\text{\textbf{0}}}(\lambda |\beta ) &=&f(\alpha ,\beta
/q,\gamma q)\bar{K}_{-}(\lambda |\gamma )_{21}b(\lambda |\alpha ,\beta
)b(1/\lambda |\alpha ,\beta /q), \\
\text{\textsc{b}}_{\text{\textbf{1}}}(\lambda |\beta ) &=&f(\alpha ,\beta
/q,\gamma q)\bar{K}_{-}(\lambda |\gamma )_{21}b(\lambda q|\alpha ,\beta
)b(q/(\lambda )|\alpha ,\beta /q),
\end{eqnarray}%
respectively on the left and the right. This fixes the values of the \textsc{%
b}$_{-}(\beta )$ and \textsc{b}$_{-,a}(\beta )$ to those stated in this
theorem. Note that the condition $(\ref{Special-B-simple})$ implies that:%
\begin{equation}
\text{\textsc{b}}_{-}(\beta /q^{2a})\neq 0\text{ \ }\forall a\in \left\{
0,...,p-1\right\} .
\end{equation}

Let us now prove that the states $(\ref{Def-left-P-B-basis})$ and $(\ref%
{Def-right-P-B-basis})$ are all nonzero states. The reasoning is done
explicitly only for the left case as for the right one we can proceed
similarly. We know by construction that the state $\langle \Omega _{\beta }|$
is nonzero so let us assume by induction that the same is true for the state 
$\langle \beta ,$\textbf{h}$^{(0)}|=\langle \beta ,h_{1}^{(0)},...,h_{%
\mathsf{N}}^{(0)}|$ with $h_{j}^{(0)}\in \{0,...,p-2\}$ and let us show that 
$\langle \beta ,$\textbf{h}$_{j}^{(0)}|=\langle \beta
,h_{1}^{(0)},...,h_{j}^{(0)}+1,...,h_{\mathsf{N}}^{(0)}|$ is nonzero. We
have that:%
\begin{equation}
\langle \beta ,\mathbf{h}_{j}^{(0)}|\mathcal{A}_{-}(\zeta
_{j}^{(h_{j}^{(0)}+1)}|\beta q^{2})=\mathsf{A}_{-}(\zeta
_{j}^{(h_{j}^{(0)}+1)})\langle \beta ,\mathbf{h}^{(0)}|\neq \text{\b{0} \ }\
\forall j\in \{1,...,\mathsf{N}\}
\end{equation}%
so that $\langle \beta ,\mathbf{h}_{j}^{(0)}|$ is nonzero. Using this we can
prove that all the states $\langle \beta ,h_{1}^{(0)}+x_{1},...,h_{\mathsf{N}%
}^{(0)}+x_{\mathsf{N}}|$ with $x_{j}\in \{0,1\}$ for any $j\in \{1,...,%
\mathsf{N}\}$ are nonzero, which just proves the validity of the induction.
Note that the same statements hold if we substitute the given value of $%
\beta $ fixed in $\left( \ref{Gauge-fix-local}\right) $ with any value $\beta /q^{2a}
$ for any $a\in \left\{ 1,...,p-1\right\} $; i.e. we have that $\langle
\Omega _{\beta /q^{2a}}|$ is nonzero and from that we prove similarly the
induction. 

Let us now prove that the sets of left and right states define respectively
a left and a right basis of the linear space of the representation. Let us
consider the linear combination to zero of the left states:%
\begin{equation}
\text{\b{0}}=\sum_{\mathbf{k}\in Z_{p}^{\mathsf{N}}}c_{\text{\textbf{k}}%
}\langle \beta ,\mathbf{k}|,  \label{Linear to zero}
\end{equation}%
and let us act on it with the following product of operator:%
\begin{align}
\mathcal{B}_{-,\text{\textbf{h}}}(\beta )& \equiv \mathcal{B}_{-}(\zeta
_{1}^{(0)}|\beta )\mathcal{B}_{-}(\zeta _{1}^{(1)}|\beta /q^{2})\cdots 
\mathcal{B}_{-}(\zeta _{1}^{(p-1)}|\beta /q^{2(p-2)})  \notag \\
& \times \mathcal{B}_{-}(\zeta _{2}^{(0)}|\beta /q^{2(p-1)})\mathcal{B}%
_{-}(\zeta _{2}^{(1)}|\beta /q^{2p})\cdots \mathcal{B}_{-}(\zeta
_{2}^{(p-1)}|\beta /q^{4(p-1)-2})  \notag \\
& \cdots \times \mathcal{B}_{-}(\zeta _{\mathsf{N}}^{(0)}|\beta /q^{2(%
\mathsf{N}-1)(p-1)})\cdots \mathcal{B}_{-}(\zeta _{2}^{(p-1)}|\beta /q^{2%
\mathsf{N}(p-1)-2}),
\end{align}%
where the generic monomial in it:%
\begin{equation}
\mathcal{B}_{-}(\zeta _{m}^{(0)}|\beta /q^{2(m-1)(p-1)})\mathcal{B}%
_{-}(\zeta _{2}^{(1)}|\beta /q^{2(m-1)(p-1)+2})\cdots \mathcal{B}_{-}(\zeta
_{2}^{(p-1)}|\beta /q^{2m(p-1)-2}),
\end{equation}%
contains only the $p-1$ arguments $\zeta _{m}^{(k_{m})}$ with $k_{m}\in
\{0,...,p-1\}\backslash \{h_{n}\}$ and \textbf{h}$\equiv \{h_{1},...,h_{%
\mathsf{N}}\}$ is a generic element of $Z_{p}^{\mathsf{N}}$. Then, it easily
to understand that it holds:%
\begin{eqnarray}
\text{\b{0}} &=&\left( \sum_{\mathbf{k}\in Z_{p}^{\mathsf{N}}}c_{\text{%
\textbf{k}}}\langle \beta ,\mathbf{k}|\right) \mathcal{B}_{-,\text{\textbf{h}%
}}(\beta )=c_{\text{\textbf{h}}}\langle \beta ,\mathbf{h}|\mathcal{B}_{-,%
\text{\textbf{h}}}(\beta )  \notag \\
&=&c_{\text{\textbf{h}}}\prod_{n\mathbf{=1}}^{\mathsf{N}}\prod_{k_{n}\in
\{0,...,p-1\}\backslash \{h_{n}\}}\text{\textsc{b}}_{\text{\textbf{h}}%
}(\zeta _{n}^{(k_{n})}|\beta /q^{2(n-1)(p-1)+k_{n}^{\prime }})\langle \beta
/q^{2\mathsf{N}(p-1)},\mathbf{h}|,
\end{eqnarray}%
where $k_{n}^{\prime }=k_{n}$ if $k_{n}^{\prime }<h_{n}$ and $k_{n}^{\prime
}=k_{n}-1$ if $h_{n}<k_{n}$. Now the simplicity of the pseudo-spectrum of $%
\mathcal{B}_{-}(\lambda |\beta )$ implies that:%
\begin{equation}
\prod_{n\mathbf{=1}}^{\mathsf{N}}\prod_{k_{n}\in \{0,...,p-1\}\backslash
\{h_{n}\}}\text{\textsc{b}}_{\text{\textbf{h}}}(\zeta _{n}^{(k_{n})}|\beta
/q^{2(n-1)(p-1)+k_{n}^{\prime }})\neq 0,
\end{equation}%
from which we derive:%
\begin{equation}
c_{\text{\textbf{h}}}=0  \label{null-coeff}
\end{equation}%
having already proven:%
\begin{equation}
\langle \beta /q^{2\mathsf{N}(p-1)},\mathbf{h}|\neq \text{\b{0}}.
\end{equation}%
The generality of the chosen \textbf{h}$\in Z_{p}^{\mathsf{N}}$ implies that
the linear combination to zero $\left( \ref{Linear to zero}\right) $ is
satisfied if and only if $\left( \ref{null-coeff}\right) $ holds for any 
\textbf{h}$\in Z_{p}^{\mathsf{N}}$ , that is the left pseudo-eigenstates $%
\langle \beta ,\mathbf{k}|$ are a left basis.
\end{proof}

Note that in the bulk of the paper we have chosen to present the
construction of the SoV-basis starting from a state $|\Omega _{\beta
}\rangle $ associated to the pseudo-eigenvalue \textsc{b}$_{\text{\textbf{0}}%
}(\lambda |\beta )$ just to simplify the simultaneous presentation of the
left and right basis; in fact, we can construct the right basis also
starting from the state $|\bar{\Omega}_{\beta }\rangle $ associated to 
\textsc{b}$_{\text{\textbf{1}}}(\lambda |\beta )$, which is the state
constructed directly here for the considered special representations.

\subsubsection{Pseudo diagonalizability and simplicity of $\mathcal{B}_{-}(%
\protect\lambda |\protect\beta )$: general representations}

In this section we prove the Theorem \ref{B-Pseudo-diago} stating the pseudo
diagonalizability and simplicity of the operator family $\mathcal{B}%
_{-}(\lambda |\beta )$ for almost all the values of the boundary-bulk-gauge
parameters. Let us first prove the following lemma:

\begin{lemma}
There exists at least one left and one right pseudo-eigenstate $|\Omega _{\beta
}\rangle $ and\ $\langle \Omega _{\beta }|$ of the one parameter family of
pseudo-commuting operators $\mathcal{B}_{-}(\lambda |\beta )$ satisfying the
condition $\left( \ref{L/R-Pseudo-eigenstate}\right) $ with
pseudo-eigenvalue \textsc{b}$_{\text{\textbf{0}}}(\lambda |\beta )$
satisfying the conditions $\left( \ref{SOV-cond1}\right) $ and $\left( \ref%
{SOV-cond2}\right) $.
\end{lemma}

\begin{proof}
The operator family $\mathcal{B}_{-}(\lambda |\beta )$ admits the following
representation:%
\begin{equation}
\mathcal{B}_{-}(\lambda |\beta )=(\frac{\lambda ^{2}}{q}-\frac{q}{\lambda
^{2}})\sum_{a=0}^{\mathsf{N}}\Lambda ^{a}\mathcal{\bar{B}}_{-,a}(\beta
)T_{\beta }^{-2},
\end{equation}%
where the following commutation relations holds:%
\begin{equation}
\mathcal{\bar{B}}_{-,a}(\beta )T_{\beta }=T_{\beta }\mathcal{\bar{B}}%
_{-,a}(\beta q),\text{ \ }\left[ \mathcal{\bar{B}}_{-,a}(\beta ),\mathcal{%
\bar{B}}_{-,b}(\beta )\right] =0\text{ \ }\forall a,b\in \{1,...,\mathsf{N}%
\label{betashifteuh}
\},
\end{equation}%
as a consequence of the commutation relations $\left( \ref{B-comm}\right) $.
The result of the previous section implies that for some special choice of
the boundary-bulk-gauge parameters all the operators $\mathcal{\bar{B}}%
_{-,a,\beta }$ are invertible as $\mathcal{B}_{-}(\lambda |\beta )$ is
pseudo-diagonalizable and it admits the following representation:%
\begin{equation}
\mathcal{B}_{-}(\lambda |\beta )=\text{\textsc{b}}_{-}(\beta )(\frac{\lambda
^{2}}{q}-\frac{q}{\lambda ^{2}})\prod_{a=1}^{\mathsf{N}}(\frac{\lambda }{%
\mathcal{B}_{-,a}(\beta )}-\frac{\mathcal{B}_{-,a}(\beta )}{\lambda }%
)(\lambda B_{-,a}(\beta )-\frac{1}{\lambda \mathcal{B}_{-,a}(\beta )}%
)T_{\beta }^{-2},
\end{equation}%
where the $\mathcal{B}_{-,a}(\beta )$ are commuting and invertible
operators. Then the fact that this operators depend continuously on these
parameters implies that this statement is true for almost any values of
these parameters. This also implies that for almost all the value of the
boundary-bulk-gauge parameters we can use the above representation for $%
\mathcal{B}_{-}(\lambda |\beta )$.

We can now recall that, thanks to the result of the Lemma A.1 of our
previous paper, we can always find a nonzero simultaneous eigenstate of
commuting operators such as the $\mathcal{B}_{-,a}(\beta )$ for any $a\in
\{1,...,\mathsf{N}\}$. This is a pseudo-eigenstate of the operator family $%
\mathcal{B}_{-}(\lambda |\beta )$.

Now, for the same set of representations considered in the previous section
we know that the pseudo-eigenvalues of $\mathcal{B}_{-}(\lambda |\beta )$
satisfy the conditions $\left( \ref{SOV-cond1}\right) $ and $\left( \ref%
{SOV-cond2}\right) $. Then, we can use once again the continuity argument to
argue that the eigenvalues on the common eigenstate still satisfy $\left( %
\ref{SOV-cond1}\right) $ and $\left( \ref{SOV-cond2}\right) $.
\end{proof}

We can now prove the Theorem \ref{B-Pseudo-diago}, by using the results of
the previous sections.

\begin{proof}[Proof of Theorem \protect\ref{B-Pseudo-diago}]
The proof of the pseudo-diagonalizability of $\mathcal{B}_{-}(\lambda |\beta
)$ is a direct consequence of the previous lemma. Indeed, under the
conditions $\left( \ref{SOV-cond1}\right) $ and $\left( \ref{SOV-cond2}%
\right) $ we can prove that all the left and right states are well defined
and nonzero states which are pseudo-eigenstates of $\mathcal{B}_{-}(\lambda
|\beta )$ associated to different pseudo-eigenvalues as a consequence of the
gauge transformed commutation relations. The proof of the fact that the
states $(\ref{Def-left-P-B-basis})$ and $(\ref{Def-right-P-B-basis})$ are
all nonzero is done reproducing the argument presented in the proof of
Theorem $\ref{B-special-P-diago}$.

The statements about the spectral decomposition of the identity of the
theorem have been already given in Proposition \ref{Spectral-SoV}.
\end{proof}

\section{Properties of cofactor}

In this appendix we prove a lemma giving the main properties of the \textit{%
cofactors} of the matrix $D_{\tau }(\lambda ).$

\begin{lemma}
\label{Cofactor-prop}The matrix $D_{\tau }(\lambda )$ has at least rank $p-1$
for any $\lambda \in \mathbb{C}$, up to at most a finite number of values.
The following symmetries:%
\begin{eqnarray}
\text{\textsc{C}}_{i+h,j+h}(\lambda ) &=&\text{\textsc{C}}_{i,j}(\lambda
q^{h})\text{ \ \ }\forall i,j,h\in \{1,...,p\},  \label{Sym-1} \\
\text{\textsc{C}}_{i,j}(-\lambda ) &=&\text{\textsc{C}}_{i,j}(\lambda )\text{%
\ \ \ \ }\forall i,j\in \{1,...,p\},  \label{Sym-2} \\
\text{\textsc{C}}_{1,1}(1/\lambda ) &=&\text{\textsc{C}}_{1,1}(\lambda ),%
\text{ \ \textsc{C}}_{1,2}(1/\lambda )=\text{\textsc{C}}_{1,p}(\lambda ),
\label{Sym-3}
\end{eqnarray}%
hold. Moreover, the cofactors \textsc{C}$_{1,1}(\lambda )$, \textsc{C}$%
_{1,2}(\lambda )$ and \textsc{C}$_{1,p}(\lambda )$ are polynomials in $%
\lambda $ of maximal degree $\left( p-1\right) (2\mathsf{N}+4)$ which admit
the following decomposition:%
\begin{equation}
\text{\textsc{C}}_{1,1}(\lambda )=\widehat{\text{\textsc{C}}}_{1,1}(\lambda
)\left( \lambda ^{2}-\frac{1}{\lambda ^{2}}\right) ^{2}\prod_{\substack{ k=1 
\\ k\neq \left( p-1\right) /2}}^{p-2}\left( \lambda ^{2}q^{1+2k}-\frac{1}{%
\lambda ^{2}q^{1+2k}}\right) ,  \label{F11}
\end{equation}%
where $\widehat{\text{\textsc{C}}}_{1,1}(\lambda )$ is a polynomial in $%
\Lambda $ of degree $\left( p-1\right) (\mathsf{N}+1)$, and:%
\begin{eqnarray}
\text{\textsc{C}}_{1,2}(\lambda ) &=&\widehat{\text{\textsc{C}}}%
_{1,2}(\lambda )\left( \lambda ^{2}-\frac{1}{\lambda ^{2}}\right) ^{2}\prod 
_{\substack{ k=1  \\ k\neq \left( p-1\right) /2}}^{p-2}\left( \lambda
^{2}q^{1+2k}-\frac{1}{\lambda ^{2}q^{1+2k}}\right) ,  \label{F12} \\
\text{\textsc{C}}_{1,p}(\lambda ) &=&\widehat{\text{\textsc{C}}}%
_{1,p}(\lambda )\left( \lambda ^{2}-\frac{1}{\lambda ^{2}}\right) ^{2}\prod 
_{\substack{ k=1  \\ k\neq \left( p-1\right) /2}}^{p-2}\left( \lambda
^{2}q^{1+2k}-\frac{1}{\lambda ^{2}q^{1+2k}}\right) ,  \label{F1p}
\end{eqnarray}%
where $\widehat{\text{\textsc{C}}}_{1,2}(\lambda )$ and $\widehat{\text{%
\textsc{C}}}_{1,p}(\lambda )$ are polynomials of maximal degree $2\left(
p-1\right) (\mathsf{N}+1)$ in $\lambda $.
\end{lemma}

\begin{proof}
Let us remark that independently from the explicit form of $\tau (\lambda )$
the following identities hold:%
\begin{equation}
\text{\textsc{C}}_{1,p}(q^{1/2-p}/\mu _{a,+})=\prod_{j=1}^{p-1}\text{\textsc{%
a}}{}(\mu _{a,+}q^{j})\neq 0\text{ \ \ }\forall a\in \{1,...,\mathsf{N}\},
\end{equation}%
so that $\overline{\text{\textsc{C}}}_{1,p}(\lambda )$ is a non-zero
polynomial in $\lambda $ which implies the statement on the rank of $D_{\tau
}(\lambda )$. The proof of the above symmetry properties is standard we just
need to make some exchange of rows and columns to bring the matrix in the
determinant defining the cofactor in the l.h.s into the matrix defining the
cofactor in the r.h.s..

Let us show our statement on the form of \textsc{C}$_{1,1}(\lambda )$. In
order to do so we have to prove that \textsc{C}$_{1,1}(\lambda )$ is finite
in the points\footnote{%
As for $h=0$ the matrix \textsc{M}$_{1,1}(\pm i^{a})$ does not contain any
singular elements.} $\lambda =\pm i^{a}q^{h}$ for any $h\in \{1,...,p-1\}$.
More precisely, in the line $p-h$ there is at least one element of the
matrix \textsc{M}$_{1,1}(\lambda )$ associated to \textsc{C}$_{1,1}(\lambda
) $ which is diverging in the limit $\lambda \rightarrow \pm i^{a}q^{h}$.
Here, we have to distinguish three cases. For the case $h\neq (p\pm 1)/2$, we
can proceed as done in the bulk of the paper. We can define the matrix 
\textsc{M}$_{1,1}^{(h)}(\lambda )$ as the matrix with all the rows
coinciding with those of \textsc{M}$_{1,1}(\lambda )$ except the row $%
(p+1)/2-h$, which is obtained by summing the row $(p-1)/2-h$ and $(p+1)/2-h$
of \textsc{M}$_{1,1}(\lambda )$ and dividing them by $((\lambda
/q^{h})^{2}-(q^{h}/\lambda )^{2})$, and the row $p-h$, obtained multiplying
the row $p-h$ of \textsc{M}$_{1,1}^{(h)}(\lambda )$ by $((\lambda
/q^{h})^{2}-(q^{h}/\lambda )^{2})$. Clearly it holds det$_{p-1}\text{\textsc{%
M}}_{1,1}^{(h)}(\lambda )=\left( -1\right) ^{i+j}$\textsc{C}$_{1,1}(\lambda )
$ and all the rows of the matrix \textsc{M}$_{1,1}^{(h)}(\pm
i^{a}q^{h}\lambda )$ are finite in the limits $\lambda \rightarrow 1$ and
so the same is true for their determinants. In fact, it is possible to show
that these lines are linear dependents in each one of the matrices \textsc{M}%
$_{1,1}^{(h)}(\pm i^{a}q^{h}),$ so that:%
\begin{equation}
\text{det}_{p-1}\text{\textsc{M}}_{1,1}^{(h)}(\pm i^{a}q^{h})=0\text{ \ \ }%
\forall h\in \{1,...,p\}\backslash \text{ }\{0,(p\pm 1)/2\}.
\end{equation}

In the remaining cases, if $h=(p\pm 1)/2$ then the row $(p\pm 1)/2p-h=p$ mod$ 
(p)$ is not contained in \textsc{M}$_{1,1}(\pm i^{a}q^{h})$ so that we cannot
remove here the divergence as we have done before. However, we can proceed
differently, let us explain it in the case $h=(p+1)/2$ as in the other case
we can proceed similarly. In the last row of \textsc{M}$_{1,1}(\pm
i^{a}q^{(p+1)/2}\lambda )$ under the limit $\lambda \rightarrow 1$ the last
element tend to $\tau (i^{a}q^{1/2})$, finite nonzero value, and the next to
last tend to \textsc{a}$(\pm i^{a}q^{-1/2})=0$, all the others on this row
are zero. So that \textsc{C}$_{1,1}(\pm i^{a}q^{(p-1)/2})$ is finite iff det$%
_{p-2}D_{(1,p),(1,p)}(\pm i^{a}q^{(p+1)/2})$ is finite. This is shown using
the following expansion of the determinant:%
\begin{align}
\text{det}_{p-2}D_{(1,p),(1,p)}(\pm i^{a}q^{(p+1)/2}\lambda )& =\tau
(\lambda )\text{det}_{p-3}D_{\tau ,\left( 1,(p+1)/2,p\right) ,\left(
1,(p+1)/2,p\right) }(\pm i^{a}q^{(p+1)/2}\lambda ) \\
& +\frac{\text{\textsc{x}}(\lambda )\text{det}_{p-3}D_{\tau ,\left(
1,(p+1)/2,p\right) ,\left( 1,(p+1)/2-1,p\right) }(\pm
i^{a}q^{(p+1)/2}\lambda )}{\lambda ^{2}-1/\lambda ^{2}} \\
& -\frac{\text{\textsc{x}}(1/\lambda )\text{det}_{p-3}D_{\tau ,\left(
1,(p+1)/2,p\right) ,\left( 1,(p+1)/2+1,p\right) }(\pm
i^{a}q^{(p+1)/2}\lambda )}{\lambda ^{2}-1/\lambda ^{2}},
\end{align}%
and the identity:%
\begin{equation}
\text{det}_{p-3}D_{\tau ,\left( 1,(p+1)/2,p\right) ,\left(
1,(p+1)/2-1,p\right) }(\pm i^{a}q^{(p+1)/2})=\text{det}_{p-3}D_{\tau ,\left(
1,(p+1)/2,p\right) ,\left( 1,(p+1)/2-1,p\right) }(\pm i^{a}q^{(p+1)/2}).
\end{equation}%
Finally, let us remark that in the case $h=0$ the lines $(p-1)/2$ and $%
(p+1)/2$ of \textsc{M}$_{1,1}(\pm i^{a})$ are one the opposite of the other
so that det$_{p-1}$\textsc{M}$_{1,1}(\pm i^{a})=0$. We can so define the
matrix \textsc{M}$_{1,1}^{(0)}(\lambda )$ as the matrix with all the rows
coinciding with those of \textsc{M}$_{1,1}(\lambda )$ except the row $%
(p+1)/2 $, which is obtained by summing the row $(p-1)/2$ and $(p+1)/2$ of 
\textsc{M}$_{1,1}(\lambda )$ and dividing them by $(\lambda ^{2}-1/\lambda
^{2})$, this matrix has finite elements on the row $(p+1)/2$ also in the
limit $\lambda \rightarrow \pm i^{a}$. Similarly to the previous cases one
can show that the rows of \textsc{M}$_{1,1}^{(0)}(\pm i^{a})$ are linear
dependent so that it holds also:%
\begin{equation}
\lim_{\lambda \rightarrow \pm i^{a}}\frac{\text{det}_{p-1}\text{\textsc{M}}%
_{1,1}(\pm i^{a}\lambda )}{\lambda ^{2}-1/\lambda ^{2}}=(-1)^{a}\text{det}%
_{p-1}\text{\textsc{M}}_{1,1}^{(0)}(\pm i^{a})=0,
\end{equation}%
from which our statement on the form of \textsc{C}$_{1,1}(\lambda )$
follows. Similarly, we can prove our statement on \textsc{C}$_{1,p}(\lambda )
$.
\end{proof}

\end{document}